\documentclass[a4paper,UKenglish,cleveref, autoref, thm-restate]{lipics-v2021}

\pdfoutput=1 %

\usepackage{tikzit} 
\usepackage{tikz-cd} 

\tikzstyle{morphism}=[fill=white, draw=black, shape=rectangle]
\tikzstyle{medium box}=[fill=white, draw=black, shape=rectangle, minimum width=0.8cm]
\tikzstyle{medium large morphism}=[fill=white, draw=black, shape=rectangle, minimum width=1.2cm]
\tikzstyle{large morphism}=[fill=white, draw=black, shape=rectangle, minimum width=1.7cm]
\tikzstyle{bn}=[fill=black, draw=black, shape=circle, inner sep=1.5pt]
\tikzstyle{state}=[fill=white, draw=black, regular polygon, regular polygon sides=3, minimum width=0.8cm, shape border rotate=180, inner sep=0pt]
\tikzstyle{long state}=[fill=white, draw=black, shape=isosceles triangle, isosceles triangle apex angle=90, shape border rotate=270]
\tikzstyle{medium state}=[fill=white, draw=black, regular polygon, regular polygon sides=3, minimum width=1.3cm, inner sep=0pt, shape border rotate=180]
\tikzstyle{large state}=[fill=white, draw=black, regular polygon, regular polygon sides=3, minimum width=2.2cm, shape border rotate=180, inner sep=0pt]
\tikzstyle{wn}=[fill=white, draw=black, shape=circle, inner sep=1.5pt]
\tikzstyle{likelihood}=[fill=white, draw=black, regular polygon, regular polygon sides=3, minimum width=0.8cm, shape border rotate=0, inner sep=0pt]

\tikzstyle{arrow}=[->]
\tikzstyle{dashed box}=[-, dashed]
\tikzstyle{new edge style 0}=[-, fill={rgb,255: red,148; green,162; blue,255}, draw=none]
\input{comonoids.tikzdefs}

{
\theoremstyle{definition}
\newtheorem{algorithm}[theorem]{Algorithm}
}

\title{A categorical account of the Metropolis--Hastings algorithm} %

\author{Rob {Cornish}}{College of Computing and Data Science, Nanyang Technological University, Singapore \and \url{https://jrmcornish.github.io/} }{rob.cornish@ntu.edu.sg}{https://orcid.org/0009-0003-5947-3400}{Supported by the National Research Foundation, Singapore, under a National Research Foundation Fellowship in Artificial Intelligence (Award No.: NRFFIAI1-2024-0014).}%

\author{Andi Q. Wang}{Department of Statistics, University of Warwick, United Kingdom\and\url{https://warwick.ac.uk/fac/sci/statistics/staff/academic-research/wang/}} {andi.wang@warwick.ac.uk}{https://orcid.org/0000-0001-9551-6592}{Supported by a Discipline Hopping Award from the Prob\_AI Hub (EP/Y028783/1).}

\authorrunning{R. Cornish and A.\,Q. Wang} %

\Copyright{Rob Cornish and Andi Q. Wang} %

\ccsdesc[500]{Mathematics of computing~Bayesian computation}
\ccsdesc[500]{Theory of computation~Categorical semantics}

\keywords{Categorical probability, Metropolis--Hastings, Markov category, CD category} %

\category{} %

\relatedversion{} %

\acknowledgements{We are very grateful to Sam Power for many useful discussions throughout the course of the project.
	A big thanks to Sam Staton, Paolo Perrone, Tobias Fritz, Pedro Amorim, and Luke Ong for their support and helpful suggestions. Many thanks also to the three anonymous referees for their constructive comments which improved the manuscript.}%

\nolinenumbers %

\EventEditors{Claudia Faggian and Joost-Pieter Katoen}
\EventNoEds{2}
\EventLongTitle{41st Annual Symposium on Logic in Computer Science (LICS 2026)}
\EventShortTitle{LICS 2026}
\EventAcronym{LICS}
\EventYear{2026}
\EventDate{July 20--23, 2026}
\EventLocation{Lisbon, Portugal}
\EventLogo{}
\SeriesVolume{380}
\ArticleNo{32}

\begin{document}

\maketitle

\begin{abstract}
Metropolis--Hastings (MH) is a foundational Markov chain Monte Carlo (MCMC) algorithm.
In this paper, we ask whether it is possible to formulate and analyse MH in terms of categorical probability, using a recent involutive framework for MH-type procedures as a concrete case study.
We show how basic MCMC concepts such as invariance and reversibility can be formulated in Markov categories, and how one part of the MH kernel can be analysed using standard CD categories.
To go further, we then study enrichments of CD categories over commutative monoids.
This gives an expressive setting for reasoning abstractly about a range of important probabilistic concepts, including substochastic kernels, finite and $\sigma$-finite measures, absolute continuity, singular measures, and Lebesgue decompositions.
Using these tools, we give synthetic necessary and sufficient conditions for a general MH-type sampler to be reversible with respect to a given target distribution.
\end{abstract}

\section{Introduction}\label{sec:intro}

\emph{Markov kernels} are a fundamental concept in probability theory.
Recall that a Markov kernel $\P : \tarS \to \sndS$ encodes a family of probability distributions $\P(x, \dif y)$ on the space $\sndS$ indexed by points $x \in \tarS$.
Recent work on \emph{copy-discard (CD) categories} and \emph{Markov categories} \cite{Cho2019,Fritz2020} has sought to develop a synthetic framework for reasoning about systems of Markov kernels and their unnormalised counterparts algebraically.
This approach has been applied successfully to a number of classical results from probability and statistics, yielding streamlined and more conceptual proofs than traditional measure-theoretic arguments \cite{Fritz2021,Rischel2020,Fritz2020,Moss2023}.

In applied settings, Markov kernels are often regarded as ``stochastic functions'' and used to model \emph{conditional probability distributions}, which in turn have a huge number of applications in computational statistics, machine learning, and computer science generally.
It is therefore natural to ask: to what extent are Markov and CD categories ready for reasoning about contemporary methodology in these domains?
Are they usable ``as-is'', or are additional structures or concepts needed to handle the complexities of modern algorithms?

We consider this question in the context of \emph{Markov chain Monte Carlo} (MCMC) methods, and specifically the \emph{Metropolis--Hastings} (MH) algorithm \cite{metropolis1953equation, Hastings1970}, which is a foundational tool for sampling from complex probability distributions.
We focus on a recent line of work which studies MH through the lens of \emph{involutions} \cite{Neklyudov2020,Andrieu2020,Glatt-Holtz2023,Glatt-Holtz2024,cusumano2020automating}, building on earlier foundational work \cite{Tierney1998,Green1995}.
This perspective has already lead to a range of new methodology of interest to computational statistics practitioners \cite{Bou-Rabee2024, Liu2024, Xu2025}.

A major contribution to the involutive approach is the framework of \cite{Andrieu2020}.
This considers a Markov kernel $\Pimh : \augS \to \augS$ defined
\begin{equation} \label{eq:mh-kernel}
	\Pimh(\xi, \dif \xi') \coloneqq \alpha(\xi) \, \delta_{\phi(\xi)}(\dif \xi') + (1 - \alpha(\xi)) \, \delta_{\xi}(\dif \xi'),
\end{equation}
where $\alpha : \augS \to [0, 1]$ is a measurable function regarded as an \emph{acceptance probability}, $\phi : \augS \to \augS$ is a measurable \textit{involution} (i.e.\ $\phi = \phi^{-1}$), and $\delta$ is the Dirac delta.
Intuitively, given an input $\xi$, this kernel transitions to $\phi(\xi)$ with probability $\alpha(\xi)$, and otherwise remains at $\xi$.
By combining $\Pimh$ with a suitable \emph{state-space augmentation} (see Section \ref{sec:augmented-reversibility}), it is shown in \cite{Andrieu2020} that a very wide variety of MH-type algorithms can be recovered as special cases (see also Appendix~\ref{app:MCMC_involutions} for two examples).

The main theoretical result of \cite{Andrieu2020} is their Theorem 3, which gives sufficient conditions under which $\Pimh$ is \emph{reversible} with respect to a given target distribution.
Reversibility in turn constitutes a key correctness guarantee for many MCMC algorithms.
We state a summarised version of this result as follows (see Appendix~\ref{app:measure_theory} for our measure theory terminology and notation): %

\begin{theorem} \label{thm:andrieu2020}
	Suppose $\mu$ is a finite measure on a measurable space $\augS$, and $\phi : \augS \to \augS$ is a measurable involution.
	Then there exists a measurable $S \subseteq \augS$ such that the restriction $\restr{\mu}{S}$ and its pushforward $\restr{\mu}{S}^\phi$ are equivalent as measures, and $\restr{\mu}{S^c}$ and its pushforward $\restr{\mu}{S^c}^\phi$ are singular.
	The kernel $\Pimh$ \eqref{eq:mh-kernel} above is then $\restr{\mu}{S}$-reversible if
	\begin{equation} \label{eq:balancing-ae}
		\alpha(\xi) = \begin{cases}
			\alpha(\phi(\xi)) \, r(\xi) & \text{if $\xi \in S$,} \\
			0 & \text{otherwise}.
		\end{cases}
	\end{equation}
	where $r \coloneqq \frac{\dif \restr{\mu}{S}^\phi}{\dif\restr{\mu}{S}}$ denotes the Radon--Nikodym derivative.
\end{theorem}

\subsubsection*{Summary of main results}

Given the important implications of Theorem~\ref{thm:andrieu2020} for MH-type procedures, it serves as a natural test case for the applicability of categorical probability.
Our goal in this paper is to obtain this result synthetically using Markov and CD categories.
For this, we find that three levels of structure are needed:
\begin{itemize}
	\item In Section~\ref{sec:Markov-cats}, we use Markov categories to formulate generic aspects of MCMC algorithms, such as invariance, reversibility and state-space augmentation.
	\item In Section~\ref{sec:CD_cats}, we use CD categories to reason about the first summand in \eqref{eq:mh-kernel}: we prove reversibility of this term, and recover part of the condition \eqref{eq:balancing-ae}. %
	\item However, to study the entire kernel $\Pimh$, we need additional structure: for example, we need a way to ``add'' morphisms, and we need a way to reason about decompositions of morphisms into disjoint ``parts''.
	In Section~\ref{sec:CMon}, we therefore consider an \emph{enriched} extension of CD categories.
	This allows us to provide a synthetic generalisation of Theorem~\ref{thm:andrieu2020} in the form of Theorems \ref{thm:enrich_rev} and \ref{thm:lebesgue-decomposition-involution} below.
\end{itemize}
The appendix contains some additional background and proofs, as well as two additional examples of MCMC-related techniques (the systematic-scan Gibbs sampler, and importance sampling) within categorical probability, which may be of interest to practitioners. 

We believe these results are interesting in their own right, as they involve arguments that are both more conceptual and more general than their measure-theoretic counterparts.
We also hope our work illustrates how categorical tools may be used to analyse contemporary methodology in computational statistics and machine learning more generally. %

\subsubsection*{Technical contributions}

In order to consider the full MH kernel $\Pimh$, we study CD categories \emph{enriched} over commutative monoids, which provides a minimal structure for reasoning about addition of morphisms.
This enrichment is expressive, and permits formulating and proving results about abstract versions of the following important concepts from measure theory:
\begin{itemize}
	\item \emph{Substochastic kernels} and \emph{probabilities} (i.e.\ values in $[0, 1]$).
	\item \emph{Finite} and \emph{$\sigma$-finite} kernels.
	\item Pointwise \emph{absolute continuity} between kernels. 
	\item \emph{Singular measures} and \emph{Lebesgue decompositions}.
\end{itemize}
The enrichment combines well with string diagrams, allowing for straightforward graphical reasoning.
It also integrates naturally with existing categorical approaches to probabilistic choice: when suitable coproducts exist, the subcategory of normalised morphisms carries the structure of a \emph{distributive Markov category} \cite{ackerman2024probabilistic}, relating our ``algebraic'' treatment of convex combinations to the ``sampling'' perspective.

Importantly, we show that these tools are not merely descriptive, but may be used to obtain interesting results about probabilistic concepts.
In the context of Theorem \ref{thm:andrieu2020}, this enrichment allows us to reason about decompositions of $\mu$ into parts supported on $S$ and $S^c$, and to obtain the condition \eqref{eq:balancing-ae} via purely algebraic arguments.
It also supplies a converse to Theorem~\ref{thm:andrieu2020}, which was not given in \cite{Andrieu2020}, and permits a straightforward extension to the historically more complex \emph{skew-reversible} setting (Corollary~\ref{cor:skew}).
We believe these results suggest this enrichment may be useful in applications beyond MCMC also.

\section{MCMC foundations in Markov categories} \label{sec:Markov-cats}

\subsection{Recap: Markov categories}

We first recall the fundamentals of \emph{Markov categories}, which will allow us to formulate key Markov chain concepts such as reversibility and invariance.
We will assume familiarity with basic definitions from category theory and monoidal categories, as well as string diagrams \cite{Joyal1991, selinger2010survey}.
For readers unfamiliar with this material, we give some further details and references in Appendix~\ref{app:intro_to_category_theory}.

\begin{definition} \label{def:markov_cat}
	A \emph{Markov category} is a symmetric monoidal category $(\C, \otimes, \ind)$ where each object $\tarS$ in $\C$ is equipped with a commutative comonoid structure $(\cop_\tarS, \del_\tarS)$. In string diagrams, we denote $\cop_\tarS$ and $\del_\tarS$ respectively as
	\begin{cdiag*}
		\begin{tikzpicture}
	\begin{pgfonlayer}{nodelayer}
		\node [style=bn] (3) at (3.5, 0.5) {};
		\node [style=none] (4) at (3.5, -0.5) {};
		\node [style=none] (12) at (-4, 0) {};
		\node [style=none] (13) at (-4, -1) {};
		\node [style=none] (14) at (-5, 1) {};
		\node [style=none] (15) at (-3, 1) {};
		\node [style=bn] (16) at (-4, 0) {};
		\node [style=none] (24) at (3.5, -1) {$\tarS$};
		\node [style=none] (25) at (-4, -1.5) {$\tarS$};
		\node [style=none] (26) at (-5, 1.5) {$\tarS$};
		\node [style=none] (27) at (-3, 1.5) {$\tarS$};
		\node [style=none] (28) at (0, 0) {$\text{and}$};
	\end{pgfonlayer}
	\begin{pgfonlayer}{edgelayer}
		\draw (3) to (4.center);
		\draw (12.center) to (13.center);
		\draw [bend right=45] (12.center) to (15.center);
		\draw [bend right=45] (14.center) to (12.center);
	\end{pgfonlayer}
\end{tikzpicture}
	\end{cdiag*}
	These morphisms satisfy certain compatibility conditions with the monoidal structure (see Definition~\ref{def:Markov_cat} in the Appendix), and the following requirement for all $\P : \tarS \to \sndS$ in $\C$:
	\begin{cdiag} \label{eq:del_natural}
		\begin{tikzpicture}
	\begin{pgfonlayer}{nodelayer}
		\node [style=morphism] (0) at (-2, 0) {$\P$};
		\node [style=bn] (1) at (-2, 1.25) {};
		\node [style=none] (2) at (-2, -1.25) {};
		\node [style=none] (4) at (1.75, -1.25) {};
		\node [style=bn] (5) at (1.75, 1.25) {};
		\node [style=none] (6) at (0, 0) {$=$};
		\node [style=none] (7) at (-2, -1.75) {$\tarS$};
		\node [style=none] (8) at (1.75, -1.75) {$\tarS$};
	\end{pgfonlayer}
	\begin{pgfonlayer}{edgelayer}
		\draw (1) to (0);
		\draw (0) to (2.center);
		\draw (5) to (4.center);
	\end{pgfonlayer}
\end{tikzpicture}
	\end{cdiag}
\end{definition}

\begin{example} \label{ex:stoch-defn}
	For the purposes of obtaining Theorem \ref{thm:andrieu2020}, the most important example of a Markov category is the category $\stoch$.
	The objects of this category are measurable spaces, and the morphisms $\P : \tarS \to \sndS$ are \emph{Markov kernels}, namely functions $\P : \tarS \times \Sigma_{\sndS} \to [0, 1]$ such that each function $x \mapsto \P(x, A)$ is measurable and each function $A \mapsto \P(x, A)$ is a probability measure.
	Composition is given by the Chapman--Kolmogorov equation:
	\[
		(\Q \circ \P)(x, B) \coloneqq \int \P(x, \dif y) \, \Q(y, B)
	\]
	and may be regarded as \emph{sequential composition} of Markov kernels.
	The monoidal unit $\ind$ is the singleton measurable space $\{\ast\}$, and the monoidal product $\tarS \otimes \sndS$ is given by the usual product of measurable spaces.
	The monoidal product of morphisms is defined as:
	\begin{equation} \label{eq:stoch-monoidal-prod}
		(R \otimes Q)((x, y), A \times B) \coloneqq R(x, A) \, Q(y, B),
	\end{equation}
	which may be regarded as \emph{parallel composition}.
	The remaining structure of $\stoch$ is given in Example~\ref{exa:stoch} in the Appendix.
\end{example}

Morphisms $\pi : \ind \to \tarS$ out of the monoidal unit $\ind$ are regarded as having ``no input'' and appear in string diagrams as follows:
\begin{cdiag*}
	\begin{tikzpicture}
	\begin{pgfonlayer}{nodelayer}
		\node [style=state] (0) at (0, 0) {$\pi$};
		\node [style=none] (1) at (0, 0.25) {};
		\node [style=none] (2) at (0, 1) {};
		\node [style=none] (3) at (0, 1.5) {$\tarS$};
	\end{pgfonlayer}
	\begin{pgfonlayer}{edgelayer}
		\draw (2.center) to (1.center);
	\end{pgfonlayer}
\end{tikzpicture}

\end{cdiag*}
A morphism $\phi: \tarS \to \sndS$ is \emph{deterministic} \cite[Definition 10.1]{Fritz2020} if
\begin{cdiag} \label{eq:det_phi}
	\begin{tikzpicture}
	\begin{pgfonlayer}{nodelayer}
		\node [style=none] (0) at (0, 0) {};
		\node [style=none] (1) at (0, 0) {$=$};
		\node [style=none] (4) at (3, -0.5) {};
		\node [style=none] (5) at (3, -1.5) {};
		\node [style=none] (6) at (2, 1.5) {};
		\node [style=none] (7) at (4, 1.5) {};
		\node [style=none] (10) at (-3, -0.5) {};
		\node [style=none] (11) at (-3, -1.5) {};
		\node [style=none] (12) at (-4, 1.5) {};
		\node [style=none] (13) at (-2, 1.5) {};
		\node [style=morphism] (14) at (-3, -0.5) {$\phi$};
		\node [style=none] (15) at (-3, 0.5) {};
		\node [style=none] (16) at (2, 0.5) {};
		\node [style=none] (17) at (4, 0.5) {};
		\node [style=none] (18) at (2, 1.5) {};
		\node [style=morphism] (19) at (2, 0.5) {$\phi$};
		\node [style=morphism] (20) at (4, 0.5) {$\phi$};
		\node [style=bn] (21) at (3, -0.5) {};
		\node [style=none] (22) at (2, 2) {$\sndS$};
		\node [style=none] (23) at (4, 2) {$\sndS$};
		\node [style=none] (24) at (3, -2) {$\tarS$};
		\node [style=bn] (25) at (-3, 0.5) {};
		\node [style=none] (26) at (-4, 2) {$\sndS$};
		\node [style=none] (27) at (-2, 2) {$\sndS$};
		\node [style=none] (28) at (-3, -2) {$\tarS$};
	\end{pgfonlayer}
	\begin{pgfonlayer}{edgelayer}
		\draw (4.center) to (5.center);
		\draw (10.center) to (11.center);
		\draw (14) to (15.center);
		\draw [bend right=45] (12.center) to (15.center);
		\draw [bend right=45] (15.center) to (13.center);
		\draw [bend right=45] (16.center) to (4.center);
		\draw [bend right=45] (4.center) to (17.center);
		\draw (17.center) to (7.center);
		\draw (18.center) to (16.center);
	\end{pgfonlayer}
\end{tikzpicture}
\end{cdiag}

\begin{example}
	In $\stoch$, since $\ind$ is the singleton space, a Markov kernel $\mu : \ind \to \tarS$ encodes a single probability measure on $\tarS$, which we will denote simply as $\mu(\dif x)$.
	Likewise, every measurable function $\phi : \tarS \to \sndS$ gives rise to a deterministic Markov kernel $\Phi : \tarS \to \sndS$ defined as
	\begin{equation}\label{eq:deterministic_morphism}
		\Phi(x, \dif y) \coloneqq \delta_{\phi(x)}(\dif y),
	\end{equation}
	where $\delta$ is the Dirac delta.
	We will often identify $\Phi$ and $\phi$ in what follows.
	This recovers the concept of \emph{pushforwards} from measure theory: recall that the pushforward of a measure $\mu$ on $\tarS$ by a measurable function $\phi : \tarS \to \sndS$ is the distribution $\mu^\phi$ on $\sndS$ with $\mu^\phi(B) \coloneqq \mu(\phi^{-1}(B))$.
	In $\stoch$, this is just the composition $\mu^\phi = \phi \circ \mu$, where here the right-hand $\phi$ is regarded as a deterministic Markov kernel as in \eqref{eq:deterministic_morphism}.
\end{example}

The following concept is fundamental in Bayesian statistics and machine learning, and will be useful for us in formulating \emph{state space augmentation} in Section \ref{sec:augmented-reversibility}.

\begin{definition}[Page 9, \cite{Cho2019}] \label{def:bayesian-inverse}
	Let $\P : \tarS \to \sndS$ and $\pi : \ind \to \tarS$ be morphisms in a Markov category $\C$.
	Then $\P^\dagger : \sndS \to \tarS$ is a \emph{Bayesian inverse of $\P$ with respect to $\pi$} if
	\begin{cdiag}\label{diag:bayes_inv}
		\begin{tikzpicture}
	\begin{pgfonlayer}{nodelayer}
		\node [style=none] (0) at (0, 0) {$=$};
		\node [style=state] (1) at (-3.5, -1.5) {$\pi$};
		\node [style=none] (2) at (-3.5, -1.25) {};
		\node [style=none] (3) at (-3.5, -0.5) {};
		\node [style=none] (4) at (-4.5, 0.5) {};
		\node [style=none] (5) at (-2.5, 0.5) {};
		\node [style=bn] (6) at (-3.5, -0.5) {};
		\node [style=morphism] (7) at (-2.5, 0.5) {$\P$};
		\node [style=none] (9) at (-2.5, 0.75) {};
		\node [style=none] (10) at (-2.5, 1.5) {};
		\node [style=none] (11) at (-4.5, 1.5) {};
		\node [style=none] (12) at (-4.5, 2) {$\tarS$};
		\node [style=none] (13) at (-2.5, 2) {$\augS$};
		\node [style=state] (14) at (3.5, -2) {$\pi$};
		\node [style=none] (15) at (3.5, -1.75) {};
		\node [style=none] (16) at (3.5, 0.25) {};
		\node [style=none] (17) at (2.5, 1.25) {};
		\node [style=none] (18) at (4.5, 1.25) {};
		\node [style=bn] (19) at (3.5, 0.25) {};
		\node [style=morphism] (20) at (3.5, -0.75) {$\P$};
		\node [style=none] (22) at (4.5, 1.25) {};
		\node [style=none] (23) at (2.5, 1.25) {};
		\node [style=none] (24) at (2.5, 2.75) {$\tarS$};
		\node [style=none] (25) at (4.5, 2.75) {$\augS$};
		\node [style=morphism] (26) at (2.5, 1.25) {$\P^\dagger$};
		\node [style=none] (27) at (2.5, 1.5) {};
		\node [style=none] (28) at (2.5, 2.25) {};
		\node [style=none] (29) at (2.5, 2.25) {};
		\node [style=none] (30) at (4.5, 2.25) {};
	\end{pgfonlayer}
	\begin{pgfonlayer}{edgelayer}
		\draw (2.center) to (3.center);
		\draw [bend left=45] (5.center) to (3.center);
		\draw [bend left=45] (3.center) to (4.center);
		\draw (9.center) to (10.center);
		\draw (4.center) to (11.center);
		\draw (15.center) to (16.center);
		\draw [bend left=45] (18.center) to (16.center);
		\draw [bend left=45] (16.center) to (17.center);
		\draw (17.center) to (23.center);
		\draw (29.center) to (27.center);
		\draw (22.center) to (30.center);
	\end{pgfonlayer}
\end{tikzpicture}
	\end{cdiag}
\end{definition}

To explain the term ``Bayesian inverse'', note that in $\stoch$ the left-hand side of \eqref{diag:bayes_inv} is a Markov kernel we may sample from via
\[
	X \sim \pi(\dif x) \qquad\qquad Y \sim \P(X, \dif y) \qquad\qquad \text{return $(X, Y)$}.
\]
The equation \eqref{diag:bayes_inv} then says that $\P^\dagger(y, \dif x)$ is the conditional distribution of $X$ given $Y = y$, which is equivalently the posterior distribution of $X$ obtained via Bayes' theorem.

\subsection{Invariance and reversibility} \label{subsec:markov_chain}

Among the most fundamental concepts from Markov chain theory are \emph{invariance} and \emph{reversibility}.
In the context of MCMC, these constitute key correctness properties that ensure an algorithm targets the intended distribution.
We now show how these concepts can be formulated in Markov categories.

\begin{definition} \label{defn:invariance-reversibility}
	Let $\mu : \ind \to \tarS$ and $\P : \tarS \to \tarS$ be morphisms in a Markov category $\C$.
	We say that $\P$ is \textit{$\mu$-invariant} if $\P \circ \mu = \mu$.
	We say that $\P$ is \emph{$\mu$-reversible} if
	\begin{cdiag} \label{eq:reversibility}
		\begin{tikzpicture}
	\begin{pgfonlayer}{nodelayer}
		\node [style=none] (0) at (0, 0) {$=$};
		\node [style=state] (1) at (-3.5, -1.5) {$\mu$};
		\node [style=none] (2) at (-3.5, -1.25) {};
		\node [style=none] (3) at (-3.5, -0.5) {};
		\node [style=none] (4) at (-4.5, 0.5) {};
		\node [style=none] (5) at (-2.5, 0.5) {};
		\node [style=morphism] (6) at (-2.5, 0.5) {$\P$};
		\node [style=none] (7) at (-2.5, 1.5) {};
		\node [style=none] (8) at (-4.5, 1.5) {};
		\node [style=none] (9) at (-4.5, 2) {$\tarS$};
		\node [style=none] (10) at (-2.5, 2) {$\tarS$};
		\node [style=bn] (11) at (-3.5, -0.5) {};
		\node [style=state] (12) at (3.5, -1.5) {$\mu$};
		\node [style=none] (13) at (3.5, -1.25) {};
		\node [style=none] (14) at (3.5, -0.5) {};
		\node [style=none] (15) at (2.5, 0.5) {};
		\node [style=none] (16) at (4.5, 0.5) {};
		\node [style=morphism] (17) at (2.5, 0.5) {$\P$};
		\node [style=none] (18) at (4.5, 1.5) {};
		\node [style=none] (19) at (2.5, 1.5) {};
		\node [style=none] (20) at (2.5, 2) {$\tarS$};
		\node [style=none] (21) at (4.5, 2) {$\tarS$};
		\node [style=bn] (22) at (3.5, -0.5) {};
	\end{pgfonlayer}
	\begin{pgfonlayer}{edgelayer}
		\draw (2.center) to (3.center);
		\draw [bend left=45] (3.center) to (4.center);
		\draw [bend right=45] (3.center) to (5.center);
		\draw (4.center) to (8.center);
		\draw (7.center) to (5.center);
		\draw (13.center) to (14.center);
		\draw [bend left=45] (14.center) to (15.center);
		\draw [bend right=45] (14.center) to (16.center);
		\draw (15.center) to (19.center);
		\draw (18.center) to (16.center);
	\end{pgfonlayer}
\end{tikzpicture}
	\end{cdiag}
	Equivalently, this says that the left-hand side of \eqref{eq:reversibility} is invariant with respect to the swap isomorphism $\tarS \otimes \tarS \xrightarrow{\cong} \tarS \otimes \tarS$, or that $P$ is its own Bayesian inverse with respect to $\mu$.
\end{definition}

\begin{example} \label{ex:invariance-reversibility-stoch}
	Consider Definition~\ref{defn:invariance-reversibility} with $\C = \stoch$. %
	Here, $\tarS$ becomes a measurable space, $\mu$ a probability measure on $\tarS$, and $\P : \tarS \to \tarS$ a Markov kernel.
	Then $\P$ is $\mu$-invariant if and only if
	$
	 	\int \mu(\dif x) \, \P(x, A) = \mu(A)
	$
	for every measurable $A \subseteq \tarS$, recovering the usual definition of invariance, as observed by \cite[Section 2.4]{Moss2023}.
    Intuitively, invariance says that if a Markov chain with transition kernel $\P$ is initialised with distribution $\mu$, then its marginal distribution will remain $\mu$ at all time steps afterwards.
	Likewise, \eqref{eq:reversibility} says that $\int_A \mu(\dif x) \, \P(x, B) = \int_B \mu(\dif x') \, \P(x', A)$ for all measurable $A, B \subseteq \tarS$, or more colloquially
	\[
		\mu(\dif x) \, \P(x, \dif x') = \mu(\dif x') \, \P(x', \dif x).
	\]
	This recovers the classical \emph{detailed balance} condition from Markov chain theory, as observed by \cite[Remark~4.1.11]{Fritz2023}.
    Intuitively, detailed balance says that if a Markov chain with transition kernel $\P$ is initialised with distribution $\mu$, then the joint distribution of any two consecutive states is invariant under swapping them.
\end{example}

Invariance is usually the fundamental correctness property sought for MCMC algorithms, since it ensures that the Markov chain will converge to the desired distribution under mild conditions regardless of its initial state \cite{roberts2004general}.
Reversibility is a sufficient condition for invariance that is often convenient to verify in practice.
Markov categories are very well-suited for reasoning about these concepts, as the following result illustrates.

\begin{proposition} \label{prop:reversibility-implies-invariance}
	Let $\mu : \ind \to \tarS$ and $\P, \Q : \tarS \to \tarS$ be morphisms in a Markov category $\C$.
	If both $\P$ and $\Q$ are $\mu$-invariant, then so is $\P \circ \Q$.
	Likewise, if $\P$ is $\mu$-reversible, then it is also $\mu$-invariant.
\end{proposition}

\begin{proof}
	For the first claim, by the definition of invariance, we have that $\P \circ \Q \circ \mu = \P \circ \mu = \mu$.
	For the second claim, observe that:
	\begin{cdiag*}
		\begin{tikzpicture}
	\begin{pgfonlayer}{nodelayer}
		\node [style=state] (0) at (-8.5, -1.5) {$\mu$};
		\node [style=morphism] (1) at (-8.5, 0) {$\P$};
		\node [style=none] (2) at (-8.5, 1.25) {};
		\node [style=none] (3) at (-8.5, 1.75) {$\tarS$};
		\node [style=none] (4) at (-6.5, 0) {$=$};
		\node [style=state] (5) at (-4, -1.5) {$\mu$};
		\node [style=none] (6) at (-4, -0.25) {};
		\node [style=none] (8) at (-5, 0.75) {};
		\node [style=morphism] (9) at (-3, 0.75) {$\P$};
		\node [style=bn] (10) at (-4, -0.25) {};
		\node [style=none] (11) at (-3, 1.75) {};
		\node [style=none] (12) at (-3, 2.25) {$\tarS$};
		\node [style=none] (13) at (-1, 0) {$=$};
		\node [style=state] (14) at (2, -1.5) {$\mu$};
		\node [style=none] (15) at (2, -0.25) {};
		\node [style=morphism] (16) at (1, 1) {$\P$};
		\node [style=none] (17) at (3, 0.75) {};
		\node [style=bn] (18) at (2, -0.25) {};
		\node [style=none] (19) at (3, 1.75) {};
		\node [style=none] (20) at (3, 2.25) {$\tarS$};
		\node [style=bn] (21) at (1, 2) {};
		\node [style=none] (22) at (4.5, 0) {$=$};
		\node [style=state] (23) at (7, -1.5) {$\mu$};
		\node [style=none] (24) at (7, -0.25) {};
		\node [style=bn] (25) at (6, 1) {};
		\node [style=none] (26) at (8, 0.75) {};
		\node [style=bn] (27) at (7, -0.25) {};
		\node [style=none] (28) at (8, 1.5) {};
		\node [style=none] (29) at (8, 2) {$\tarS$};
		\node [style=none] (30) at (9.5, 0) {$=$};
		\node [style=state] (31) at (11.5, -1) {$\mu$};
		\node [style=none] (36) at (11.5, 1) {};
		\node [style=none] (37) at (11.5, 1.5) {$\tarS$};
		\node [style=bn] (38) at (-5, 2) {};
	\end{pgfonlayer}
	\begin{pgfonlayer}{edgelayer}
		\draw (2.center) to (1);
		\draw (1) to (0);
		\draw [bend right=45, looseness=0.75] (6.center) to (9);
		\draw [bend right=45] (8.center) to (6.center);
		\draw (10) to (5);
		\draw (9) to (11.center);
		\draw [bend right=45] (15.center) to (17.center);
		\draw [bend right=45, looseness=0.75] (16) to (15.center);
		\draw (18) to (14);
		\draw (17.center) to (19.center);
		\draw (21) to (16);
		\draw [bend right=45] (24.center) to (26.center);
		\draw [bend right=45] (25) to (24.center);
		\draw (27) to (23);
		\draw (26.center) to (28.center);
		\draw (36.center) to (31);
		\draw (38) to (8.center);
	\end{pgfonlayer}
\end{tikzpicture}
	\end{cdiag*}
	where the second step uses reversibility, and the others use the basic Markov category axioms (see Appendix \ref{app:intro_to_category_theory}).
\end{proof}

It is well known by MCMC practitioners that invariance does not imply reversibility in general.
However, this does hold for \emph{involutions}, where recall $\phi : \tarS \to \tarS$ is an involution if $\phi \circ \phi = \Id_{\tarS}$.

\begin{proposition} \label{prop:inv_rev}
	Let $\mu: \ind \to \tarS$ and $\phi: \tarS \to \tarS$ be morphisms in a Markov category $\C$, where $\phi$ is a deterministic involution.
	If $\phi$ is $\mu$-invariant, then it is $\mu$-reversible.
\end{proposition}

\begin{proof}
	If $\phi$ is $\mu$-invariant, then we have the following
	\begin{cdiag} \label{diag:mu_Phi_inv_pf}
		\begin{tikzpicture}
	\begin{pgfonlayer}{nodelayer}
		\node [style=bn] (78) at (-24, -0.75) {};
		\node [style=none] (80) at (-25, 0.25) {};
		\node [style=none] (81) at (-23, 0.25) {};
		\node [style=morphism] (82) at (-23, 1.25) {$\phi$};
		\node [style=none] (89) at (-23, 2.5) {};
		\node [style=none] (90) at (-25, 2.5) {};
		\node [style=none] (91) at (-24, -2) {};
		\node [style=state] (92) at (-24, -2) {$\mu$};
		\node [style=none] (93) at (-21, 0.5) {$=$};
		\node [style=bn] (97) at (-18, -1) {};
		\node [style=none] (99) at (-19, 0) {};
		\node [style=none] (100) at (-17, 0) {};
		\node [style=morphism] (101) at (-17, 0.75) {$\phi$};
		\node [style=none] (108) at (-17, 2.5) {};
		\node [style=none] (109) at (-19, 2.5) {};
		\node [style=none] (110) at (-18, -2) {};
		\node [style=state] (111) at (-18, -2) {$\mu$};
		\node [style=morphism] (112) at (-19, 1.5) {$\phi$};
		\node [style=morphism] (113) at (-19, 0) {$\phi$};
		\node [style=none] (114) at (-15, 0.5) {$=$};
		\node [style=bn] (118) at (-12, 0.25) {};
		\node [style=none] (120) at (-13, 1.25) {};
		\node [style=none] (121) at (-11, 1.25) {};
		\node [style=none] (128) at (-11, 2.5) {};
		\node [style=none] (129) at (-13, 2.5) {};
		\node [style=state] (131) at (-12, -2) {$\mu$};
		\node [style=morphism] (132) at (-13, 1.5) {$\phi$};
		\node [style=morphism] (134) at (-12, -0.75) {$\phi$};
		\node [style=none] (135) at (-9, 0.5) {$=$};
		\node [style=bn] (136) at (-6, -0.75) {};
		\node [style=none] (137) at (-7, 0.25) {};
		\node [style=none] (138) at (-5, 0.25) {};
		\node [style=none] (139) at (-5, 2.5) {};
		\node [style=none] (140) at (-7, 2.5) {};
		\node [style=state] (142) at (-6, -2) {$\mu$};
		\node [style=morphism] (143) at (-7, 1.25) {$\phi$};
		\node [style=none] (144) at (-25, 3) {$\tarS$};
		\node [style=none] (145) at (-23, 3) {$\tarS$};
		\node [style=none] (146) at (-19, 3) {$\tarS$};
		\node [style=none] (147) at (-17, 3) {$\tarS$};
		\node [style=none] (148) at (-13, 3) {$\tarS$};
		\node [style=none] (149) at (-11, 3) {$\tarS$};
		\node [style=none] (150) at (-7, 3) {$\tarS$};
		\node [style=none] (151) at (-5, 3) {$\tarS$};
	\end{pgfonlayer}
	\begin{pgfonlayer}{edgelayer}
		\draw [in=180, out=-90, looseness=1.25] (80.center) to (78);
		\draw [in=270, out=0, looseness=1.25] (78) to (81.center);
		\draw (81.center) to (82);
		\draw (80.center) to (90.center);
		\draw (89.center) to (82);
		\draw [in=180, out=-90, looseness=1.25] (99.center) to (97);
		\draw [in=270, out=0, looseness=1.25] (97) to (100.center);
		\draw (100.center) to (101);
		\draw (99.center) to (109.center);
		\draw (108.center) to (101);
		\draw [in=180, out=-90, looseness=1.25] (120.center) to (118);
		\draw [in=270, out=0, looseness=1.25] (118) to (121.center);
		\draw (120.center) to (129.center);
		\draw (121.center) to (128.center);
		\draw (78) to (92);
		\draw (97) to (111);
		\draw (118) to (134);
		\draw (131) to (134);
		\draw [in=180, out=-90, looseness=1.25] (137.center) to (136);
		\draw [in=270, out=0, looseness=1.25] (136) to (138.center);
		\draw (137.center) to (140.center);
		\draw (138.center) to (139.center);
		\draw (142) to (136);
	\end{pgfonlayer}
\end{tikzpicture}
	\end{cdiag}
	where the first step uses the involution property, the second uses determinism, and the third uses invariance.
\end{proof}

\begin{remark}
	The assumption of determinism in Proposition~\ref{prop:inv_rev} is redundant in many practical settings.
	This is because an involution is by definition an isomorphism, and for the class of \emph{positive} Markov categories \cite[Definition 11.22]{Fritz2020}, every isomorphism is necessarily deterministic \cite[Remark 11.28]{Fritz2020}.
	Many Markov categories arising in practice are positive, including $\stoch$ \cite[Example 11.25]{Fritz2020}. %
\end{remark}

\subsection{Skew-reversibility} \label{subsec:skew-reversibility}

Reversible Markov chains often exhibit a diffusive ``random walk'' behaviour, and reversible MCMC algorithms can often perform poorly in practice as a result \cite{Diaconis2000}. 
To address this, a variety of generalised notions of reversibility have been considered in the MCMC literature, including \textit{lifting} \cite{Chen1999,Diaconis2000} and \textit{skew-reversibility} \cite{Turitsyn2011,Sakai2016,thin2021nonreversible}. %
These all admit related formulations in terms of an involution acting on the state space.
Recently, \cite[Definition 2.4]{Andrieu2021} gave a general notion of \emph{$(\pi, s)$-reversibility} which subsumes these previous concepts, but at the cost of introducing heavier measure-theoretic machinery.
We show how these concepts can be streamlined considerably in the context of Markov categories.
Later on, this will yield a straightforward extension of our synthetic version of Theorem \ref{thm:andrieu2020} to the skew-reversible setting (Corollary \ref{cor:skew}).

\begin{definition} \label{def:pi_s_rev}
	Let $\pi : \ind \to \tarS$ and $\s, \P : \tarS \to \tarS$ be morphisms in a Markov category $\C$, where $\s$ is a deterministic involution that is $\pi$-invariant.
	We say that $\P$ is \emph{$(\pi, s)$-reversible} if the following equality holds:
	\begin{cdiag} \label{eq:pi_s_rev_def}
		\begin{tikzpicture}
	\begin{pgfonlayer}{nodelayer}
		\node [style=none] (0) at (0, 0) {$=$};
		\node [style=state] (1) at (-3, -1.75) {$\pi$};
		\node [style=none] (2) at (-3, -1.5) {};
		\node [style=none] (3) at (-3, -0.75) {};
		\node [style=none] (4) at (-4, 0.25) {};
		\node [style=none] (5) at (-2, 0.25) {};
		\node [style=none] (6) at (-4, 1.5) {};
		\node [style=morphism] (7) at (-2, 0.5) {$\P$};
		\node [style=none] (8) at (-2, 1.5) {};
		\node [style=morphism] (9) at (2, -0.75) {$\s$};
		\node [style=state] (11) at (3, -2.75) {$\pi$};
		\node [style=none] (12) at (3, -2.5) {};
		\node [style=none] (13) at (3, -1.75) {};
		\node [style=none] (14) at (2, -0.75) {};
		\node [style=none] (15) at (4, -0.75) {};
		\node [style=none] (16) at (2, 2.5) {};
		\node [style=morphism] (17) at (2, 0.5) {$\P$};
		\node [style=none] (18) at (4, 2.5) {};
		\node [style=morphism] (19) at (2, 1.75) {$\s$};
		\node [style=bn] (21) at (3, -1.75) {};
		\node [style=bn] (22) at (-3, -0.75) {};
		\node [style=none] (23) at (-4, 2) {$\tarS$};
		\node [style=none] (24) at (-2, 2) {$\tarS$};
		\node [style=none] (25) at (2, 3) {$\tarS$};
		\node [style=none] (26) at (4, 3) {$\tarS$};
	\end{pgfonlayer}
	\begin{pgfonlayer}{edgelayer}
		\draw [bend right=45] (4.center) to (3.center);
		\draw [bend left=45] (5.center) to (3.center);
		\draw (3.center) to (2.center);
		\draw (4.center) to (6.center);
		\draw (8.center) to (5.center);
		\draw [bend right=45] (14.center) to (13.center);
		\draw [bend left=45] (15.center) to (13.center);
		\draw (13.center) to (12.center);
		\draw (14.center) to (16.center);
		\draw (18.center) to (15.center);
	\end{pgfonlayer}
\end{tikzpicture}
	\end{cdiag}
\end{definition}

\begin{example}
	In $\C=\stoch$, in terms of measure theory, \eqref{eq:pi_s_rev_def} says:
	\[
		\pi(\dif x) \, \P(x, \dif y) = \pi(\dif y) \, (\s \circ \P \circ \s)(y, \dif x)
	\]
	which recovers the definition of $(\pi, s)$-reversibility given in \cite[Definition 2.4]{Andrieu2021} (see also \cite[Proposition 3, conclusion (c)]{Andrieu2020}).
\end{example}

By a similar argument as Proposition \ref{prop:reversibility-implies-invariance}, $(\pi, \s)$-reversibility implies $\pi$-invariance, which explains its usefulness for MCMC.
It is also clear that ordinary $\pi$-reversibility is the special case of $(\pi, \Id_\tarS)$-reversibility.
Interestingly, the following result shows we can also define $(\pi, \s)$-reversibility as a special case of $\pi$-reversibility, and recovers several results on this topic in the MCMC literature (e.g.\ \cite[Proposition~2.5]{Andrieu2021}, \cite[Equation (3)]{thin2021nonreversible}, \cite[Corollary 2]{Andrieu2020}).

\begin{proposition} \label{prop:pi-s-reversibility-condition}
	With $\pi$, $\s$, and $\P$ as in Definition \ref{def:pi_s_rev}, the following are equivalent:
	\begin{enumerate}[(1)]
		\item $\P$ is $(\pi, \s)$-reversible.
		\item $\P \circ \s$ is $\pi$-reversible.
		\item $\s \circ \P$ is $\pi$-reversible.
		\item $\P = \Q \circ \s = \s \circ \R$ for some $\pi$-reversible $\Q, \R : \tarS \to \tarS$.
	\end{enumerate}
\end{proposition}

\begin{proof}
	(1) $\Leftrightarrow$ (2): Consider the following equalities:
	\begin{cdiag*}
		\begin{tikzpicture}
	\begin{pgfonlayer}{nodelayer}
		\node [style=none] (0) at (-5, -1.25) {};
		\node [style=none] (1) at (-6, -0.25) {};
		\node [style=none] (2) at (-4, -0.25) {};
		\node [style=none] (3) at (-5, -2) {};
		\node [style=bn] (4) at (-5, -1.25) {};
		\node [style=state] (5) at (-5, -2.25) {$\pi$};
		\node [style=none] (6) at (-6, 2) {};
		\node [style=none] (7) at (-4, 2) {};
		\node [style=none] (16) at (-2, 0) {$=$};
		\node [style=morphism] (17) at (-4, -0.25) {$\s$};
		\node [style=morphism] (18) at (-4, 1) {$\P$};
		\node [style=none] (21) at (1, -0.75) {};
		\node [style=none] (22) at (0, 0.25) {};
		\node [style=none] (23) at (2, 0.25) {};
		\node [style=none] (24) at (1, -1.5) {};
		\node [style=bn] (25) at (1, -0.75) {};
		\node [style=state] (26) at (1, -1.75) {$\pi$};
		\node [style=none] (27) at (0, 1.25) {};
		\node [style=none] (28) at (2, 1.25) {};
		\node [style=morphism] (29) at (0, 0.25) {$\s$};
		\node [style=morphism] (30) at (2, 0.25) {$\P$};
		\node [style=none] (33) at (7, -2.25) {};
		\node [style=none] (34) at (6, -1.25) {};
		\node [style=none] (35) at (8, -1.25) {};
		\node [style=none] (36) at (7, -3) {};
		\node [style=bn] (37) at (7, -2.25) {};
		\node [style=state] (38) at (7, -3.25) {$\pi$};
		\node [style=none] (39) at (6, 3.25) {};
		\node [style=none] (40) at (8, 3.25) {};
		\node [style=morphism] (41) at (6, -1.25) {$\s$};
		\node [style=morphism] (42) at (6, 0) {$\P$};
		\node [style=none] (45) at (4, 0) {$=$};
		\node [style=morphism] (46) at (6, 1.25) {$\s$};
		\node [style=morphism] (47) at (6, 2.25) {$\s$};
		\node [style=none] (48) at (12.5, -1.25) {};
		\node [style=none] (49) at (11.5, -0.25) {};
		\node [style=none] (50) at (13.5, -0.25) {};
		\node [style=none] (51) at (12.5, -2) {};
		\node [style=bn] (52) at (12.5, -1.25) {};
		\node [style=state] (53) at (12.5, -2.25) {$\pi$};
		\node [style=none] (54) at (11.5, 2) {};
		\node [style=none] (55) at (13.5, 2) {};
		\node [style=morphism] (56) at (11.5, -0.25) {$\s$};
		\node [style=morphism] (57) at (11.5, 1) {$\P$};
		\node [style=none] (60) at (9.5, 0) {$=$};
		\node [style=none] (61) at (-6, 2.5) {$\tarS$};
		\node [style=none] (62) at (-4, 2.5) {$\tarS$};
		\node [style=none] (63) at (0, 1.75) {$\tarS$};
		\node [style=none] (64) at (2, 1.75) {$\tarS$};
		\node [style=none] (65) at (6, 3.75) {$\tarS$};
		\node [style=none] (66) at (8, 3.75) {$\tarS$};
		\node [style=none] (67) at (11.5, 2.5) {$\tarS$};
		\node [style=none] (68) at (13.5, 2.5) {$\tarS$};
	\end{pgfonlayer}
	\begin{pgfonlayer}{edgelayer}
		\draw [in=270, out=180] (0.center) to (1.center);
		\draw (0.center) to (3.center);
		\draw [in=270, out=0] (0.center) to (2.center);
		\draw (6.center) to (1.center);
		\draw (7.center) to (2.center);
		\draw [in=270, out=180] (21.center) to (22.center);
		\draw (21.center) to (24.center);
		\draw [in=270, out=0] (21.center) to (23.center);
		\draw (27.center) to (22.center);
		\draw (28.center) to (23.center);
		\draw [in=270, out=180] (33.center) to (34.center);
		\draw (33.center) to (36.center);
		\draw [in=270, out=0] (33.center) to (35.center);
		\draw (39.center) to (34.center);
		\draw (40.center) to (35.center);
		\draw [in=270, out=180] (48.center) to (49.center);
		\draw (48.center) to (51.center);
		\draw [in=270, out=0] (48.center) to (50.center);
		\draw (54.center) to (49.center);
		\draw (55.center) to (50.center);
	\end{pgfonlayer}
\end{tikzpicture}
	\end{cdiag*}
	The first step always holds by Proposition~\ref{prop:inv_rev}, and the third because $\s$ is an involution.
	The middle step holds if and only if $\P$ is $(\pi, \s)$-reversible, as can be seen by applying $\s$ to both left-hand wires and using the involution property.
    Likewise, the left- and right-hand sides are equal precisely when $P \circ s$ is $\pi$-reversible, which gives the statement.

	(1) $\Leftrightarrow$ (3) follows by a similar argument.

	(2) $\Leftrightarrow$ (4) $\Leftrightarrow$ (3) now follows by taking $\Q \coloneqq \P \circ \s$ and $\R \coloneqq \s \circ \P$.
\end{proof}

\subsection{Reversibility via augmented state spaces} \label{sec:augmented-reversibility}

Given a target distribution on a space $\tarS$, a common MCMC technique is to consider a target on an \emph{augmented space} $\augS$ that is related to $\tarS$, but simpler to work with.
Although widely used, augmentation techniques are often presented on a case-by-case basis for specific choices of the components involved.
However, Markov categories permit a unified explanation of these techniques as follows.

\begin{proposition}\label{prop:marginal_rev_new2}
	Let $\pi : \ind \to \tarS$ be a morphism in a Markov category $\C$.
	Suppose we have a composition in $\C$ as follows:
	\begin{cdiag} \label{eq:Q_T_Q}
		\begin{tikzcd}
			\tarS \ar{r}{\R} & \augS \ar{r}{\P} & \augS \ar{r}{\R^\dagger} & \tarS.
		\end{tikzcd}
	\end{cdiag}
	If $\R^\dagger$ is a Bayesian inverse of $\R$ with respect to $\pi$, and 
	$\P$ is $(\R \circ \pi)$-reversible (resp.\ invariant), then \eqref{eq:Q_T_Q} is $\pi$-reversible (resp.\ invariant) 
\end{proposition}
\begin{proof}
	We show the reversibility statement; the invariance statement follows similarly.
	By basic string diagram manipulations:
	\begin{cdiag*}
		\begin{tikzpicture}
	\begin{pgfonlayer}{nodelayer}
		\node [style=bn] (3) at (-3, -1.75) {};
		\node [style=none] (4) at (-3, -2.5) {};
		\node [style=none] (5) at (-3, -2.75) {};
		\node [style=state] (6) at (-3, -2.75) {$\pi$};
		\node [style=none] (7) at (0, 0.5) {$=$};
		\node [style=none] (8) at (-4, -0.75) {};
		\node [style=none] (9) at (-2, -0.75) {};
		\node [style=morphism] (11) at (-2, 0.75) {$\P$};
		\node [style=none] (12) at (-4, 3.25) {};
		\node [style=bn] (16) at (3, -0.25) {};
		\node [style=none] (17) at (3, -2.25) {};
		\node [style=none] (18) at (3, -2.5) {};
		\node [style=state] (19) at (3, -2.5) {$\pi$};
		\node [style=none] (20) at (2, 0.75) {};
		\node [style=none] (21) at (4, 0.75) {};
		\node [style=morphism] (22) at (4, 1) {$\P$};
		\node [style=morphism] (23) at (2, 1) {$\R^\dagger$};
		\node [style=none] (24) at (2, 1) {};
		\node [style=morphism] (25) at (-2, 2.25) {$\R^\dagger$};
		\node [style=none] (26) at (2, 2.25) {};
		\node [style=none] (27) at (-2, 3.25) {};
		\node [style=none] (28) at (4, 3.25) {};
		\node [style=none] (29) at (2, 3.25) {};
		\node [style=morphism] (62) at (4, 2.25) {$\R^\dagger$};
		\node [style=morphism] (64) at (3, -1.25) {$\R$};
		\node [style=none] (65) at (12, 0.5) {$=$};
		\node [style=none] (82) at (6, 0.5) {$=$};
		\node [style=none] (86) at (9, -0.25) {};
		\node [style=none] (87) at (9, -2.25) {};
		\node [style=none] (88) at (9, -2.5) {};
		\node [style=state] (89) at (9, -2.5) {$\pi$};
		\node [style=none] (90) at (8, 0.75) {};
		\node [style=none] (91) at (10, 0.75) {};
		\node [style=none] (92) at (10, 0.75) {};
		\node [style=morphism] (93) at (8, 1) {$\P$};
		\node [style=morphism] (94) at (8, 2.25) {$\R^\dagger$};
		\node [style=none] (95) at (10, 3.25) {};
		\node [style=none] (96) at (8, 3.25) {};
		\node [style=morphism] (97) at (10, 2.25) {$\R^\dagger$};
		\node [style=morphism] (98) at (9, -1.25) {$\R$};
		\node [style=morphism] (99) at (-2, -0.75) {$\R$};
		\node [style=bn] (100) at (9, -0.25) {};
		\node [style=bn] (101) at (15, -1.75) {};
		\node [style=none] (102) at (15, -2.5) {};
		\node [style=none] (103) at (15, -2.75) {};
		\node [style=state] (104) at (15, -2.75) {$\pi$};
		\node [style=none] (105) at (14, -0.75) {};
		\node [style=none] (106) at (16, -0.75) {};
		\node [style=none] (107) at (16, 3.25) {};
		\node [style=none] (108) at (14, 3.25) {};
		\node [style=morphism] (111) at (14, -0.75) {$\R$};
		\node [style=morphism] (112) at (14, 0.75) {$\P$};
		\node [style=morphism] (113) at (14, 2.25) {$\R^\dagger$};
		\node [style=none] (114) at (-4, 3.75) {$\tarS$};
		\node [style=none] (115) at (-2, 3.75) {$\tarS$};
		\node [style=none] (116) at (2, 3.75) {$\tarS$};
		\node [style=none] (117) at (4, 3.75) {$\tarS$};
		\node [style=none] (118) at (8, 3.75) {$\tarS$};
		\node [style=none] (119) at (10, 3.75) {$\tarS$};
		\node [style=none] (120) at (14, 3.75) {$\tarS$};
		\node [style=none] (121) at (16, 3.75) {$\tarS$};
	\end{pgfonlayer}
	\begin{pgfonlayer}{edgelayer}
		\draw (3) to (4.center);
		\draw [in=-180, out=-90] (8.center) to (3);
		\draw (12.center) to (8.center);
		\draw [in=-180, out=-90] (20.center) to (16);
		\draw [bend left=45] (21.center) to (16);
		\draw (23) to (20.center);
		\draw (22) to (21.center);
		\draw (24.center) to (23);
		\draw (22) to (28.center);
		\draw (11) to (25);
		\draw (25) to (27.center);
		\draw (29.center) to (26.center);
		\draw (26.center) to (24.center);
		\draw (16) to (64);
		\draw (17.center) to (64);
		\draw [in=-180, out=-90] (90.center) to (86.center);
		\draw [bend left=45] (91.center) to (86.center);
		\draw (92.center) to (91.center);
		\draw (92.center) to (95.center);
		\draw (96.center) to (94);
		\draw (94) to (93);
		\draw (90.center) to (93);
		\draw (87.center) to (98);
		\draw (98) to (86.center);
		\draw [bend right=45] (3) to (9.center);
		\draw (9.center) to (11);
		\draw (101) to (102.center);
		\draw [in=-180, out=-90] (105.center) to (101);
		\draw (108.center) to (105.center);
		\draw [bend right=45] (101) to (106.center);
		\draw (106.center) to (107.center);
	\end{pgfonlayer}
\end{tikzpicture}
	\end{cdiag*}
	where the first and third steps use the definition of a Bayesian inverse (Definition~\ref{def:bayesian-inverse}), and the second uses reversibility of $\P$.
\end{proof}

In the context of involutive MH, Andrieu et. al. \cite{Andrieu2020} combine state space augmentation with their kernel \eqref{eq:mh-kernel} to recover a wide variety of existing MCMC algorithms.
The scheme they consider involves introducing auxiliary random variables taking values in a space $\auxS$, so that the augmented space has the form $\augS = \tarS \otimes \auxS$.
The following result captures their approach in the setting of Markov categories.

\begin{corollary} \label{cor:derived-reversible-morphism}
	Let $\pi : \ind \to \tarS$, $\Q : \tarS \to \auxS$, and $\P : \tarS \otimes \auxS \to \tarS \otimes \auxS$ be morphisms in a Markov category $\C$, and denote:
	\begin{cdiag} \label{eq:mu-definition-in-proof}
		\begin{tikzpicture}
	\begin{pgfonlayer}{nodelayer}
		\node [style=state] (0) at (-2, -0.5) {$\mu$};
		\node [style=none] (1) at (-2.5, -0.25) {};
		\node [style=none] (2) at (-2.5, 0.5) {};
		\node [style=none] (3) at (-2.5, 1) {$\tarS$};
		\node [style=none] (4) at (-1.5, -0.25) {};
		\node [style=none] (5) at (-1.5, 0.5) {};
		\node [style=none] (6) at (-1.5, 1) {$\auxS$};
		\node [style=none] (7) at (0, 0) {$\coloneqq$};
		\node [style=state] (8) at (2.25, -1.5) {$\pi$};
		\node [style=none] (9) at (2.25, -1.25) {};
		\node [style=none] (10) at (2.25, -0.5) {};
		\node [style=none] (11) at (1.25, 0.5) {};
		\node [style=none] (12) at (3.25, 0.5) {};
		\node [style=morphism] (13) at (3.25, 0.5) {$\Q$};
		\node [style=none] (14) at (3.25, 1.5) {};
		\node [style=none] (15) at (1.25, 1.5) {};
		\node [style=bn] (16) at (2.25, -0.5) {};
		\node [style=none] (17) at (1.25, 2) {$\tarS$};
		\node [style=none] (18) at (3.25, 2) {$\auxS$};
		\node [style=none] (19) at (13.25, -2.25) {};
		\node [style=none] (20) at (13.25, -1.25) {};
		\node [style=none] (21) at (12.25, -0.25) {};
		\node [style=none] (22) at (14.25, -0.25) {};
		\node [style=bn] (23) at (13.25, -1.25) {};
		\node [style=none] (24) at (14.25, 1.5) {};
		\node [style=morphism] (25) at (14.25, 0) {$\Q$};
		\node [style=none] (26) at (12.25, 1.5) {};
		\node [style=medium large morphism] (27) at (13.25, 1.5) {$\P$};
		\node [style=none] (28) at (12.25, 1.75) {};
		\node [style=none] (29) at (14.25, 1.75) {};
		\node [style=none] (30) at (12.25, 3) {};
		\node [style=bn] (31) at (14.25, 2.5) {};
		\node [style=none] (32) at (13.25, -2.75) {$\tarS$};
		\node [style=none] (33) at (12.25, 3.5) {$\tarS$};
		\node [style=none] (34) at (11, 0) {$\coloneqq$};
		\node [style=morphism] (35) at (9.5, 0) {$\Pi$};
		\node [style=none] (36) at (9.5, 0.25) {};
		\node [style=none] (37) at (9.5, 1.5) {};
		\node [style=none] (38) at (9.5, -0.25) {};
		\node [style=none] (39) at (9.5, -1.5) {};
		\node [style=none] (40) at (9.5, 2) {$\tarS$};
		\node [style=none] (41) at (9.5, -2) {$\tarS$};
	\end{pgfonlayer}
	\begin{pgfonlayer}{edgelayer}
		\draw (1.center) to (2.center);
		\draw (5.center) to (4.center);
		\draw (9.center) to (10.center);
		\draw [bend left=45] (12.center) to (10.center);
		\draw [bend right=315] (10.center) to (11.center);
		\draw (12.center) to (14.center);
		\draw (15.center) to (11.center);
		\draw [bend left=45] (22.center) to (20.center);
		\draw [bend right=315] (20.center) to (21.center);
		\draw (19.center) to (20.center);
		\draw (24.center) to (22.center);
		\draw (21.center) to (26.center);
		\draw (29.center) to (31);
		\draw (28.center) to (30.center);
		\draw (39.center) to (38.center);
		\draw (36.center) to (37.center);
	\end{pgfonlayer}
\end{tikzpicture}
	\end{cdiag}
	If $\P$ is $\mu$-reversible, then $\Pi$ is $\pi$-reversible.
\end{corollary}

\begin{proof}
	By defining
	\begin{cdiag} \label{eq:noise-deletion}
		\begin{tikzpicture}
	\begin{pgfonlayer}{nodelayer}
		\node [style=none] (0) at (2.25, -0.5) {};
		\node [style=none] (1) at (1.25, 0.5) {};
		\node [style=none] (2) at (3.25, 0.5) {};
		\node [style=none] (3) at (2.25, -1.5) {};
		\node [style=none] (4) at (2.25, -2) {$\tarS$};
		\node [style=bn] (5) at (2.25, -0.5) {};
		\node [style=none] (6) at (3.25, 1.5) {};
		\node [style=none] (7) at (1.25, 1.5) {};
		\node [style=morphism] (8) at (3.25, 0.5) {$\Q$};
		\node [style=none] (9) at (1.25, 2) {$\tarS$};
		\node [style=none] (10) at (3.25, 2) {$\auxS$};
		\node [style=none] (11) at (12.5, -0.75) {};
		\node [style=none] (14) at (14, -0.75) {};
		\node [style=none] (15) at (12.5, -1.25) {$\tarS$};
		\node [style=none] (16) at (14, -1.25) {$\auxS$};
		\node [style=bn] (17) at (14, 0.25) {};
		\node [style=none] (19) at (12.5, 0.75) {};
		\node [style=none] (20) at (12.5, 1.25) {$\tarS$};
		\node [style=none] (21) at (0, 0) {$\coloneqq$};
		\node [style=medium box] (22) at (-2, 0) {$\R$};
		\node [style=none] (24) at (-2, -1.25) {};
		\node [style=none] (25) at (-2.5, 0.25) {};
		\node [style=none] (26) at (-1.5, 0.25) {};
		\node [style=none] (27) at (-2.5, 1.25) {};
		\node [style=none] (28) at (-1.5, 1.25) {};
		\node [style=none] (29) at (-2.5, 1.75) {$\tarS$};
		\node [style=none] (30) at (-1.5, 1.75) {$\auxS$};
		\node [style=none] (31) at (-2, -0.25) {};
		\node [style=none] (32) at (-2, -1.75) {$\tarS$};
		\node [style=none] (33) at (11, 0) {$\coloneqq$};
		\node [style=medium box] (34) at (9, 0) {$\R^\dagger$};
		\node [style=none] (35) at (8.5, -0.25) {};
		\node [style=none] (36) at (9.5, -0.25) {};
		\node [style=none] (37) at (8.5, -1.25) {};
		\node [style=none] (38) at (9.5, -1.25) {};
		\node [style=none] (39) at (8.5, -1.75) {$\tarS$};
		\node [style=none] (40) at (9.5, -1.75) {$\auxS$};
		\node [style=none] (41) at (9, 0.25) {};
		\node [style=none] (42) at (9, 1.25) {};
		\node [style=none] (43) at (9, 1.75) {$\tarS$};
	\end{pgfonlayer}
	\begin{pgfonlayer}{edgelayer}
		\draw [bend left=45] (2.center) to (0.center);
		\draw [bend left=45] (0.center) to (1.center);
		\draw (3.center) to (0.center);
		\draw (6.center) to (2.center);
		\draw (7.center) to (1.center);
		\draw (17) to (14.center);
		\draw (19.center) to (11.center);
		\draw (26.center) to (28.center);
		\draw (25.center) to (27.center);
		\draw (31.center) to (24.center);
		\draw (38.center) to (36.center);
		\draw (37.center) to (35.center);
		\draw (41.center) to (42.center);
	\end{pgfonlayer}
\end{tikzpicture}
	\end{cdiag}
	we have $\mu = \R \circ \pi$ and $\Pi = \R^\dagger \circ \P \circ \R$.
	It is straightforward to check that $\R^\dagger$ is a Bayesian inverse of $\R$ with respect to $\pi$, and so the result follows from Proposition \ref{prop:marginal_rev_new2}, with $\augS = \tarS \otimes \auxS$.
\end{proof}

\begin{example}\label{exa:MH}
	When $\C = \stoch$, and given an input $x \in \tarS$, a sampling procedure for \eqref{eq:mu-definition-in-proof} may be read off as follows:
	\[
        Z \sim \Q(x, \dif z) \qquad\qquad (X', Z') \sim \P((x, Z), (\dif x', \dif z')) \qquad\qquad \text{return $X'$}.
	\]
	When $\P$ is the kernel $\Pimh$ defined in \eqref{eq:mh-kernel}, this recovers the augmentation scheme from \cite[Algorithm 1]{Andrieu2020}.
	In turn, by choosing $\auxS$ and the involution $\phi$ in $\Pimh$ appropriately, \cite{Andrieu2020} show this recovers a wide variety of existing MH-type algorithms used in practice (see Appendix \ref{app:MCMC_involutions} for two well-known examples).
\end{example}

\section{The first Metropolis term in CD categories} \label{sec:CD_cats}

While Markov kernels are the ultimate object of interest in many computational statistics applications, it is often useful to consider intermediate constructions involving kernels that are not necessarily normalised.
For example, to show reversibility of the MH kernel \eqref{eq:mh-kernel}, the authors of \cite{Andrieu2020} initially consider just its first summand, i.e.\
\begin{equation}\label{eq:mh-first-summand}
	\alpha(\xi) \, \delta_{\phi(\xi)}(\dif \xi'),
\end{equation}
which does not necessarily integrate to one.
Later, they combine this with the second summand to show reversibility of the full Markov kernel.
We therefore move from Markov categories to the more general setting of \emph{copy-discard (CD) categories} \cite{Cho2019}, which permit the consideration of unnormalised kernels like \eqref{eq:mh-first-summand}.

\subsection{Recap: CD categories}

CD categories are essentially just Markov categories without the requirement that \eqref{eq:del_natural} holds.
More formally:

\begin{definition}[\cite{Cho2019}]\label{def:cd_cat}
	A \emph{copy-discard (CD) category} is a symmetric monoidal category $(\C, \otimes, \ind)$ where all objects are equipped commutative comonoid structure $(\cop_\tarS, \del_\tarS)$ that is suitably compatible with the monoidal structure (see Definition \ref{def:cd_cat_full} in the Appendix).
\end{definition}

\begin{example} \label{exa:sfkern}
	Given measurable spaces $\tarS$ and $\sndS$, a \emph{kernel} is a function $\K : \tarS \times \Sigma_{\sndS} \to [0, \infty]$ such that each $x \mapsto \K(x, A)$ is a measurable function and each $A \mapsto \K(x, A)$ is a measure.
	In other words, a kernel is a Markov kernel that need not integrate to one.
	Measurable spaces and kernels form a category $\Kern$, with composition and identities defined analogously to $\stoch$ (Example \ref{ex:stoch-defn}).
    (Some care is required to show that composition in $\Kern$ is associative; see Appendix \ref{app:kern} for details.)
	However, $\Kern$ is \textit{not} a CD category in an obvious way, since Fubini's theorem does not hold for general kernels, and so \eqref{eq:stoch-monoidal-prod} no longer defines a monoidal product in this setting \cite[Section 1.2]{Staton2017}.
	A solution is to consider \emph{$s$-finite} kernels instead \cite{Staton2017} (see also Remark~\ref{rem:sigma_finite} for why \emph{$\sigma$-finiteness} is not used).
	Recall that a kernel $\K$ is $s$-finite if it can be written as a countable sum $\K = \sum_{n=1}^\infty \K_i$ where each kernel $\K_n$ satisfies $\sup_{x \in \tarS} \K_n(x, \sndS) < \infty$ \cite[Page~56]{Kallenberg2021}.
	It is shown in \cite[Example 7.2]{Cho2019} that measurable spaces and $s$-finite kernels form a CD category that we denote by $\sfkern$, with its CD structure defined analogously to $\stoch$ (Example \ref{ex:stoch-defn}).
\end{example}

\begin{remark} \label{rem:normalised-morphisms}
	The axiom \eqref{eq:del_natural} that distinguishes CD categories from Markov categories may be regarded as a \emph{normalisation} condition.
	For example, in $\sfkern$, this condition holds for $\K : \tarS \to \sndS$ if and only if $\K(x, \sndS) = 1$ for all $x \in \tarS$.
	Accordingly, in a general CD category, we say that morphisms satisfying \eqref{eq:del_natural} are \emph{normalised}.\footnote{\cite[Definition 3.2]{Cho2019} refers to such morphisms as \emph{causal}.}
	Given any CD category $\C$, the objects of $\C$ and its normalised morphisms form a Markov category that we call $\Cnorm$ \cite[Section 7]{Cho2019}.
\end{remark}

In a CD category, morphisms of the form $\w : \tarS \to \ind$ play a special role.
Following \cite[Section 7]{Cho2019}, we refer to these morphisms as \emph{effects}, and depict them in string diagrams as
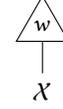
\begin{cdiag*}
	\begin{tikzpicture}
	\begin{pgfonlayer}{nodelayer}
		\node [style=none] (3) at (0, 0) {};
		\node [style=likelihood] (4) at (0, 0.25) {$\w$};
		\node [style=none] (5) at (0, -1) {};
		\node [style=none] (6) at (0, -1.5) {$\tarS$};
	\end{pgfonlayer}
	\begin{pgfonlayer}{edgelayer}
		\draw (3.center) to (5.center);
	\end{pgfonlayer}
\end{tikzpicture}

\end{cdiag*}

\begin{example} \label{exa:effects-sfkern}
	Effects in $\sfkern$ correspond to nonnegative measurable functions \cite[Example 7.2]{Cho2019}.
	Explicitly, since $\ind \cong \{\ast\}$ is the singleton space (Example \ref{exa:sfkern}), a function $\w : \tarS \times \Sigma_\ind \to [0, \infty]$ is a kernel if and only if $x \mapsto \w(x, \ind)$ is measurable.
	In what follows, we will identify an effect with its corresponding measurable function to streamline notation (so $\w(x, \ind)$ becomes simply $\w(x)$).
\end{example}

Effects recover \emph{Radon--Nikodym derivatives} from measure theory.
This is made precise via the following definition, which is a special case of the \emph{likelihoods} considered by \cite[Definition 8.1]{Cho2019}.

\begin{definition} \label{def:radon-nikodym}
	Let $\C$ be a CD category, and let $\pi, \mu: \ind \to \tarS$ and $\rx : \tarS \to \ind$ be morphisms in $\C$.
	We call $\rx$ a \emph{Radon--Nikodym derivative of $\pi$ with respect to $\mu$} and write $\rx = \frac{\dif \pi}{\dif \mu}$ if
	\begin{cdiag} \label{diag:IS_basic}
		\begin{tikzpicture}
	\begin{pgfonlayer}{nodelayer}
		\node [style=none] (21) at (0, -0.25) {$=$};
		\node [style=none] (22) at (4, 0.5) {};
		\node [style=none] (23) at (3, -0.5) {};
		\node [style=none] (24) at (2, 0.5) {};
		\node [style=none] (25) at (3, -1.25) {};
		\node [style=bn] (26) at (3, -0.5) {};
		\node [style=state] (27) at (3, -1.5) {$\mu$};
		\node [style=likelihood] (28) at (2, 0.75) {$\rx$};
		\node [style=none] (29) at (4, 1.5) {};
		\node [style=none] (30) at (4, 2) {$\tarS$};
		\node [style=state] (31) at (-2, -0.5) {$\pi$};
		\node [style=none] (32) at (-2, -0.25) {};
		\node [style=none] (33) at (-2, 0.5) {};
		\node [style=none] (34) at (-2, 1) {$\tarS$};
	\end{pgfonlayer}
	\begin{pgfonlayer}{edgelayer}
		\draw [bend right=315] (23.center) to (24.center);
		\draw [bend left=45] (22.center) to (23.center);
		\draw (23.center) to (25.center);
		\draw (22.center) to (29.center);
		\draw (33.center) to (32.center);
	\end{pgfonlayer}
\end{tikzpicture}
	\end{cdiag}
\end{definition}

\begin{example} \label{exa:radon-nikodym}
	Consider Definition \ref{def:radon-nikodym} in the case of $\sfkern$.
	The condition \eqref{diag:IS_basic} holds if and only if
	$
		\pi(A) = \int_{A} \rx(x) \, \mu(\dif x)
	$
	for all measurable $A \subseteq \tarS$.
	This says precisely that $\rx$ is a Radon--Nikodym derivative of $\pi$ with respect to $\mu$ in the usual measure-theoretic sense.
\end{example}

\begin{remark} \label{rem:CD-cat-Markov-cat}
	In Section \ref{sec:Markov-cats}, we gave definitions of \textit{determinism}, \textit{invariance}, and \textit{reversibility} in the setting of Markov categories.
	In what follows, we will import these concepts directly into the setting of CD categories, noting that they continue to make sense without the assumption that all morphisms are normalised.
	For example, when we speak of a deterministic morphism in a CD category, we simply mean that the equation \eqref{eq:det_phi} holds.\footnote{\cite{fritz2023weakly} describes such a morphism as \emph{copyable}, and defines ``deterministic'' to mean \emph{both} copyable and normalised.}
	We will also import Propositions \ref{prop:inv_rev} and \ref{prop:pi-s-reversibility-condition}, whose proofs do not require normalisation either.
\end{remark}

\subsection{The first summand of Metropolis--Hastings}\label{subsec:Synthetic-MH}

Recall that our overall goal is to understand when $\Pimh$ defined in \eqref{eq:mh-kernel} is reversible with respect to a given target measure $\mu$.
In the proof given by Andrieu et al.\ \cite[Theorem 3]{Andrieu2020}, considerable effort is spent showing $\mu$-reversibility of the first summand \eqref{eq:mh-first-summand}. %
Their argument requires manipulating chains of integrals that do not necessarily shed light on the key underlying structure at play.

CD categories allow us to reason about the summand \eqref{eq:mh-first-summand} in a more conceptual way.
In string diagrams, this kernel may be written as follows:
\begin{cdiag} \label{eq:mh-first-term}
	\begin{tikzpicture}
	\begin{pgfonlayer}{nodelayer}
		\node [style=none] (0) at (-3, -0.75) {};
		\node [style=none] (1) at (-4, 0.25) {};
		\node [style=none] (2) at (-2, 0.25) {};
		\node [style=likelihood] (3) at (-4, 0.5) {$\alpha$};
		\node [style=none] (4) at (-3, -1.75) {};
		\node [style=bn] (5) at (-3, -0.75) {};
		\node [style=none] (6) at (-3, -2.25) {$\augS$};
		\node [style=none] (7) at (-2, 1.75) {};
		\node [style=morphism] (8) at (-2, 0.5) {$\phi$};
		\node [style=none] (9) at (-2, 2.25) {$\augS$};
	\end{pgfonlayer}
	\begin{pgfonlayer}{edgelayer}
		\draw [bend left=45] (0.center) to (1.center);
		\draw [bend right=45] (0.center) to (2.center);
		\draw (4.center) to (0.center);
		\draw (7.center) to (2.center);
	\end{pgfonlayer}
\end{tikzpicture}
\end{cdiag}
CD categories also admit the following extension of Proposition \ref{prop:inv_rev}.
Notice this requires Radon--Nikodym derivatives (Definition \ref{def:radon-nikodym}), and so is not directly available in the Markov category setting.

\begin{proposition} \label{prop:reversibility-up-to-reweighting}
	Suppose $\mu : \ind \to \augS$ and $\phi : \augS \to \augS$ are morphisms in a CD category $\C$, where $\phi$ is a deterministic involution.
	If a Radon--Nikodym derivative $r = \frac{\dif (\phi \circ \mu)}{\dif \mu}$ exists, then
	\begin{cdiag} \label{eq:reversibility-up-to-reweighting}
		\begin{tikzpicture}
	\begin{pgfonlayer}{nodelayer}
		\node [style=state] (3) at (-4, -2.75) {$\mu$};
		\node [style=morphism] (12) at (-6, 0.75) {$\phi$};
		\node [style=none] (16) at (-2.5, 2.25) {};
		\node [style=likelihood] (21) at (-4, 1.25) {$\rx$};
		\node [style=none] (23) at (-5, -0.75) {};
		\node [style=none] (24) at (-4, -1.75) {};
		\node [style=state] (29) at (-11, -2.25) {$\mu$};
		\node [style=none] (40) at (-8, 0) {$=$};
		\node [style=none] (41) at (-10, 1.75) {};
		\node [style=none] (42) at (-10, 2.25) {$\augS$};
		\node [style=none] (43) at (-2.5, 2.75) {$\augS$};
		\node [style=none] (44) at (-12, 0) {};
		\node [style=none] (45) at (-10, 0) {};
		\node [style=none] (46) at (-11, -1) {};
		\node [style=morphism] (47) at (-10, 0.5) {$\phi$};
		\node [style=bn] (48) at (-11, -1) {};
		\node [style=none] (49) at (-12, 1.75) {};
		\node [style=none] (50) at (-12, 2.25) {$\augS$};
		\node [style=none] (51) at (-2.5, -0.25) {};
		\node [style=none] (52) at (-6, 0.25) {};
		\node [style=none] (53) at (-4, 0.25) {};
		\node [style=bn] (54) at (-5, -0.75) {};
		\node [style=bn] (55) at (-4, -1.75) {};
		\node [style=none] (56) at (-6, 1) {};
		\node [style=none] (57) at (-6, 2.25) {};
		\node [style=none] (58) at (-6, 2.75) {$\augS$};
	\end{pgfonlayer}
	\begin{pgfonlayer}{edgelayer}
		\draw [in=-180, out=-90] (23.center) to (24.center);
		\draw (24.center) to (3);
		\draw (46.center) to (29);
		\draw [bend right=315] (46.center) to (44.center);
		\draw [bend right=45] (46.center) to (45.center);
		\draw (41.center) to (45.center);
		\draw (49.center) to (44.center);
		\draw [bend left=45] (51.center) to (24.center);
		\draw (51.center) to (16.center);
		\draw [bend right=45] (52.center) to (23.center);
		\draw [bend left=45] (53.center) to (23.center);
		\draw (53.center) to (21);
		\draw (52.center) to (12);
		\draw (57.center) to (56.center);
	\end{pgfonlayer}
\end{tikzpicture}
	\end{cdiag}
	In other words, $\phi$ is $\mu$-reversible ``up to a correction factor of $\rx$''.
\end{proposition}

\begin{proof}
	This follows by the same argument as in \eqref{diag:mu_Phi_inv_pf}, but now using the Radon--Nikodym derivative $r$ in the third step.
\end{proof}

We can now give necessary and sufficient conditions for reversibility of the summand \eqref{eq:mh-first-summand} entirely in terms of string diagrams.

\begin{theorem} \label{thm:mh-categorical}
	Suppose $\mu : \ind \to \augS$, $\phi : \augS \to \augS$, and $\alpha, \rx : \augS \to \ind$ are morphisms in a CD category $\C$, where $\phi$ is a deterministic involution and $\rx = \frac{\dif (\phi \circ \mu)}{\dif \mu}$.
	Then \eqref{eq:mh-first-term} is $\mu$-reversible if and only if
	\begin{cdiag} \label{eq:algebraic-balancing-condition}
		\begin{tikzpicture}
	\begin{pgfonlayer}{nodelayer}
		\node [style=bn] (28) at (-11, -1.25) {};
		\node [style=state] (29) at (-11, -2.25) {$\mu$};
		\node [style=likelihood] (34) at (-12, 0.25) {$\alpha$};
		\node [style=none] (40) at (-8, -0.25) {$=$};
		\node [style=none] (41) at (-10, 1.75) {};
		\node [style=none] (42) at (-10, 2.25) {$\augS$};
		\node [style=state] (43) at (-4, -3) {$\mu$};
		\node [style=morphism] (44) at (-6, 0.25) {$\phi$};
		\node [style=none] (45) at (-2.5, 2.5) {};
		\node [style=likelihood] (46) at (-4, 1.75) {$\rx$};
		\node [style=none] (47) at (-5, -1) {};
		\node [style=none] (48) at (-4, -2) {};
		\node [style=none] (49) at (-2.5, 3) {$\augS$};
		\node [style=none] (50) at (-2.5, -0.5) {};
		\node [style=none] (51) at (-6, 0) {};
		\node [style=none] (52) at (-4, 0) {};
		\node [style=bn] (53) at (-5, -1) {};
		\node [style=bn] (54) at (-4, -2) {};
		\node [style=none] (55) at (-6, 0.75) {};
		\node [style=none] (56) at (-6, 1.5) {};
		\node [style=likelihood] (57) at (-6, 1.75) {$\alpha$};
	\end{pgfonlayer}
	\begin{pgfonlayer}{edgelayer}
		\draw (28) to (29);
		\draw [in=-180, out=-90] (34) to (28);
		\draw [in=0, out=-90, looseness=0.75] (41.center) to (28);
		\draw [in=-180, out=-90] (47.center) to (48.center);
		\draw (48.center) to (43);
		\draw [bend left=45] (50.center) to (48.center);
		\draw (50.center) to (45.center);
		\draw [bend right=45] (51.center) to (47.center);
		\draw [bend left=45] (52.center) to (47.center);
		\draw (52.center) to (46);
		\draw (51.center) to (44);
		\draw (56.center) to (55.center);
	\end{pgfonlayer}
\end{tikzpicture}
	\end{cdiag}
\end{theorem}

\begin{proof}
	By basic string diagram manipulations, \eqref{eq:mh-first-term} is $\mu$-reversible if and only $\phi$ is reversible with respect to the left-hand side of \eqref{eq:algebraic-balancing-condition}.
	By Proposition~\ref{prop:inv_rev}, this holds if and only if $\phi$ leaves this morphism invariant, namely:
	\begin{cdiag} \label{eq:mh-proof-1}
		\begin{tikzpicture}
	\begin{pgfonlayer}{nodelayer}
		\node [style=none] (0) at (4.5, 0.5) {};
		\node [style=state] (1) at (3.25, -2.25) {$\mu$};
		\node [style=bn] (2) at (3.25, -1.25) {};
		\node [style=likelihood] (3) at (2, 0) {$\alpha$};
		\node [style=none] (4) at (0, 0) {$=$};
		\node [style=morphism] (5) at (4.5, 0.5) {$\phi$};
		\node [style=none] (6) at (4.5, 2) {};
		\node [style=none] (7) at (4.5, 2.5) {$\augS$};
		\node [style=bn] (15) at (-3, -1.25) {};
		\node [style=state] (16) at (-3, -2.25) {$\mu$};
		\node [style=likelihood] (17) at (-4, 0.25) {$\alpha$};
		\node [style=none] (18) at (-2, 1.75) {};
		\node [style=none] (19) at (-2, 2.25) {$\augS$};
	\end{pgfonlayer}
	\begin{pgfonlayer}{edgelayer}
		\draw [in=180, out=-90] (3) to (2);
		\draw (6.center) to (5);
		\draw (2) to (1);
		\draw [in=-90, out=0] (2) to (5);
		\draw (15) to (16);
		\draw [in=-180, out=-90] (17) to (15);
		\draw [in=0, out=-90, looseness=0.75] (18.center) to (15);
	\end{pgfonlayer}
\end{tikzpicture}
	\end{cdiag}
	By postcomposing both sides of \eqref{eq:reversibility-up-to-reweighting} with $\alpha\otimes \Id_\augS$, Proposition \ref{prop:reversibility-up-to-reweighting} now shows that \eqref{eq:mh-proof-1} is equivalent to \eqref{eq:algebraic-balancing-condition}.
\end{proof}

\section{Semiadditive CD categories} \label{sec:CMon}

In the previous section, we considered the first summand of the Metropolis--Hastings kernel \eqref{eq:mh-kernel}.
However, we did not have a way to reason about the full kernel $\Pimh$. %
One obvious difficulty is that $\Pimh$ is written as a \emph{sum} of two morphisms, which is not something we can express in a general CD category.
It is also unclear how to express other concepts required for Theorem~\ref{thm:andrieu2020}, such as the rejection probability $1 - \alpha(\tarS)$, or the decomposition of $\augS$ into a region $S$ on which the target and its pushforward under $\phi$ are equivalent measures.

To proceed, we therefore consider a natural \emph{enrichment} \cite{kelly1982basic} of CD categories.
Specifically, we consider CD categories whose hom-sets $\C(\tarS, \sndS)$ are equipped with the structure of a \emph{commutative monoid}, which allows us to reason about ``sums'' of morphisms in an abstract way.
Enrichment over commutative monoids arises naturally in a range of categorical settings.
Previous work \cite{gogioso2018categorical, tull2020categorical} has used this enrichment to model summation of processes in classical probabilistic and quantum settings.
Related enrichments have also appeared in the context of effectus theory, which considers enrichment over \emph{partial} commutative monoids \cite{cho2015introduction}, as well as in cartesian differential categories \cite{blute2009cartesian}, which may be characterised as \emph{skew enriched} over commutative monoids \cite{garner2021cartesian}.
In this work, we show that this simple enrichment captures a surprising number of concepts from classical probability theory, allowing us to reason about the full MH kernel and to give a self-contained reversibility proof that recovers Theorem \ref{thm:andrieu2020} in $\sfkern$.
The enrichment also interacts well with string diagrams, enabling intuitive graphical reasoning (see Remark \ref{rem:diagrammatic-reasoning-semiadditive}).

\subsection{Basic definitions}

Recall that a \emph{commutative monoid} is a set $\M$ equipped with a commutative and associative binary operation, and a distinguished unit element.
In other words, a commutative monoid is an abelian group without the requirement of inverses.
Standard examples include $(\N, +, 0)$ and $(\N, \times, 1)$, the natural numbers under addition and multiplication respectively. 
A \emph{homomorphism} between commutative monoids $\M$ and $\M'$ is a function $\varphi : \M \to \M'$ that preserves the binary operation and unit element.
Commutative monoids and homomorphisms together form a category called $\CMon$.

\begin{definition} \label{def:CMon-enrichment}
	A \emph{$\CMon$-enrichment} of a category $\C$ is a choice of commutative monoid structure $(+, 0)$ for each hom-set $\C(\tarS, \sndS)$ that makes composition of morphisms bilinear as follows:
	\begin{equation} \label{eq:bilinearity-of-composition}
		\begin{aligned}
			\P \circ 0 = 0& \qquad\qquad \P \circ (\Q + \R) = \P \circ \Q + \P \circ \R \\
			0 \circ \P = 0& \qquad\qquad (\Q + \R) \circ \P = \Q \circ \P + \R \circ \P.
		\end{aligned}
	\end{equation}
	When $\C$ has a monoidal product $\otimes$, we say that a $\CMon$-enrichment is \emph{monoidal} if the monoidal product is also bilinear:
	\begin{equation} \label{eq:monoidal-enrichment} 
		\begin{aligned}
			\P \otimes 0 = 0 &\qquad\qquad \P \otimes (\Q + \R) = (\P \otimes \Q) + (\P \otimes \R) \\
			0 \otimes \P = 0 &\qquad\qquad (\Q + \R) \otimes \P = (\Q \otimes \P) + (\R \otimes \P).
		\end{aligned}
	\end{equation}
	Here the morphisms in \eqref{eq:bilinearity-of-composition} and \eqref{eq:monoidal-enrichment} range over all choices that make the expressions typecheck.
	Note also that we are using the same symbols $+$ and $0$ to denote the addition operation and zero maps across different hom-sets.
\end{definition}

We are especially interested in $\CMon$-enrichments of CD categories.
This leads to the following definition.

\begin{definition}
	A CD category is \emph{semiadditive}\footnote{Note that our terminology does not imply the existence of \emph{biproducts}, as is sometimes the case in the literature.}
	if it is equipped with a monoidal $\CMon$-enrichment.
\end{definition}

We emphasise that for what follows, $\CMon$-enrichment is the only additional structure we will require beyond that of a vanilla CD category.
In particular, all the definitions and results below make sense without additional machinery beyond this.

\begin{remark} \label{rem:diagrammatic-reasoning-semiadditive}
	The combination of \eqref{eq:bilinearity-of-composition} and \eqref{eq:monoidal-enrichment} has a useful interpretation in terms of string diagrams.
	The conditions involving zero maps mean that zero maps \emph{annihilate} string diagrams: if any morphism of a string diagram is zero, then the entire diagram is also zero, regardless of the other morphisms it contains.
	The remaining additivity conditions mean that if we are given two string diagrams that are identical apart from one specific subdiagram, then their sum can be computed by replacing that subdiagram with the sum of those two subdiagrams.
	For example:
	
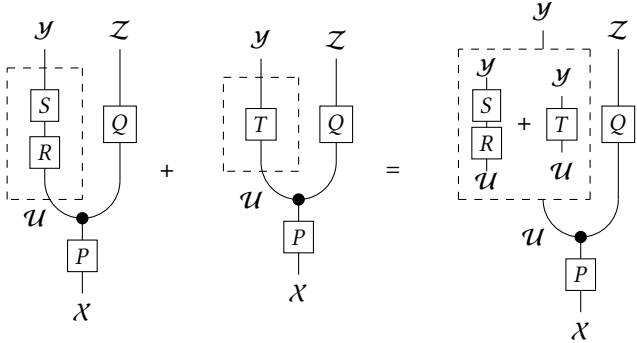
\begin{cdiag*} %
		\begin{tikzpicture}
	\begin{pgfonlayer}{nodelayer}
		\node [style=none] (0) at (-4, -1.25) {};
		\node [style=none] (1) at (-5, -0.25) {};
		\node [style=none] (2) at (-3, -0.25) {};
		\node [style=none] (3) at (-4, -3.25) {};
		\node [style=none] (4) at (-5, 3.25) {};
		\node [style=none] (5) at (-3, 3.25) {};
		\node [style=morphism] (6) at (-5, 0.5) {$\R$};
		\node [style=morphism] (7) at (-4, -2.25) {$\P$};
		\node [style=bn] (8) at (-4, -1.25) {};
		\node [style=morphism] (9) at (-5, 1.75) {$\S$};
		\node [style=morphism] (10) at (-3, 1.25) {$\Q$};
		\node [style=none] (11) at (-5, 3.75) {$\sndS$};
		\node [style=none] (12) at (-3, 3.75) {$\auxS$};
		\node [style=none] (13) at (-4, -3.75) {$\tarS$};
		\node [style=none] (14) at (-1.75, 0) {$+$};
		\node [style=none] (15) at (1.75, -0.75) {};
		\node [style=none] (16) at (0.75, 0.25) {};
		\node [style=none] (17) at (2.75, 0.25) {};
		\node [style=none] (18) at (1.75, -2.75) {};
		\node [style=none] (19) at (0.75, 3) {};
		\node [style=none] (20) at (2.75, 3) {};
		\node [style=morphism] (21) at (0.75, 1.25) {$\T$};
		\node [style=morphism] (22) at (1.75, -1.75) {$\P$};
		\node [style=bn] (23) at (1.75, -0.75) {};
		\node [style=morphism] (25) at (2.75, 1.25) {$\Q$};
		\node [style=none] (26) at (0.75, 3.5) {$\sndS$};
		\node [style=none] (27) at (2.75, 3.5) {$\auxS$};
		\node [style=none] (28) at (1.75, -3.25) {$\tarS$};
		\node [style=none] (29) at (4.25, 0) {$=$};
		\node [style=none] (30) at (9.25, -2) {};
		\node [style=none] (31) at (8.25, -1) {};
		\node [style=none] (32) at (10.25, -1) {};
		\node [style=none] (33) at (9.25, -4) {};
		\node [style=none] (35) at (10.25, 3.75) {};
		\node [style=morphism] (37) at (9.25, -3) {$\P$};
		\node [style=bn] (38) at (9.25, -2) {};
		\node [style=morphism] (40) at (10.25, 1.25) {$\Q$};
		\node [style=none] (41) at (8.25, 4.25) {$\sndS$};
		\node [style=none] (42) at (10.25, 4.25) {$\auxS$};
		\node [style=none] (43) at (9.25, -4.5) {$\tarS$};
		\node [style=none] (44) at (6, 3.25) {};
		\node [style=none] (45) at (6, -1) {};
		\node [style=none] (46) at (9.5, -1) {};
		\node [style=none] (47) at (9.5, 3.25) {};
		\node [style=morphism] (48) at (6.75, 1.75) {$\S$};
		\node [style=morphism] (49) at (6.75, 0.5) {$\R$};
		\node [style=morphism] (50) at (8.75, 1.25) {$\T$};
		\node [style=none] (51) at (7.75, 1.25) {$+$};
		\node [style=none] (52) at (6.75, 2.5) {};
		\node [style=none] (53) at (6.75, -0.25) {};
		\node [style=none] (54) at (8.75, 2) {};
		\node [style=none] (55) at (8.75, 0.25) {};
		\node [style=none] (56) at (8.25, 3.25) {};
		\node [style=none] (57) at (8.25, 3.75) {};
		\node [style=none] (58) at (6.75, -0.5) {$\U$};
		\node [style=none] (59) at (6.75, 2.75) {$\sndS$};
		\node [style=none] (60) at (8.75, -0.25) {$\U$};
		\node [style=none] (61) at (8.75, 2.5) {$\sndS$};
		\node [style=none] (62) at (-5.25, -1.25) {$\U$};
		\node [style=none] (63) at (0.5, -0.75) {$\U$};
		\node [style=none] (64) at (-6, 2.75) {};
		\node [style=none] (65) at (-4, 2.75) {};
		\node [style=none] (66) at (-4, -0.75) {};
		\node [style=none] (67) at (-6, -0.75) {};
		\node [style=none] (68) at (-0.25, 2.5) {};
		\node [style=none] (69) at (-0.25, 0) {};
		\node [style=none] (70) at (1.75, 0) {};
		\node [style=none] (71) at (1.75, 2.5) {};
		\node [style=none] (72) at (8, -2) {$\U$};
	\end{pgfonlayer}
	\begin{pgfonlayer}{edgelayer}
		\draw [bend left=45] (2.center) to (0.center);
		\draw [bend left=45] (0.center) to (1.center);
		\draw (3.center) to (0.center);
		\draw (4.center) to (1.center);
		\draw (5.center) to (2.center);
		\draw [bend left=45] (17.center) to (15.center);
		\draw [bend left=45] (15.center) to (16.center);
		\draw (18.center) to (15.center);
		\draw (19.center) to (16.center);
		\draw (20.center) to (17.center);
		\draw [bend left=45] (32.center) to (30.center);
		\draw [bend left=45] (30.center) to (31.center);
		\draw (33.center) to (30.center);
		\draw (35.center) to (32.center);
		\draw [style=dashed box] (47.center) to (44.center);
		\draw [style=dashed box] (44.center) to (45.center);
		\draw [style=dashed box] (45.center) to (46.center);
		\draw [style=dashed box] (46.center) to (47.center);
		\draw (53.center) to (52.center);
		\draw (55.center) to (54.center);
		\draw (57.center) to (56.center);
		\draw [style=dashed box] (67.center) to (64.center);
		\draw [style=dashed box] (64.center) to (65.center);
		\draw [style=dashed box] (65.center) to (66.center);
		\draw [style=dashed box] (66.center) to (67.center);
		\draw [style=dashed box] (70.center) to (69.center);
		\draw [style=dashed box] (69.center) to (68.center);
		\draw [style=dashed box] (68.center) to (71.center);
		\draw [style=dashed box] (71.center) to (70.center);
	\end{pgfonlayer}
\end{tikzpicture}
	\end{cdiag*}
	where here the dashed boxes indicate the subdiagrams that differ.
	This provides a convenient tool for reasoning about semiadditive CD categories that we will make use of below.
\end{remark}

Categories of kernels (Example \ref{exa:sfkern}) admit natural enrichments over commutative monoids as follows.

\begin{proposition}\label{prop:CMon-sfkern}
	Let $\tarS$ and $\sndS$ be measurable spaces, and define the zero kernel $0 : \tarS \to \sndS$ and the sum $\P + \Q : \tarS \to \sndS$ of kernels $\P, \Q : \tarS \to \sndS$ respectively as follows:
	\begin{align*}
		0(x, \dif y) &\coloneqq 0 \\
		(\P + \Q)(x, \dif y) &\coloneqq \P(x, \dif y) + \Q(x, \dif y).
	\end{align*}
	This constitutes a $\CMon$-enrichment of $\Kern$, and a monoidal $\CMon$-enrichment of $\sfkern$ (so in particular $\sfkern$ is semiadditive).
\end{proposition}

\begin{proof}
	This follows straightforwardly by checking the commutative monoid axioms and the compatibility conditions \eqref{eq:bilinearity-of-composition} and \eqref{eq:monoidal-enrichment}.
	The key point is that kernels take values in $[0, \infty]$, which itself is a commutative monoid under addition.
\end{proof}

\begin{remark} \label{rem:effect-multiplication}
	Let $\C$ be a CD category.
	As observed by \cite[Section 3.1]{fritz2023weakly}, for any object $\tarS$, the effects $C(\tarS, \ind)$ obtain an additional commutative monoid structure $(\ast, 1)$.
	Specifically, the operation $\ast$ is defined for effects $\vx, \wx : \tarS \to \ind$ as
	\begin{cdiag*}
		\begin{tikzpicture}
	\begin{pgfonlayer}{nodelayer}
		\node [style=none] (8) at (2, 1) {};
		\node [style=none] (9) at (4, 1) {};
		\node [style=none] (10) at (3, 0) {};
		\node [style=none] (11) at (3, -1) {};
		\node [style=none] (12) at (3, -1.5) {$\tarS$};
		\node [style=likelihood] (13) at (2, 1) {$\vx$};
		\node [style=likelihood] (14) at (4, 1) {$\wx$};
		\node [style=bn] (15) at (3, 0) {};
		\node [style=none] (16) at (0, 0) {$\coloneqq$};
		\node [style=likelihood] (17) at (-6.75, 0.75) {$\vx$};
		\node [style=likelihood] (18) at (-3.25, 0.75) {$\wx$};
		\node [style=none] (19) at (-5, 0) {$\ast$};
		\node [style=none] (20) at (-6.75, -0.75) {};
		\node [style=none] (21) at (-3.25, -0.75) {};
		\node [style=none] (22) at (-6.75, -1.25) {$\tarS$};
		\node [style=none] (23) at (-3.25, -1.25) {$\tarS$};
	\end{pgfonlayer}
	\begin{pgfonlayer}{edgelayer}
		\draw [bend right=315] (10.center) to (8.center);
		\draw [bend right=45] (10.center) to (9.center);
		\draw (11.center) to (10.center);
		\draw (17) to (20.center);
		\draw (21.center) to (18);
	\end{pgfonlayer}
\end{tikzpicture}
	\end{cdiag*}
	and the unit is given by the deletion map, which will write suggestively in what follows as
	\begin{equation}
		1 \coloneqq \del_\tarS : \tarS \to \ind.
		\label{eq:1del}
	\end{equation}
	(Note that, as for $+$ and $0$ above, we do not denote the dependence of $\ast$ and $1$ on $\tarS$ explicitly.)
	In $\sfkern$, the effect $\vx \ast \wx$ corresponds to the pointwise multiplication of the nonnegative functions associated to $\vx$ and $\wx$ (Example \ref{exa:effects-sfkern}).

	When $\C$ is semiadditive, effects $\C(\tarS, \ind)$ become a \emph{commutative semiring} (i.e.\ a ring without additive inverses): their ``addition'' structure is obtained from the $\CMon$-enrichment $(+, 0)$, and their ``multiplication'' structure is obtained from $(\ast, 1)$ just defined.
	The semiring axioms can be shown diagrammatically using Remark~\ref{rem:diagrammatic-reasoning-semiadditive}: for example, we obtain $0 \ast \wx = 0$ because zero maps annihilate general morphisms, and the distributive law $\wx \ast (\vx + \ux) = (\wx \ast \vx) + (\wx \ast \ux)$ holds because
	\begin{cdiag*}
		\begin{tikzpicture}
	\begin{pgfonlayer}{nodelayer}
		\node [style=none] (0) at (-4.25, 0.25) {};
		\node [style=none] (1) at (-1.25, 0.25) {};
		\node [style=none] (2) at (-2.25, -0.75) {};
		\node [style=none] (3) at (-2.25, -1.75) {};
		\node [style=none] (4) at (3.75, 1) {};
		\node [style=none] (5) at (5.75, 1) {};
		\node [style=none] (6) at (4.75, 0) {};
		\node [style=none] (7) at (4.75, -1) {};
		\node [style=none] (8) at (7, 0) {$+$};
		\node [style=none] (9) at (8.25, 1) {};
		\node [style=none] (10) at (10.25, 1) {};
		\node [style=none] (11) at (9.25, 0) {};
		\node [style=none] (12) at (9.25, -1) {};
		\node [style=none] (13) at (2.25, 0) {$=$};
		\node [style=likelihood] (14) at (-4.25, 0.5) {$\wx$};
		\node [style=none] (16) at (-3.25, 0.25) {};
		\node [style=none] (17) at (0.75, 0.25) {};
		\node [style=none] (19) at (-2.25, 1.5) {};
		\node [style=none] (20) at (-2.25, 0.5) {};
		\node [style=none] (21) at (-0.25, 1.5) {};
		\node [style=none] (22) at (-0.25, 0.5) {};
		\node [style=none] (23) at (-2.25, 0.5) {};
		\node [style=likelihood] (24) at (-2.25, 1.75) {$\vx$};
		\node [style=likelihood] (25) at (-0.25, 1.75) {$\ux$};
		\node [style=none] (26) at (-1.25, 1) {$+$};
		\node [style=none] (27) at (-3.25, 3) {};
		\node [style=none] (28) at (0.75, 3) {};
		\node [style=bn] (29) at (-2.25, -0.75) {};
		\node [style=likelihood] (30) at (3.75, 1.25) {$\wx$};
		\node [style=likelihood] (31) at (5.75, 1.25) {$\vx$};
		\node [style=likelihood] (32) at (10.25, 1.25) {$\ux$};
		\node [style=likelihood] (33) at (8.25, 1.25) {$\wx$};
		\node [style=bn] (34) at (9.25, 0) {};
		\node [style=bn] (35) at (4.75, 0) {};
		\node [style=none] (36) at (-2.25, -2.25) {$\tarS$};
		\node [style=none] (37) at (4.75, -1.5) {$\tarS$};
		\node [style=none] (38) at (9.25, -1.5) {$\tarS$};
	\end{pgfonlayer}
	\begin{pgfonlayer}{edgelayer}
		\draw [in=-90, out=-180, looseness=1.25] (2.center) to (0.center);
		\draw [bend right=45] (2.center) to (1.center);
		\draw (2.center) to (3.center);
		\draw [bend right=315] (6.center) to (4.center);
		\draw [bend right=45] (6.center) to (5.center);
		\draw (6.center) to (7.center);
		\draw [bend right=315] (11.center) to (9.center);
		\draw [bend right=45] (11.center) to (10.center);
		\draw (11.center) to (12.center);
		\draw [style=dashed box] (17.center) to (16.center);
		\draw (23.center) to (19.center);
		\draw (22.center) to (21.center);
		\draw [style=dashed box] (27.center) to (16.center);
		\draw [style=dashed box] (28.center) to (17.center);
		\draw [style=dashed box] (28.center) to (27.center);
	\end{pgfonlayer}
\end{tikzpicture}
	\end{cdiag*}
\end{remark}

\begin{remark} \label{rem:reweighting-action}
	Also as noted by \cite[Section 3.1]{fritz2023weakly}, in a CD category $\C$, the monoid of effects $\C(\tarS, \ind)$ \emph{acts} on the hom-sets $\C(\tarS, \sndS)$. %
	Specifically, given an effect $\wx : \tarS \to \ind$ and a morphism $\P : \tarS \to \sndS$, we obtain a new morphism
	\begin{cdiag} \label{diag:reweighting-action}
		\begin{tikzpicture}
	\begin{pgfonlayer}{nodelayer}
		\node [style=morphism] (0) at (-3, 0) {$\wx \lact \P$};
		\node [style=none] (1) at (0, 0) {$\coloneqq$};
		\node [style=none] (2) at (-3, 0.25) {};
		\node [style=none] (3) at (-3, 1.5) {};
		\node [style=none] (4) at (-3, -0.25) {};
		\node [style=none] (5) at (-3, -1.5) {};
		\node [style=none] (6) at (-3, 2) {$\sndS$};
		\node [style=none] (7) at (-3, -2) {$\tarS$};
		\node [style=none] (8) at (2.5, 0.25) {};
		\node [style=none] (9) at (3.5, -0.75) {};
		\node [style=none] (10) at (4.5, 0.25) {};
		\node [style=none] (11) at (3.5, -1.75) {};
		\node [style=bn] (12) at (3.5, -0.75) {};
		\node [style=likelihood] (13) at (2.5, 0.5) {$\wx$};
		\node [style=none] (14) at (4.5, 1.75) {};
		\node [style=none] (15) at (4.5, 2.25) {$\sndS$};
		\node [style=none] (16) at (3.5, -2.25) {$\tarS$};
		\node [style=morphism] (18) at (4.5, 0.75) {$\P$};
	\end{pgfonlayer}
	\begin{pgfonlayer}{edgelayer}
		\draw (5.center) to (4.center);
		\draw (2.center) to (3.center);
		\draw (11.center) to (9.center);
		\draw [bend left=45] (9.center) to (8.center);
		\draw [bend right=45] (9.center) to (10.center);
		\draw (14.center) to (10.center);
	\end{pgfonlayer}
\end{tikzpicture}
	\end{cdiag}
	For example, the first MH term \eqref{eq:mh-first-term} may be expressed in terms of this action as $\alpha \cdot \phi$.
	In the case of $\sfkern$, it is straightforward to check that the morphism $\wx \cdot \K$ is kernel given by
	\[
		(\wx \cdot \P)(x, \dif y) = \wx(x) \, \P(x, \dif y),
	\]
	and so we think of $\wx \cdot \P$ as being given by $\P$ \emph{reweighted} according to $\wx$ on the input side.
\end{remark}

\begin{remark} \label{rem:preorder-enrichment}
	It is a standard fact that a commutative monoid $\M$ admits a canonical \emph{preorder} (i.e.\ a reflexive and transitive relation) $\leq$ defined by $a \leq b$ if there exists $c \in \M$ such that $a + c = b$.
	This extends to the morphisms of every $\CMon$-enriched category $\C$ accordingly.
	It is also easy to check that this preorder is preserved by pre- and post-composition: if $\P \leq \Q$, then
    $
		\R \circ \P \circ \S \leq \R \circ \Q \circ \S
	$
	for all compatibly typed morphisms $\R$ and $\S$.
	From the perspective of enriched category theory, this means that every $\CMon$-enriched category is also enriched over the category of preorders.

	When $\C$ is monoidally $\CMon$-enriched, this preorder is also compatible with the monoidal product: if $\P \leq \Q$, then some basic manipulations show that
	$
		\R \otimes \P \otimes \S \leq \R \otimes \Q \otimes \S
	$
    holds for all morphisms $\R$ and $\S$.
	It follows that the induced enrichment over preorders is also monoidal in this case.
\end{remark}

\begin{remark} \label{rem:extra-properties-semiadditive}
    For a small number of results below, we will consider $\CMon$-enrichments that satisfy certain additional compatibility conditions.
    For completeness, we give a list here (which can be skipped on a first reading).
    We will say that:
	\begin{itemize}
		\item $\C$ is \emph{zero-sum-free} if $\P + \Q = 0$ implies $\P = \Q = 0$;
		\item $\C$ is \emph{zero-monic} if $1 \circ \P = 0$ implies $\P = 0$ (where recall from \eqref{eq:1del} that $1 = \del$);
		\item $\C$ has \emph{no zero divisors} if $\P \otimes \Q = 0$ implies $\P = 0$ or $\Q = 0$.
	\end{itemize}
	Note that the last two conditions need additional structure (e.g.\ a CD category) to make sense.
	The zero-monic condition was previously introduced in the context of \emph{effectus theory} \cite[Definition 31]{cho2015introduction}.

	It is straightforward to check that being zero-sum-free is equivalent to $0$ being a bottom element with respect to the preorder from Remark \ref{rem:preorder-enrichment}.
	In other words, $\P \leq 0$ implies $\P = 0$.
	Some basic manipulations show that $\sfkern$ satisfies all three of these requirements, and $\Kern$ satisfies the first two.
\end{remark}

\subsection{Substochasticity, probabilities, and convex combinations} \label{sec:probabilities-and-convex-combinations}

Semiadditive CD categories admit the following abstract definition of \emph{substochastic kernels} and \emph{probabilities}. 
Moreover, using the formalism from Remarks \ref{rem:effect-multiplication} and \ref{rem:preorder-enrichment}, this mirrors notationally the usual definitions from classical probability theory.

\begin{definition}
	Let $\P: \tarS \to \sndS$ be a morphism in a semiadditive CD category $\C$.
	Then $\P$ is \emph{substochastic} if it holds that
	\begin{equation} \label{eq:substochastic-morphism}
		1 \circ \P \leq 1,
	\end{equation}
	where recall that $1 = \del_\sndS$ denotes the multiplicative unit in the commutative monoid of effects (Remark \ref{rem:effect-multiplication}).
	A \emph{probability} is defined as a substochastic effect.
\end{definition}

\begin{remark} \label{rem:substochastic-interpretation}
	Since $\Id_\ind = \del_\ind = 1$ in a CD category (Definition \ref{def:cd_cat_full}), a probability is equivalently an effect $\alpha : \tarS \to \ind$  with
	$
		\alpha \leq 1.
	$
	By Remark \ref{rem:preorder-enrichment}, this in turn means equivalently that there exists some effect $\comp{\alpha} : \tarS \to \ind$ such that
	\[
		\alpha + \comp{\alpha} = 1.
	\]
	Note that since $+$ is commutative, $\comp{\alpha}$ is also a probability.
\end{remark}

\begin{example} \label{exa:substoch_sfkern}
	Recall from Example~\ref{exa:effects-sfkern} that effects in $\sfkern$ correspond to nonnegative measurable functions.
	Given $\P : \tarS \to \sndS$ in $\sfkern$, it is straightforward to check that the effect $\one \circ \P$ corresponds to the measurable function $x \mapsto \P(x, \sndS)$, and that $1 \circ \P \leq 1$ if and only if
	$
		\P(x, \sndS) \le 1
	$
	holds for all $x\in \tarS$.
	In this way we recover the usual notion of a \emph{substochastic kernel}.
	Likewise, a probability $\alpha : \tarS \to \ind$ in $\sfkern$ corresponds to a measurable function $\tarS \to [0,1]$, as the name suggests.
\end{example}

\begin{remark} \label{rem:convex-combinations}
	Using probabilities, we can express \emph{convex combinations} of morphisms in a semiadditive CD category.
	Specifically, given a probability $\alpha : \tarS \to \ind$, we may regard
	the morphism $\alpha \lact \P + \comp{\alpha} \lact \Q$ as an abstract convex combination of $\P, \Q : \tarS \to \sndS$, where the action $\cdot$ is defined in \eqref{diag:reweighting-action}. %
	In string diagrams, this may be written as
	\begin{cdiag} \label{eq:convex-sum-of-normalised-morphisms}
		\begin{tikzpicture}
	\begin{pgfonlayer}{nodelayer}
		\node [style=none] (0) at (0, 0) {$+$};
		\node [style=none] (1) at (-3, -1) {};
		\node [style=none] (2) at (-4, 0) {};
		\node [style=none] (3) at (-2, 0) {};
		\node [style=none] (4) at (-3, -2) {};
		\node [style=bn] (5) at (-3, -1) {};
		\node [style=likelihood] (6) at (-4, 0.25) {$\alpha$};
		\node [style=none] (8) at (-2, 1.5) {};
		\node [style=none] (9) at (3.5, -1) {};
		\node [style=none] (10) at (2.5, 0) {};
		\node [style=none] (11) at (4.5, 0) {};
		\node [style=none] (12) at (3.5, -2) {};
		\node [style=bn] (13) at (3.5, -1) {};
		\node [style=likelihood] (14) at (2.5, 0.25) {$\comp{\alpha}$};
		\node [style=none] (16) at (4.5, 1.5) {};
		\node [style=none] (17) at (-2, 2) {$\sndS$};
		\node [style=none] (18) at (-3, -2.5) {$\tarS$};
		\node [style=none] (19) at (3.5, -2.5) {$\tarS$};
		\node [style=none] (20) at (4.5, 2) {$\sndS$};
		\node [style=morphism] (21) at (-2, 0.5) {$\P$};
		\node [style=morphism] (22) at (4.5, 0.5) {$\Q$};
	\end{pgfonlayer}
	\begin{pgfonlayer}{edgelayer}
		\draw [bend left=45] (3.center) to (1.center);
		\draw [bend right=315] (1.center) to (2.center);
		\draw (4.center) to (1.center);
		\draw (8.center) to (3.center);
		\draw [bend left=45] (11.center) to (9.center);
		\draw [bend right=315] (9.center) to (10.center);
		\draw (12.center) to (9.center);
		\draw (16.center) to (11.center);
	\end{pgfonlayer}
\end{tikzpicture}
	\end{cdiag}
	Intuitively, here $\P$ is weighted by $\alpha$, while $\Q$ is weighted by $\comp{\alpha}$.
	In $\sfkern$, by Remark \ref{rem:reweighting-action}, this corresponds to the kernel
	\begin{equation}\label{eq:cvx_sfkern}
		\alpha(x) \, \P(x, \dif y) + \left (1 - \alpha(x)\right ) \, \Q(x, \dif y),
	\end{equation}
	which recovers the usual notion of convex combinations of kernels.
\end{remark}

Before proceeding we note the following structural result that will be useful for us below.

\begin{proposition} \label{prop:substochastic-CD-category}
	Let $\C$ be a semiadditive CD category.
	Then the objects and substochastic morphisms of $\C$ form a CD category, with its CD category structure inherited from $\C$.
\end{proposition}

\begin{proof}
	We first show substochasticity is preserved by composition.
	Let $\P : \tarS \to \sndS$ and $\Q : \sndS \to \auxS$ be substochastic.
	Hence $\one \circ \Q \leq \one$, and so precomposing both sides by $\P$ gives
	\[
		\one \circ \Q \circ \P \leq \one \circ \P \leq \one
	\]
	by Remark \ref{rem:preorder-enrichment}, where we use substochasticity of $\P$ in the second step.
	A similar argument (also using Remark \ref{rem:preorder-enrichment}) shows that substochastic morphisms are closed under monoidal products also.

	Next, observe that every normalised (Remark \ref{rem:normalised-morphisms}) morphism $\R : \tarS \to \sndS$ (see \cite[Definition 2.3]{Cho2019}) in $\C$ by definition satisfies
	$
		1 \circ \R = 1
	$
	(since recall we defined $\one = \del$ in Remark \ref{rem:effect-multiplication}).
	Since $\leq$ is reflexive, every normalised morphism $\R$ is therefore substochastic.
	It is also standard to show that the structure morphisms of a CD category are all normalised, i.e.,\ identities, copy and delete maps, and the symmetric monoidal structure maps \cite[Section 7]{Cho2019}.
	In this way, the required CD category structure is inherited from $\C$.
\end{proof}

\subsubsection{Aside: connection with distributive Markov categories}

Statisticians are trained automatically to associate a convex combination like \eqref{eq:cvx_sfkern} with a sampling procedure as follows:
\begin{align} \label{eq:mh-sampling-procedure}
	B \sim \bernoulli(\alpha(x)) \qquad
	Y \sim \P(x, \dif y) \qquad
	Y' \sim \Q(x, \dif y) \qquad
	\text{return $Y$ if $B = 1$ else $Y'$}
\end{align}
Here $\bernoulli(a)$ denotes a Bernoulli distribution (i.e.\ a coin toss) with success probability $a$.
However, while this makes sense in $\sfkern$, it is not clear how to interpret \eqref{eq:mh-sampling-procedure} semantically inside a general semiadditive CD category, since in general these do not supply morphisms corresponding to $\bernoulli$ or the ``if-else'' term here.
On the other hand, \cite[Definition 2]{ackerman2024probabilistic} recently introduced \emph{distributive Markov categories}, which are Markov categories equipped with coproducts that distribute over the monoidal product (see Appendix \ref{sec:appendix-substochasticity}), and which do provide such morphisms in a canonical way.
However, since effects in Markov categories are always trivial, distributive Markov categories in general do not supply morphisms corresponding to individual probability \emph{values} such as $\alpha : \tarS \to \ind$.
In turn, this means they cannot be used to interpret convex combinations like \eqref{eq:cvx_sfkern} directly.

For our goal of proving Theorem \ref{thm:andrieu2020}, it is sufficient to work entirely in semiadditive CD categories, as we do below.
However, for general probability theory, it seems useful to have a setting that supports both these perspectives.
This is provided by the following result, which we prove in Appendix \ref{sec:appendix-substochasticity}.

\begin{proposition} \label{prop:distributive}
	Let $\C$ be a semiadditive CD category with an initial object and binary coproducts whose inclusions are normalised and deterministic.
	Then these coproducts equip $\Cnorm$ (Remark \ref{rem:normalised-morphisms}) with the structure of a distributive Markov category.
\end{proposition}

\subsection{Cancellative and finite morphisms} \label{sec:cancellative}

In this subsection we build up to the definition of a \emph{finitely cancellative} CD category (Definition \ref{def:finitely_c}), which we require for the converse part of Theorem \ref{thm:enrich_rev} below.
To begin, it is often useful to know that a commutative monoid $\M$ is \emph{cancellative}.
This property says that if $\P + \Q = \P + \R$ holds in $\M$, then also $\Q = \R$.
Unfortunately, due to the possibility of obtaining infinite mass, kernels are not cancellative in general.
For example, if $\P(x, A) = \infty$ for all $x$ and $A$, then $\P + \Q = \P$ for any $\Q$.
This motivates the following elementwise definition instead:

\begin{definition}\label{def:cancellative}
	Let $\P : \tarS \to \sndS$ be a morphism in a category $\C$ equipped with a $\CMon$-enrichment.
	We say that $\P$ is \emph{cancellative} if for all $\Q, \R : \tarS \to \sndS$ such that $\P + \Q = \P + \R$, it holds that $\Q = \R$.
\end{definition}

Note that we define this purely in terms of a $\CMon$-enrichment (not a CD structure).
It therefore applies to the category $\Kern$ (Example~\ref{exa:sfkern}) as well as to $\sfkern$.

The following basic result holds for cancellative morphisms in any $\CMon$-enriched category.

\begin{proposition} \label{prop:cancellative-inequality}
	Let $\P, \Q : \tarS \to \sndS$ be morphisms in a $\CMon$-enriched category $\C$.
	If $\P \leq \Q$ and $\Q$ is cancellative, then so is $\P$.
\end{proposition}

\begin{proof}
	By assumption, there is some $\R : \tarS \to \sndS$ such that $\P + \R = \Q$.
	If we have $\P + \S = \P + \T$ for some $\S, \T : \tarS \to \sndS$, then
	\[
		\Q + \S = (\P + \R) + \S = (\P + \R) + \T = \Q + \T,
	\]
	and so $\S = \T$ by cancellativity of $\Q$.
\end{proof}

For kernels, cancellativity recovers the notion of \emph{$\sigma$-finiteness}.
Recall that a measure $\mu$ on a measurable space $\tarS$ is $\sigma$-finite if there exists disjoint measurable $A_1,A_2,\ldots \subseteq \tarS$ such that $\tarS = \bigcup_{n=1}^\infty A_n$ and each $\mu(A_n) < \infty$.
The following result is proven in Appendix~\ref{sec:appendix-sigma-finiteness}.

\begin{proposition} \label{prop:sigma-finiteness-cancellative}
	Let $\K : \tarS \to \sndS$ be a kernel (i.e.\ a morphism in $\Kern$).
	If $\K$ is pointwise $\sigma$-finite (i.e.\ each $\K(x, \dif y)$ is $\sigma$-finite), then $\K$ is cancellative.
	The converse holds if all singleton subsets of $\tarS$ are measurable.
\end{proposition}

The following basic result would be tedious to show in terms of $\sigma$-finiteness, but is straightforward in terms of cancellativity.

\begin{proposition} \label{prop:cancellative-isomorphisms}
	Cancellative morphisms in a $\CMon$-enriched category are preserved under composition with isomorphisms.
\end{proposition}

\begin{proof}
	Let $\P : \tarS \to \sndS$ be cancellative, and let $\varphi : \W \to \tarS$ and $\psi : \sndS \to \auxS$ be isomorphisms.
	Suppose $\Q, \R : \W \to \auxS$ satisfy
	\[
		\psi \circ \P \circ \varphi + \Q = \psi \circ \P \circ \varphi + \R.
	\]
	Bilinearity of composition (Definition~\ref{def:CMon-enrichment}) implies that
	\[
		\P + \psi^{-1} \circ \Q \circ \varphi^{-1} = \P + \psi^{-1} \circ \R \circ \varphi^{-1},
	\]
	and hence $\psi^{-1} \circ \Q \circ \varphi^{-1} = \psi^{-1} \circ \R \circ \varphi^{-1}$ by cancellativity of $\P$.
	Applying  $\psi$ and $\varphi$ on either side then gives $\Q = \R$.
\end{proof}

A closely related notion to cancellativity is \emph{finiteness}. 
Our terminology here is motivated by the case of kernels (see Proposition \ref{prop:sfkern-finite} below, proven in Appendix \ref{sec:appendix-sfkern-finite}).

\begin{definition}
	A morphism $\P : \tarS \to \sndS$ in a semiadditive CD category $\C$ is \emph{finite} if the effect $\one \circ \P$ is cancellative.
\end{definition}

\begin{proposition} \label{prop:sfkern-finite}
	Let $P : \tarS \to \sndS$ be a kernel (i.e.\ a morphism in $\Kern$).
	Then $P$ is finite if and only if $P(x, \sndS) < \infty$ for all $x \in \tarS$.
\end{proposition}

In $\Kern$ (and hence $\sfkern$), it is clear that finiteness implies cancellativity (i.e.\ $\sigma$-finiteness) by Proposition \ref{prop:sigma-finiteness-cancellative}.
It is therefore natural to ask whether this relationship holds generally.
We conjecture that it does not always hold, and so introduce the following definition to capture the cases where it does.

\begin{definition}\label{def:finitely_c}
	A semiadditive CD category $\C$ is \emph{finitely cancellative} if all finite morphisms in $\C$ are cancellative.
\end{definition}

Finiteness interacts with substochasticity in a natural way, as the following result shows.

\begin{proposition} \label{prop:finite-substochastic-composition}
	Let $\P : \tarS \to \sndS$ and $\Q : \sndS \to \auxS$ be morphisms in a semiadditive CD category $\C$.
	If $\P$ is finite and $\Q$ is substochastic, then $\Q \circ \P$ is finite.
\end{proposition}

\begin{proof}
	By substochasticity of $\Q$ and Remark \ref{rem:preorder-enrichment}, we have $\one \circ \Q \circ \P \leq \one \circ \P$.
	Since $\P$ is finite, the right-hand side is cancellative.
	By Proposition~\ref{prop:cancellative-inequality}, the left-hand side is also cancellative, so $\Q \circ \P$ is finite.
\end{proof}

\begin{remark} \label{rem:almost-sure-equality}
	Let $\mu : \ind \to \tarS$ and $\P, \Q : \tarS \to \sndS$ be morphisms in a CD category $\C$.
	Following \cite[Definition 5.1]{Cho2019}, we will say that \emph{$\P = \Q$ $\mu$-almost everywhere} if
	\begin{cdiag} \label{eq:almost_sure}
		\begin{tikzpicture}
	\begin{pgfonlayer}{nodelayer}
		\node [style=bn] (0) at (-3.25, -0.5) {};
		\node [style=none] (1) at (-3.25, -1.5) {};
		\node [style=none] (2) at (-4.25, 0.5) {};
		\node [style=none] (3) at (-2.25, 0.5) {};
		\node [style=state] (4) at (-3.25, -1.5) {$\mu$};
		\node [style=morphism] (5) at (-2.25, 0.5) {$\P$};
		\node [style=none] (6) at (-2.25, 1.5) {};
		\node [style=none] (7) at (-4.25, 1.5) {};
		\node [style=none] (9) at (0, 0) {$=$};
		\node [style=bn] (10) at (3, -0.5) {};
		\node [style=none] (11) at (3, -1.5) {};
		\node [style=none] (12) at (2, 0.5) {};
		\node [style=none] (13) at (4, 0.5) {};
		\node [style=state] (14) at (3, -1.5) {$\mu$};
		\node [style=morphism] (15) at (4, 0.5) {$\Q$};
		\node [style=none] (16) at (4, 1.5) {};
		\node [style=none] (17) at (2, 1.5) {};
		\node [style=none] (18) at (-4.25, 2) {$\tarS$};
		\node [style=none] (19) at (-2.25, 2) {$\sndS$};
		\node [style=none] (20) at (2, 2) {$\tarS$};
		\node [style=none] (21) at (4, 2) {$\sndS$};
	\end{pgfonlayer}
	\begin{pgfonlayer}{edgelayer}
		\draw (0) to (1.center);
		\draw [in=-180, out=-90] (2.center) to (0);
		\draw [in=270, out=0] (0) to (3.center);
		\draw (6.center) to (5);
		\draw (7.center) to (2.center);
		\draw (10) to (11.center);
		\draw [in=-180, out=-90] (12.center) to (10);
		\draw [in=270, out=0] (10) to (13.center);
		\draw (16.center) to (15);
		\draw (17.center) to (12.center);
	\end{pgfonlayer}
\end{tikzpicture}
	\end{cdiag}
	Notice that this equation makes sense in any CD category.
	However, \cite{Cho2019} consider this specifically in the context of Markov categories (or \emph{affine CD categories} in their terminology).
	One of their reasons for this is that, in $\sfkern$, the condition \eqref{eq:almost_sure} does not necessarily imply (for example) that
	\begin{equation} \label{eq:almost-sure-instantiation}
		\P(x, A) = \Q(x, A) \quad \text{for $\mu$-almost all $x \in \tarS$,}
	\end{equation}
	since pathologies can arise if $\mu$ is infinite.
	However, \eqref{eq:almost-sure-instantiation} \emph{does} follow from \eqref{eq:almost_sure} if $\mu$ is $\sigma$-finite.
	Accordingly, we will allow ourselves to consider the equation in full generality, knowing we can recover the expected meaning in $\sfkern$ by restricting to a cancellative $\mu$.
\end{remark}

\subsection{The full MH kernel} \label{sec:full-mh-kernel}

Semiadditive CD categories allow us to reason about the full MH kernel $\Pimh$ from \eqref{eq:mh-kernel}, rather than just its first term.
For this, simply observe that $\Pimh$ is a convex combination of two kernels: one that applies the involution $\phi$, and the identity kernel that leaves the input unchanged.
Accordingly, suppose $\phi : \augS \to \augS$ is a deterministic involution and $\alpha : \augS \to \ind$ is a probability in a semiadditive CD category $\C$.
The following morphism in $\C$ then constitutes an abstract version of $\Pimh$:
\[
	\Pmh \coloneqq \alpha \lact \phi + \comp{\alpha} \lact \Id_{\augS}.
\]
In string diagrams, this is as follows:
\begin{cdiag} \label{eq:mh-kernel-abstract}
	\begin{tikzpicture}
	\begin{pgfonlayer}{nodelayer}
		\node [style=none] (0) at (-3.25, 0) {$+$};
		\node [style=none] (1) at (-5.5, -1) {};
		\node [style=none] (2) at (-6.5, 0) {};
		\node [style=none] (3) at (-4.5, 0) {};
		\node [style=none] (4) at (-5.5, -2) {};
		\node [style=none] (5) at (-6.5, 0) {};
		\node [style=likelihood] (7) at (-6.5, 0.25) {$\alpha$};
		\node [style=bn] (8) at (-5.5, -1) {};
		\node [style=none] (9) at (-4.5, 2) {};
		\node [style=morphism] (10) at (-4.5, 0.5) {$\phi$};
		\node [style=none] (11) at (-4.5, 2.5) {$\augS$};
		\node [style=none] (12) at (-5.5, -2.5) {$\augS$};
		\node [style=none] (13) at (-0.5, -1) {};
		\node [style=none] (14) at (-1.5, 0) {};
		\node [style=none] (15) at (0.5, 0) {};
		\node [style=none] (16) at (-0.5, -2) {};
		\node [style=none] (17) at (-1.5, 0) {};
		\node [style=likelihood] (19) at (-1.5, 0.25) {$\comp{\alpha}$};
		\node [style=bn] (20) at (-0.5, -1) {};
		\node [style=none] (21) at (0.5, 2) {};
		\node [style=none] (23) at (0.5, 2.5) {$\augS$};
		\node [style=none] (24) at (-0.5, -2.5) {$\augS$};
		\node [style=none] (25) at (-8.75, 0) {$\coloneqq$};
		\node [style=morphism] (26) at (-11, 0) {$\Pmh$};
		\node [style=none] (27) at (-11, -0.25) {};
		\node [style=none] (28) at (-11, -1.5) {};
		\node [style=none] (29) at (-11, 0.25) {};
		\node [style=none] (30) at (-11, 1.5) {};
		\node [style=none] (31) at (-11, 2) {$\augS$};
		\node [style=none] (32) at (-11, -2) {$\augS$};
	\end{pgfonlayer}
	\begin{pgfonlayer}{edgelayer}
		\draw (4.center) to (1.center);
		\draw [bend left=45] (1.center) to (2.center);
		\draw [bend right=45] (1.center) to (3.center);
		\draw (5.center) to (2.center);
		\draw (9.center) to (3.center);
		\draw (16.center) to (13.center);
		\draw [bend left=45] (13.center) to (14.center);
		\draw [bend right=45] (13.center) to (15.center);
		\draw (17.center) to (14.center);
		\draw (21.center) to (15.center);
		\draw (30.center) to (29.center);
		\draw (28.center) to (27.center);
	\end{pgfonlayer}
\end{tikzpicture}
\end{cdiag}

\begin{example} \label{exa:MH-kernel-sfkern}
	Take $\C \coloneqq \sfkern$, and suppose $\phi$ is obtained by lifting an involutive measurable function to a deterministic kernel $\delta_{\phi(\xi)}(\dif \xi')$ as in \eqref{eq:deterministic_morphism}, where $\delta$ is the Dirac delta.
	Likewise, the identity $\Id_\augS$ in $\sfkern$ becomes the kernel $\delta_{\xi}(\dif \xi')$.
	By Remark \ref{rem:convex-combinations}, we may therefore write
	\[
		\Pmh(\xi, \dif \xi') = \alpha(\xi) \, \delta_{\phi(\xi)}(\dif \xi') + (1 - \alpha(\xi)) \, \delta_{\xi}(\dif \xi'),
	\]
	which recovers the original MH kernel \eqref{eq:mh-kernel} exactly.
\end{example}

We want to give conditions for our abstract MH kernel $\Pmh$ to be reversible.
For this we require the following two results, which relate reversibility of a sum to reversibility of its summands.

\begin{proposition} \label{prop:sum-of-reversible-is-reversible}
	Let $\mu : \ind \to \tarS$ and $\P, \Q : \tarS \to \tarS$ be morphisms in a semiadditive CD category $\C$.
	If $\P$ and $\Q$ and are $\mu$-reversible, then so is $\P + \Q$.
\end{proposition}

\begin{proof}
	If $P$ and $Q$ are $\mu$-reversible, then
	\begin{cdiag} \label{eq:reversibility-of-sum}
		\begin{tikzpicture}
	\begin{pgfonlayer}{nodelayer}
		\node [style=none] (0) at (-0.75, 0) {$=$};
		\node [style=state] (1) at (-3, -2) {$\mu$};
		\node [style=none] (2) at (-3, -1.75) {};
		\node [style=none] (3) at (-3, -1) {};
		\node [style=none] (9) at (-4, 0) {};
		\node [style=none] (10) at (-2, 0) {};
		\node [style=bn] (11) at (-3, -1) {};
		\node [style=morphism] (12) at (-4, 0.5) {$\P+\Q$};
		\node [style=none] (13) at (-4, 0.75) {};
		\node [style=none] (14) at (-4, 1.5) {};
		\node [style=none] (15) at (-2, 1.5) {};
		\node [style=none] (16) at (-4, 2) {$\tarS$};
		\node [style=none] (17) at (-2, 2) {$\tarS$};
		\node [style=state] (18) at (1.75, -2) {$\mu$};
		\node [style=none] (19) at (1.75, -1.75) {};
		\node [style=none] (20) at (1.75, -1) {};
		\node [style=none] (21) at (0.75, 0) {};
		\node [style=none] (22) at (2.75, 0) {};
		\node [style=bn] (23) at (1.75, -1) {};
		\node [style=morphism] (24) at (0.75, 0.5) {$\P$};
		\node [style=none] (25) at (0.75, 0.75) {};
		\node [style=none] (26) at (0.75, 1.5) {};
		\node [style=none] (27) at (2.75, 1.5) {};
		\node [style=none] (28) at (0.75, 2) {$\tarS$};
		\node [style=none] (29) at (2.75, 2) {$\tarS$};
		\node [style=state] (30) at (6.25, -2) {$\mu$};
		\node [style=none] (31) at (6.25, -1.75) {};
		\node [style=none] (32) at (6.25, -1) {};
		\node [style=none] (33) at (5.25, 0) {};
		\node [style=none] (34) at (7.25, 0) {};
		\node [style=bn] (35) at (6.25, -1) {};
		\node [style=morphism] (36) at (5.25, 0.5) {$\Q$};
		\node [style=none] (37) at (5.25, 0.75) {};
		\node [style=none] (38) at (5.25, 1.5) {};
		\node [style=none] (39) at (7.25, 1.5) {};
		\node [style=none] (40) at (5.25, 2) {$\tarS$};
		\node [style=none] (41) at (7.25, 2) {$\tarS$};
		\node [style=none] (42) at (3.75, 0) {$+$};
		\node [style=none] (43) at (5.25, 0.25) {};
		\node [style=none] (44) at (8.5, 0) {$=$};
		\node [style=state] (45) at (10.75, -2) {$\mu$};
		\node [style=none] (46) at (10.75, -1.75) {};
		\node [style=none] (47) at (10.75, -1) {};
		\node [style=none] (48) at (9.75, 0) {};
		\node [style=none] (49) at (11.75, 0) {};
		\node [style=bn] (50) at (10.75, -1) {};
		\node [style=morphism] (51) at (11.75, 0.5) {$\P$};
		\node [style=none] (52) at (9.75, 0) {};
		\node [style=none] (53) at (9.75, 1.5) {};
		\node [style=none] (54) at (11.75, 1.5) {};
		\node [style=none] (55) at (9.75, 2) {$\tarS$};
		\node [style=none] (56) at (11.75, 2) {$\tarS$};
		\node [style=state] (57) at (15.25, -2) {$\mu$};
		\node [style=none] (58) at (15.25, -1.75) {};
		\node [style=none] (59) at (15.25, -1) {};
		\node [style=none] (60) at (14.25, 0) {};
		\node [style=none] (61) at (16.25, 0) {};
		\node [style=bn] (62) at (15.25, -1) {};
		\node [style=morphism] (63) at (16.25, 0.5) {$\Q$};
		\node [style=none] (64) at (14.25, 0.25) {};
		\node [style=none] (65) at (14.25, 1.5) {};
		\node [style=none] (66) at (16.25, 1.5) {};
		\node [style=none] (67) at (14.25, 2) {$\tarS$};
		\node [style=none] (68) at (16.25, 2) {$\tarS$};
		\node [style=none] (69) at (13.25, 0) {$+$};
		\node [style=none] (70) at (14.25, 0.25) {};
		\node [style=none] (71) at (17.75, 0) {$=$};
		\node [style=state] (72) at (20, -2) {$\mu$};
		\node [style=none] (73) at (20, -1.75) {};
		\node [style=none] (74) at (20, -1) {};
		\node [style=none] (75) at (19, 0) {};
		\node [style=none] (76) at (21, 0) {};
		\node [style=bn] (77) at (20, -1) {};
		\node [style=morphism] (78) at (21, 0.5) {$\P+\Q$};
		\node [style=none] (79) at (19, 0) {};
		\node [style=none] (80) at (19, 1.5) {};
		\node [style=none] (81) at (21, 1.5) {};
		\node [style=none] (82) at (19, 2) {$\tarS$};
		\node [style=none] (83) at (21, 2) {$\tarS$};
	\end{pgfonlayer}
	\begin{pgfonlayer}{edgelayer}
		\draw (3.center) to (2.center);
		\draw [bend left=45] (10.center) to (3.center);
		\draw [bend left=45] (3.center) to (9.center);
		\draw (14.center) to (13.center);
		\draw (10.center) to (15.center);
		\draw (20.center) to (19.center);
		\draw [bend left=45] (22.center) to (20.center);
		\draw [bend left=45] (20.center) to (21.center);
		\draw (26.center) to (25.center);
		\draw (22.center) to (27.center);
		\draw (32.center) to (31.center);
		\draw [bend left=45] (34.center) to (32.center);
		\draw [bend left=45] (32.center) to (33.center);
		\draw (38.center) to (37.center);
		\draw (34.center) to (39.center);
		\draw (43.center) to (33.center);
		\draw (21.center) to (24);
		\draw (47.center) to (46.center);
		\draw [bend left=45] (49.center) to (47.center);
		\draw [bend left=45] (47.center) to (48.center);
		\draw (53.center) to (52.center);
		\draw (49.center) to (54.center);
		\draw (59.center) to (58.center);
		\draw [bend left=45] (61.center) to (59.center);
		\draw [bend left=45] (59.center) to (60.center);
		\draw (65.center) to (64.center);
		\draw (61.center) to (66.center);
		\draw (70.center) to (60.center);
		\draw (74.center) to (73.center);
		\draw [bend left=45] (76.center) to (74.center);
		\draw [bend left=45] (74.center) to (75.center);
		\draw (80.center) to (79.center);
		\draw (76.center) to (81.center);
	\end{pgfonlayer}
\end{tikzpicture}
	\end{cdiag}
	The first and third steps use Remark \ref{rem:diagrammatic-reasoning-semiadditive}, and the second uses reversibility of $\P$ and $\Q$.
\end{proof}

\begin{proposition} \label{prop:reversibility-of-difference}
	Let $\mu : \ind \to \tarS$ and $\P, \Q : \tarS \to \tarS$ be morphisms in a finitely cancellative semiadditive CD category $\C$, where $\mu$ is finite and $\Q$ is substochastic.
	If $\P + \Q$ and $\Q$ are both $\mu$-reversible, then so is $\P$.
\end{proposition}

\begin{proof}
	Suppose $\P + \Q$ and $\Q$ are $\mu$-reversible.
	Then similar to \eqref{eq:reversibility-of-sum}, we have
	\begin{cdiag*}
		\begin{tikzpicture}
	\begin{pgfonlayer}{nodelayer}
		\node [style=state] (13) at (1, -2) {$\mu$};
		\node [style=none] (14) at (1, -1.75) {};
		\node [style=none] (15) at (1, -1) {};
		\node [style=none] (16) at (0, 0) {};
		\node [style=none] (17) at (2, 0) {};
		\node [style=bn] (18) at (1, -1) {};
		\node [style=none] (20) at (0, 0) {};
		\node [style=none] (21) at (0, 1.5) {};
		\node [style=none] (22) at (2, 1.5) {};
		\node [style=none] (23) at (0, 2) {$\tarS$};
		\node [style=none] (24) at (2, 2) {$\tarS$};
		\node [style=none] (39) at (-1, 0) {$=$};
		\node [style=state] (40) at (-3, -2) {$\mu$};
		\node [style=none] (41) at (-3, -1.75) {};
		\node [style=none] (42) at (-3, -1) {};
		\node [style=none] (43) at (-4, 0) {};
		\node [style=none] (44) at (-2, 0) {};
		\node [style=bn] (45) at (-3, -1) {};
		\node [style=morphism] (46) at (-4, 0.5) {$P+Q$};
		\node [style=none] (47) at (-4, 0.75) {};
		\node [style=none] (48) at (-4, 1.5) {};
		\node [style=none] (49) at (-2, 1.5) {};
		\node [style=none] (50) at (-4, 2) {$\tarS$};
		\node [style=none] (51) at (-2, 2) {$\tarS$};
		\node [style=morphism] (52) at (2, 0.5) {$P + Q$};
		\node [style=state] (53) at (-12, -2) {$\mu$};
		\node [style=none] (54) at (-12, -1.75) {};
		\node [style=none] (55) at (-12, -1) {};
		\node [style=none] (56) at (-13, 0) {};
		\node [style=none] (57) at (-11, 0) {};
		\node [style=bn] (58) at (-12, -1) {};
		\node [style=morphism] (59) at (-13, 0.5) {$P$};
		\node [style=none] (60) at (-13, 0) {};
		\node [style=none] (61) at (-13, 1.5) {};
		\node [style=none] (62) at (-11, 1.5) {};
		\node [style=none] (63) at (-13, 2) {$\tarS$};
		\node [style=none] (64) at (-11, 2) {$\tarS$};
		\node [style=state] (65) at (-8, -2) {$\mu$};
		\node [style=none] (66) at (-8, -1.75) {};
		\node [style=none] (67) at (-8, -1) {};
		\node [style=none] (68) at (-9, 0) {};
		\node [style=none] (69) at (-7, 0) {};
		\node [style=bn] (70) at (-8, -1) {};
		\node [style=morphism] (71) at (-9, 0.5) {$Q$};
		\node [style=none] (72) at (-9, 0.25) {};
		\node [style=none] (73) at (-9, 1.5) {};
		\node [style=none] (74) at (-7, 1.5) {};
		\node [style=none] (75) at (-9, 2) {$\tarS$};
		\node [style=none] (76) at (-7, 2) {$\tarS$};
		\node [style=none] (77) at (-10.25, 0) {$+$};
		\node [style=none] (78) at (-9, 0.25) {};
		\node [style=none] (80) at (-6, 0) {$=$};
		\node [style=none] (81) at (3.25, -5.5) {$=$};
		\node [style=state] (82) at (5.25, -7.5) {$\mu$};
		\node [style=none] (83) at (5.25, -7.25) {};
		\node [style=none] (84) at (5.25, -6.5) {};
		\node [style=none] (85) at (4.25, -5.5) {};
		\node [style=none] (86) at (6.25, -5.5) {};
		\node [style=bn] (87) at (5.25, -6.5) {};
		\node [style=morphism] (88) at (6.25, -5) {$P$};
		\node [style=none] (89) at (4.25, -5.5) {};
		\node [style=none] (90) at (4.25, -4) {};
		\node [style=none] (91) at (6.25, -4) {};
		\node [style=none] (92) at (4.25, -3.5) {$\tarS$};
		\node [style=none] (93) at (6.25, -3.5) {$\tarS$};
		\node [style=state] (94) at (9.75, -7.5) {$\mu$};
		\node [style=none] (95) at (9.75, -7.25) {};
		\node [style=none] (96) at (9.75, -6.5) {};
		\node [style=none] (97) at (8.75, -5.5) {};
		\node [style=none] (98) at (10.75, -5.5) {};
		\node [style=bn] (99) at (9.75, -6.5) {};
		\node [style=morphism] (100) at (8.75, -5) {$Q$};
		\node [style=none] (101) at (8.75, -5.25) {};
		\node [style=none] (102) at (8.75, -4) {};
		\node [style=none] (103) at (10.75, -4) {};
		\node [style=none] (104) at (8.75, -3.5) {$\tarS$};
		\node [style=none] (105) at (10.75, -3.5) {$\tarS$};
		\node [style=none] (106) at (7.5, -5.5) {$+$};
		\node [style=none] (107) at (8.75, -5.25) {};
		\node [style=state] (108) at (-3, -7.5) {$\mu$};
		\node [style=none] (109) at (-3, -7.25) {};
		\node [style=none] (110) at (-3, -6.5) {};
		\node [style=none] (111) at (-4, -5.5) {};
		\node [style=none] (112) at (-2, -5.5) {};
		\node [style=bn] (113) at (-3, -6.5) {};
		\node [style=morphism] (114) at (-2, -5) {$P$};
		\node [style=none] (115) at (-4, -5.5) {};
		\node [style=none] (116) at (-4, -4) {};
		\node [style=none] (117) at (-2, -4) {};
		\node [style=none] (118) at (-4, -3.5) {$\tarS$};
		\node [style=none] (119) at (-2, -3.5) {$\tarS$};
		\node [style=state] (120) at (1, -7.5) {$\mu$};
		\node [style=none] (121) at (1, -7.25) {};
		\node [style=none] (122) at (1, -6.5) {};
		\node [style=none] (123) at (0, -5.5) {};
		\node [style=none] (124) at (2, -5.5) {};
		\node [style=bn] (125) at (1, -6.5) {};
		\node [style=morphism] (126) at (2, -5) {$Q$};
		\node [style=none] (127) at (0, -5.25) {};
		\node [style=none] (128) at (0, -4) {};
		\node [style=none] (129) at (2, -4) {};
		\node [style=none] (130) at (0, -3.5) {$\tarS$};
		\node [style=none] (131) at (2, -3.5) {$\tarS$};
		\node [style=none] (132) at (-0.75, -5.5) {$+$};
		\node [style=none] (133) at (0, -5.25) {};
		\node [style=none] (134) at (-5, -5.5) {$=$};
	\end{pgfonlayer}
	\begin{pgfonlayer}{edgelayer}
		\draw (15.center) to (14.center);
		\draw [bend left=45] (17.center) to (15.center);
		\draw [bend left=45] (15.center) to (16.center);
		\draw (21.center) to (20.center);
		\draw (17.center) to (22.center);
		\draw (42.center) to (41.center);
		\draw [bend left=45] (44.center) to (42.center);
		\draw [bend left=45] (42.center) to (43.center);
		\draw (48.center) to (47.center);
		\draw (44.center) to (49.center);
		\draw (43.center) to (46);
		\draw (55.center) to (54.center);
		\draw [bend left=45] (57.center) to (55.center);
		\draw [bend left=45] (55.center) to (56.center);
		\draw (61.center) to (60.center);
		\draw (57.center) to (62.center);
		\draw (67.center) to (66.center);
		\draw [bend left=45] (69.center) to (67.center);
		\draw [bend left=45] (67.center) to (68.center);
		\draw (73.center) to (72.center);
		\draw (69.center) to (74.center);
		\draw (78.center) to (68.center);
		\draw (84.center) to (83.center);
		\draw [bend left=45] (86.center) to (84.center);
		\draw [bend left=45] (84.center) to (85.center);
		\draw (90.center) to (89.center);
		\draw (86.center) to (91.center);
		\draw (96.center) to (95.center);
		\draw [bend left=45] (98.center) to (96.center);
		\draw [bend left=45] (96.center) to (97.center);
		\draw (102.center) to (101.center);
		\draw (98.center) to (103.center);
		\draw (107.center) to (97.center);
		\draw (110.center) to (109.center);
		\draw [bend left=45] (112.center) to (110.center);
		\draw [bend left=45] (110.center) to (111.center);
		\draw (116.center) to (115.center);
		\draw (112.center) to (117.center);
		\draw (122.center) to (121.center);
		\draw [bend left=45] (124.center) to (122.center);
		\draw [bend left=45] (122.center) to (123.center);
		\draw (128.center) to (127.center);
		\draw (124.center) to (129.center);
		\draw (133.center) to (123.center);
	\end{pgfonlayer}
\end{tikzpicture}
	\end{cdiag*}
	where we use reversibility of $\P + \Q$ in the first step, and reversibility of $\Q$ in the second.
	But now the summand involving $\Q$ is a composition of a finite morphism (namely $\mu$) and a substochastic morphism (using Proposition \ref{prop:substochastic-CD-category}), and is therefore finite by Proposition \ref{prop:finite-substochastic-composition}.
	Since $\C$ is finitely cancellative, we may cancel this summand on both sides to obtain reversibility of $\P$.
\end{proof}

We can now give necessary and sufficient conditions for reversibility of the full MH kernel $\Pmh$ from \eqref{eq:mh-kernel-abstract} in a finitely cancellative semiadditive CD category.
 
\begin{theorem} \label{thm:enrich_rev}
	Let $\C$ be a finitely cancellative semiadditive CD category, and suppose
	\begin{itemize}
		\item $\mu : \ind \to \augS$ is finite
		\item $\phi : \augS \to \augS$ is a deterministic involution
		\item $\alpha : \augS \to \ind$ is a probability
		\item $\rx = \frac{\dif (\phi \circ \mu)}{\dif \mu}$ is a Radon--Nikodym derivative.
	\end{itemize}
	Then $\Pmh$ defined in \eqref{eq:mh-kernel-abstract} is $\mu$-reversible if and only if the following ``balancing condition'' holds:
	\begin{equation} \label{eq:algebraic-balancing-condition-synthetic}
		\alpha = (\alpha \circ \phi) \ast \rx \qquad \text{$\mu$-almost everywhere.}
	\end{equation}
\end{theorem}

\begin{proof}
	By the definition of $\ast$ (Remark \ref{rem:effect-multiplication}) and $\mu$-almost everywhere (Remark \ref{rem:almost-sure-equality}), the balancing condition \eqref{eq:algebraic-balancing-condition-synthetic} becomes precisely \eqref{eq:algebraic-balancing-condition}.
	From Theorem \ref{thm:mh-categorical}, it follows that the first summand $\alpha \lact \phi$ of $\Pmh$ is $\mu$-reversible if and only if the balancing condition holds.
	The second summand $\comp{\alpha} \lact \Id_{\augS}$ is always $\mu$-reversible, as can be seen from inspection of its string diagram in \eqref{eq:mh-kernel-abstract}.
	By Proposition \ref{prop:substochastic-CD-category}, it is also substochastic, since $\comp{\alpha}$ is a probability (Remark \ref{rem:substochastic-interpretation}).
	The result now follows by Propositions \ref{prop:sum-of-reversible-is-reversible} and \ref{prop:reversibility-of-difference}.
\end{proof}

Instantiating this result in $\sfkern$ now yields the following:

\begin{corollary} \label{cor:andrieu2020}
	Suppose $\mu$ is a finite measure on a measurable space $\augS$, and $\phi : \augS \to \augS$ is a measurable involution such the pushforward measure $\mu^\phi$ is absolutely continuous with respect to $\mu$.
	Then the MH kernel $\Pimh$ defined in \eqref{eq:mh-kernel} is $\mu$-reversible if and only if
	\begin{equation} \label{eq:andrieu2020-balancing-condition}
		\alpha(\xi) = \alpha(\phi(\xi)) \, \rx(\xi) \qquad \text{for $\mu$-almost all $\xi \in \augS$,}
	\end{equation}
	where $\rx \coloneqq \frac{\dif \mu^\phi}{\dif\mu}$ denotes the Radon--Nikodym derivative.
\end{corollary}

\begin{proof}
	We apply Theorem \ref{thm:enrich_rev} with $\C \coloneqq \sfkern$.
	In this setting $\mu$ becomes an finite measure on a measurable space $\augS$ (Proposition \ref{prop:sfkern-finite}), and $\mu$-reversibility corresponds to the usual measure-theoretic notion of reversibility (Example \ref{ex:invariance-reversibility-stoch}).
	Any choice of involutive measurable function $\phi : \augS \to \augS$ lifts to become a determinsitic involutive kernel, and $\alpha$ and $\rx$ correspond to measurable functions $\augS \to [0, 1]$ and $\augS \to [0, \infty]$ respectively (Examples~\ref{exa:substoch_sfkern} and \ref{exa:effects-sfkern}).
	The condition $\rx = \frac{\dif (\phi \circ \mu)}{\dif \mu}$ says that $\rx$ is a Radon--Nikodym derivative in the usual measure-theoretic sense (Example \ref{exa:radon-nikodym}).
	The morphism $\Pmh$ becomes $\Pimh$ (Example \ref{exa:MH-kernel-sfkern}), while the balancing condition becomes \eqref{eq:andrieu2020-balancing-condition} (Remarks~\ref{rem:effect-multiplication} and \ref{rem:almost-sure-equality}), which gives the result.
\end{proof}

Corollary \ref{cor:andrieu2020} fully recovers Theorem \ref{thm:andrieu2020} in the case that $\mu$ and $\mu^\phi$ are both absolutely continuous with respect to each other (which ensures the Radon--Nikodym derivative $r$ exists).
It also supplies a converse statement, which was not given in \cite{Andrieu2020}.
In Section \ref{subsec:abs_cty} below, we show how we may also reason about absolute continuity within semiadditive CD categories, thereby recovering the rest of Theorem \ref{thm:andrieu2020} as well.

\subsubsection{Extension to skew-reversibility}

We can combine Theorem~\ref{thm:enrich_rev} with our earlier Proposition~\ref{prop:pi-s-reversibility-condition} (see also Remark \ref{rem:CD-cat-Markov-cat}) to obtain necessary and sufficient conditions for the more complex property of skew-reversibility (Section \ref{subsec:skew-reversibility}). %
Consider the morphism $s \circ \Pmh$ obtained by postcomposing our MH kernel with an involution $s : \augS \to \augS$.
From bilinearity of composition, we have
\begin{align*}
	s \circ \Pmh &= s \circ (\alpha \cdot \phi + \alpha^c \cdot \Id_\augS)
		= \alpha \cdot (s \circ \phi) + \alpha^c \cdot (s \circ \Id_\augS),
\end{align*}
where the second step is easy to check using string diagrams.
In the context of $\sfkern$, this recovers the following MH-like kernel considered by \cite[Section 4]{Andrieu2020}:
\begin{equation} \label{eq:skew-mh-kernel}
	\Pi_{\mathrm{sMH}}(\xi, \dif \xi') \coloneqq \alpha(\xi) \, \delta_{s(\phi(\xi))}(\dif \xi') + (1 - \alpha(\xi)) \, \delta_{s(\xi)}(\dif \xi').
\end{equation}
When $s$ is $\mu$-invariant, Proposition~\ref{prop:pi-s-reversibility-condition} implies that $\Pmh$ is $(\mu,s)$-reversible if and only if $\Pmh$ is $\mu$-reversible.
Together with Theorem~\ref{thm:enrich_rev}, this gives directly the following:

\begin{corollary} \label{cor:skew}
	Given the setup of Corollary \ref{cor:andrieu2020}, suppose we also have a measurable involution $s:\augS \to \augS$ that is $\mu$-invariant.
	Then $\Pi_{\mathrm{sMH}}$ is $(\mu,s)$-reversible if and only if the balancing condition \eqref{eq:andrieu2020-balancing-condition} holds.
\end{corollary}

Combined with our results on absolute continuity in Section~\ref{subsec:abs_cty} below, this gives a more streamlined and more general version of \cite[Proposition~3(c)]{Andrieu2020}, which gives sufficient conditions for $\Pi_{\mathrm{sMH}}$ to be reversible but does not supply a converse statement.

\subsection{Absolute continuity} \label{subsec:abs_cty}

From Remark \ref{rem:preorder-enrichment}, the hom-sets of a $\CMon$-enriched category $\C$ admit a preorder $\leq$.
They also admit a separate preorder related to \emph{absolute continuity}.
Recall that a measure $\nu$ is absolutely continuous with respect to another measure $\mu$ if $\mu(A) = 0$ implies $\nu(A) = 0$.
Intuitively, this says that $\nu$ assigns mass only to regions where $\mu$ does.
When both are absolutely continuous with respect to each other, they are said to be \emph{equivalent}.
The following extends these ideas to arbitrary morphisms in a $\CMon$-enriched category.

\begin{definition} \label{def:ac-preorder}
	Let $\P, \Q : \tarS \to \sndS$ be morphisms in a $\CMon$-enriched category $\C$.
	We write $\P \ll \Q$ if for all compatible morphisms $\R$ and $\S$ in $\C$ we have
	\begin{equation} \label{eq:ac-preorder-condition}
		\S \circ \Q \circ \R = 0 \implies \S \circ \P \circ \R = 0.
	\end{equation}
	We write $\P \equiv \Q$ if both $\P \ll \Q$ and $\Q \ll \P$ hold.
\end{definition}

Note that this does not require a CD structure: only a $\CMon$-enrichment.
In particular, we can ask what happens when $\C = \Kern$ (Example \ref{exa:sfkern}).
This is answered by the following result, proven in Section \ref{sec:appendix-ac-sfkern} of the Appendix.

\begin{proposition} \label{prop:ac-sfkern}
	Let $\P, \Q : \tarS \to \sndS$ be kernels (i.e.\ morphisms in $\Kern$).
	Then $\P \ll \Q$ if and only if $\P(x, \dif y)$ is absolutely continuous with respect to $\Q(x, \dif y)$ for all $x \in \tarS$.
	Accordingly, $\P \equiv \Q$ if and only if $\P(x, \dif y)$ and $\Q(x, \dif y)$ are equivalent as measures for all $x \in \tarS$.
\end{proposition}

It is straightforward to check that $\ll$ is reflexive and transitive and is therefore a preorder.
As for $\leq$ in Remark \ref{rem:preorder-enrichment}, it also induces an enrichment of $\C$ over preorders, as the following result shows.
We omit the proof, which is immediate from Definition \ref{def:ac-preorder}.

\begin{proposition} \label{prop:ac-preorder-enrichment}
	$\ll$ is preserved by pre- and post-composition.
\end{proposition}

The next two results establish basic facts about the preorder $\ll$.
Proposition \ref{prop:ac-zero-monic} shows that it has $0$ as the unique bottom element, and Proposition \ref{prop:ac-preorder-enrichment-conical} shows it is a less informative preorder than $\leq$ in many cases of interest.

\begin{proposition} \label{prop:ac-zero-monic}
    It holds that $0 \ll \P$ for every morphism $\P$. Also, if $\P \ll 0$, then $\P = 0$.
\end{proposition}

\begin{proof}
	We have directly from \eqref{eq:bilinearity-of-composition} that $\S \circ 0 \circ \R = 0$ for all compatible morphisms $\R$ and $\S$.
	This immediately gives the first claim.
	If $\P \ll 0$, then also $\S \circ \P \circ \R = 0$ for all compatible $\R$ and $\S$.
	Taking $\R$ and $\S$ to be identities shows $\P = 0$, which gives the second claim.
\end{proof}

\begin{proposition} \label{prop:ac-preorder-enrichment-conical}
	Let $\C$ be a $\CMon$-enriched category.
	If $\C$ is zero-sum-free (Remark \ref{rem:extra-properties-semiadditive}), then $\P \leq \Q$ implies $\P \ll \Q$.
\end{proposition}

\begin{proof}
	Suppose $\P \leq \Q$, and $\S \circ \Q \circ \R = 0$.
	Then $\S \circ \P \circ \R \leq \S \circ \Q \circ \R$ by Remark \ref{rem:preorder-enrichment}, and so $\S \circ \P \circ \R = 0$ since $\C$ is zero-sum-free.
\end{proof}

A distinct notion of absolute continuity was previously introduced by \cite[Definition 2.2.1]{Fritz2023} for Markov categories.
Their definition applies to morphisms with potentially different domains: given $\P : \tarS \to \sndS$ and $\Q : \auxS \to \sndS$, they say that \emph{$\P$ is absolutely continuous with respect to $\Q$} if $\Q$-almost sure equality implies $\P$-almost sure equality (where their notion of \emph{almost sure equality} is an extension of \eqref{eq:almost_sure}; see \cite[Definition 2.1.1]{Fritz2023}).
In $\stoch$, this becomes: %
\[
	\text{$\Q(z, A) = 0$ for all $z \in \auxS$} \implies \text{$\P(x, A) = 0$ for all $x \in \tarS$.}
\]
Note that this is not a ``pointwise'' condition, even when $\tarS = \auxS$, and as such is generally distinct from the concept we consider (cf.\ Proposition \ref{prop:ac-sfkern} above).
However, when $\tarS = \auxS = \ind$, the following result shows that their notion implies ours under mild assumptions on $\C$.
See Section \ref{sec:appendix-ac-comparison} of the Appendix for the proof.

\begin{proposition} \label{prop:ac-comparison}
	Let $\C$ be a semiadditive CD category that is zero-monic and has no zero divisors (Remark \ref{rem:extra-properties-semiadditive}).
	If $\mu, \nu : \ind \to \tarS$ are morphisms in $\C$ such that $\nu$ is absolutely continuous with respect to $\mu$ in the sense of \cite[Definition 2.2.1]{Fritz2023}, then $\nu \ll \mu$ in our sense.
\end{proposition}

\subsection{Lebesgue decompositions}

A fundamental concept in measure theory is the notion of a \emph{Lebesgue decomposition}.
Given measures $\nu$ and $\mu$ on the same measurable space, recall that a Lebesgue decomposition of $\nu$ with respect to $\mu$ is an expression of the form $\nu = \nu_\ac + \nu_\si$,
where $\nu_\ac$ is a measure that is absolutely continuous with respect to $\mu$, and $\nu_\si$ is another measure that is \emph{singular} with respect to $\mu$.
Here singular means there exists a measurable set $S$ such that $\nu_\si(S) = 0$ and $\mu(S^c) = 0$.
Intuitively, $\nu_\si$ and $\mu$ assign mass to disjoint regions of the space.
It is a classical result that a Lebesgue decomposition always exists whenever $\nu$ and $\mu$ are $\sigma$-finite \cite[Section~32]{halmos2013measure}.
We now establish an abstract version of this concept in the setting of $\CMon$-enriched categories.

\begin{definition}
	Let $\C$ be a $\CMon$-enriched category.
	A \emph{meet} of morphisms $\P, \Q : \tarS \to \sndS$ in $\C$ is a greatest lower bound of $\P$ and $\Q$ with respect to the preorder $\ll$.
\end{definition}

Meets do not need always need to exist.
Moreover, when they do, they are usually not unique, since $\ll$ is only a preorder.

\begin{proposition} \label{prop:meet-isomorphism}
	Meets are preserved by pre- and post-composition with isomorphisms.
\end{proposition}

\begin{proof}
	Suppose $\R$ is a meet of $\P$ and $\Q$.
	Given compatible isomorphisms $\varphi$ and $\psi$, Proposition \ref{prop:ac-preorder-enrichment} immediately gives that $\psi \circ \R \circ \varphi$ is a lower bound of $\psi \circ \P \circ \varphi$ and $\psi \circ \Q \circ \varphi$ with respect to $\ll$.
	Now suppose that $\S$ is some other lower bound.
	Then $\psi^{-1} \circ \S \circ \varphi^{-1}$ is a lower bound of $\P$ and $\Q$, and so $\psi^{-1} \circ \S \circ \varphi^{-1} \ll \R$.
	Proposition~\ref{prop:ac-preorder-enrichment} then gives $\S \ll \psi \circ \R \circ \varphi$ as required.
\end{proof}

Meets allow us to capture the notion of \textit{singular measures} abstractly via the following definition.
Proposition \ref{prop:singular-measures-meet} then shows that this recovers the usual measure-theoretic notion of singularity in the case of kernels (see Section \ref{sec:appendix-singular-measures-meet} of the Appendix for the proof).

\begin{definition}
	Given morphisms $\P, \Q : \tarS \to \sndS$ in a $\CMon$-enriched category $\C$, we write $\P \perp \Q$ if $0$ is a meet of $\P$ and $\Q$.
\end{definition}

\begin{proposition} \label{prop:singular-measures-meet}
	Suppose $\P, \Q : \tarS \to \sndS$ are kernels (i.e.\ morphisms in $\Kern$).
	If $\P(x, \dif y)$ and $\Q(x, \dif y)$ are singular for all $x \in \tarS$, then $\P \perp \Q$.
	The converse also holds when all singleton subsets of $\tarS$ are measurable, and $\P$ and $\Q$ are pointwise $\sigma$-finite.
\end{proposition}

The following result shows that $\perp$ and $\ll$ interact as measure-theoretic intuitions (in terms of ``assigning mass'') would suggest.

\begin{proposition} \label{prop:ac-singular-transitivity}
	If $\P \ll \Q$ and $\Q \perp \R$, then $\P \perp \R$.
\end{proposition}

\begin{proof}
	If we have $\S$ with $\S \ll \P$ and $\S \ll \R$, then $\S \ll \Q$ also, and so $\S = 0$ since $0$ is a meet of $Q$ and $R$.
\end{proof}

The preceding definitions now allow us to give an abstract version of the notion of a Lebesgue decomposition in a $\CMon$-enriched category as follows.

\begin{definition}
	Let $\P, \Q : \tarS \to \sndS$ be morphisms in a $\CMon$-enriched category $\C$.
	A \emph{$\Q$-decomposition} of $\P$ is a pair of morphisms $\P_\ac, \P_\si : \tarS \to \sndS$ such that
	\begin{itemize}
		\item $\P = \P_\ac + \P_\si$,
		\item $\P_\ac$ is a meet of $\P$ and $\Q$,
		\item $\P_\si \perp \Q$.
	\end{itemize}
\end{definition}

The following result shows that this definition recovers the classical notion of a Lebesgue decomposition in the case of kernels.
See Section \ref{sec:appendix-lebesgue-decomposition-meet} of the Appendix for the proof.

\begin{proposition} \label{prop:lebesgue-decomposition-sfkern}
	Suppose $\P, \Q, \P_\ac, \P_\si : \tarS \to \sndS$ are kernels (i.e.\ morphisms in $\Kern$).
	If $\P(x, \dif y) = \P_\ac(x, \dif y) + \P_\si(x, \dif y)$ is a Lebesgue decomposition with respect to $\Q(x, \dif y)$ for all $x \in \tarS$, then $(\P_\ac, \P_\si)$ is a $\Q$-decomposition of $\P$.
	The converse holds when all singleton subsets of $\tarS$ are measurable, and $\P_\ac$ and $\P_\si$ are pointwise $\sigma$-finite.
\end{proposition}

\subsubsection{Completing Theorem \ref{thm:andrieu2020}}

We now wish to complete the proof of Theorem \ref{thm:andrieu2020} using the tools developed above.
It remains to show the first part of this result, which involves a decomposition of the target measure $\mu$ according to an involution $\phi$.
Lebesgue decompositions capture this abstractly via the following result.

\begin{theorem} \label{thm:lebesgue-decomposition-involution}
	Let $\C$ be a zero-sum-free $\CMon$-enriched category.
	Let $\P : \tarS \to \sndS$ and $\phi : \sndS \to \sndS$ be morphisms in $\C$, where $\phi$ is an involution, and suppose $(\P_\ac, \P_\si)$ is a $(\phi \circ \P)$-decomposition of $\P$.
	Then we have both of the following:
	\begin{enumerate}[(1)]
		\item $\P_\ac \equiv \phi \circ \P_\ac$
		\item $\P_\si \perp \phi \circ \P_\si$
	\end{enumerate}
\end{theorem}

\begin{proof}
	(1) By Proposition \ref{prop:meet-isomorphism}, $\phi \circ \P_\ac$ is a meet of $\phi \circ \P$ and $\P = \phi \circ \phi \circ \P$.
	Since $(\P_\ac, \P_\si)$ is a $(\phi \circ \P)$-decomposition, we also have $\P_\ac \ll \phi \circ \P$ (by definition) and $\P_\ac \ll \P$ (by Proposition \ref{prop:ac-preorder-enrichment-conical}, since $\P_\ac \leq \P$).
	The meet property then implies $\P_\ac \ll \phi \circ \P_\ac$, which establishes one direction of (1).
	The other direction follows from Proposition \ref{prop:ac-preorder-enrichment}, which gives
    \[
		\phi \circ \P_\ac \ll \phi \circ \phi \circ \P_\ac = \P_\ac
	\]
	and hence $\phi \circ \P_\ac \equiv \P_\ac$ as required.

	(2) Suppose $R \ll \P_\si$ and $R \ll \phi \circ \P_\si$.
	Since $\P_\si \leq \P$, we have $\P_\si \ll \P$ and hence $\phi \circ \P_\si \ll \phi \circ \P$ by Propositions \ref{prop:ac-preorder-enrichment-conical} and \ref{prop:ac-preorder-enrichment} respectively.
	It follows that $R \ll \phi \circ \P$.
	But now since $0$ is a meet of $\P_\si$ and $\phi \circ \P$, we therefore have $R \ll 0$ as required.
\end{proof}

Instantiating this result in $\sfkern$ now gives the remaining part of Theorem \ref{thm:andrieu2020}:

\begin{corollary}
	Suppose $\mu$ is a $\sigma$-finite measure on $\augS$, and $\phi : \augS \to \augS$ is a measurable involution.
	Then there exists a measurable $S \subseteq \augS$ such that the restriction $\restr{\mu}{S}$ and its pushforward $\restr{\mu}{S}^\phi$ are equivalent, and $\restr{\mu}{S^c}$ and its pushforward $\restr{\mu}{S^c}^\phi$ are singular.
\end{corollary}

\begin{proof}
	Since $\mu$ and $\mu^\phi$ are both $\sigma$-finite measures, there exists a $\mu^\phi$-decomposition $\mu = \mu_\ac + \mu_\si$ \cite[Section 32, Theorem C]{halmos2013measure}.
	Regarding these as morphisms in $\sfkern$, Theorem \ref{thm:lebesgue-decomposition-involution} gives $\mu_\ac \equiv \mu_\ac^\phi$ and $\mu_\si \perp \mu_\si^\phi$, and Proposition \ref{prop:ac-singular-transitivity} gives $\mu_\ac \perp \mu_\si$.
	This shows $\mu_\ac$ and $\mu_\ac^\phi$ are equivalent (Proposition \ref{prop:ac-sfkern})
	and $\mu_\si$ and $\mu_\si^\phi$ are singular (Proposition \ref{prop:singular-measures-meet}).
	Similarly, $\mu_\ac$ and $\mu_\si$ are singular, which by definition just says there is a measurable $S \subseteq \augS$ with $\mu_\ac(S^c) = 0$ and $\mu_\si(S) = 0$, and hence $\mu_\ac = \restr{\mu}{S}$ and $\mu_\si = \restr{\mu}{S^c}$ as required.
\end{proof}
\section{Conclusion and future work}

In this work we developed an enrichment of CD categories over commutative monoids, providing a minimal setting for reasoning about sums of morphisms, and used it to study the involutive MH framework of \cite{Andrieu2020}.
This gave a purely algebraic derivation of their ``balancing condition'' in Theorem \ref{thm:andrieu2020}, together with a converse result and an immediate extension to the skew-reversible setting.

An interesting direction for future work is to give general conditions under which Lebesgue decompositions exist in this setting, rather than only isolating the properties that such decompositions ought to satisfy.
This appears to require a suitable version of the Radon--Nikodym theorem, which we conjecture may be attainable by assuming that the preorder $\leq$ induced by the enrichment is complete.
We also expect that making this precise will involve an explicit notion of \emph{predicates}, suggesting an interesting connection with \emph{effectus theory} \cite{cho2015introduction}, which we leave to future work.

\bibliography{bib-mcmc_cats}

\appendix

\section{Measure theory terminology and notation} \label{app:measure_theory}

We include here some basic terminology and notation from measure theory that we use in the paper.

Recall that a \emph{kernel} is a function $\K : \tarS \times \Sigma_{\sndS} \to [0, \infty]$ such that each function $x \mapsto \K(x, A)$ is measurable, and each function $A \mapsto \K(x, A)$ is a measure.
A \emph{Markov kernel} is similar, but with the codomain restricted to $[0, 1]$, so that each $A \mapsto \K(x, A)$ becomes a probability measure.
When describing kernels, it is often convenient to use ``infinitesimal'' notation.
So for example, we denote by $\K(x, \dif y)$ the measure $A \mapsto \K(x, A)$ on $\sndS$ obtained from a fixed input $x \in \tarS$.
(Sometimes in the literature this measure is written as $\K(x, \cdot)$ or $\K(x, -)$ instead.)

For other notation and terminology we use, let $\mu$ and $\nu$ be measures on a measurable space $\tarS$.
\begin{itemize}
	\item Given measurable $A \subseteq \tarS$, the \emph{restriction} of $\mu$ to $A$ is the measure $\restr{\mu}{A}$ on $\tarS$ defined by $\restr{\mu}{A}(B) := \mu(A \cap B)$ for all measurable $B \subseteq \tarS$.
	\item Given a measurable function $\phi : \tarS \to \sndS$, the \emph{pushforward of $\mu$ by $\phi$} is the measure $\mu^\phi$ on $\sndS$ defined by $\mu^\phi(A) := \mu(\phi^{-1}(A))$ for all measurable $A \subseteq \sndS$.
	\item The measure $\nu$ is \emph{absolutely continuous} with respect to $\mu$ if for all measurable $A \subseteq \tarS$, $\mu(A) = 0$ implies $\nu(A) = 0$.
	\item The measures $\mu$ and $\nu$ are \emph{equivalent} if they are both absolutely continuous with respect to each other.
	\item The measures $\mu$ and $\nu$ are \emph{singular} with respect to each other if there is some measurable $A \subseteq \tarS$ such that $\mu(A) = 0$ and $\nu(A^c) = 0$.
	\item Given $x \in \tarS$, the \emph{Dirac delta} at $x$ is the measure defined for any measurable $A\subset \tarS$ by
	\[
		\delta_x(A) := \begin{cases}
			1, & x \in A, \\
			0, & x \not\in A,
		\end{cases}
	\]
\end{itemize}

\section{Additional background on involutive MH} \label{app:MCMC_involutions}

We provide here some additional discussion of the involutive MH framework of \cite{Andrieu2020}: how it relates to our Theorem \ref{thm:andrieu2020}, and how existing MH algorithms can be recovered within it.

\subsection{Form of the acceptance probability}

In \cite{Andrieu2020}, the authors state their version of Theorem \ref{thm:andrieu2020} in a slightly different style than we do.
Specifically, they start with a function of the form $a: [0,\infty) \to [0,1]$ (referred to as a \emph{balancing function} \cite{Zanella2020}) with the property that
\[
	a(t) = \begin{cases}
		0 & t = 0; \\
		t \cdot a(1/t) & t > 0.
	\end{cases}
\]
They then \emph{define} $\alpha \coloneqq a \circ \rx$, and show that it satisfies the balancing condition \eqref{eq:balancing-ae} as a consequence of this choice.
Then, in their proof of reversibility, they make no further use of $a$: they only use the fact that it implies \eqref{eq:balancing-ae}.
This motivates our approach, which is of course implied whenever $\alpha$ has the form $a \circ \rx$ as in \cite{Andrieu2020}.

\subsection{Recovering existing MH algorithms}

\subsubsection{Classical Metropolis--Hastings}

The Metropolis--Hastings algorithm is typically presented in the following form (assuming for convenience that all the distributions and kernels involved have positive densities):

\begin{algorithm} \label{alg:MH}
	Given a target density $\pi$ on $\tarS$, proposal kernel $\Q:\tarS \to \tarS$ with density $q(x,y)$ and current state $x\in \tarS$:
	\begin{enumerate}
		\item Sample $Y\sim \Q(x,\dif y)$.
		\item Return $Y$ with probability $\min\left(1, \frac{\pi(Y) \, q(Y,x)}{\pi(x) \, q(x,Y)}\right)$.
		\item Otherwise return $x$.
	\end{enumerate}
\end{algorithm}

At first glance, it is not immediately obvious how this algorithm can be recovered in terms of the involutive MH kernel $\Pimh$ from \eqref{eq:mh-kernel}, and so we briefly outline how this is done here.
As described in Section \ref{sec:augmented-reversibility}, the key idea is use a state space augmentation of the form in Corollary \ref{cor:derived-reversible-morphism}.
Specifically, let $\auxS \coloneqq \tarS$, and take $Q : \tarS \to \tarS$ to be some proposal kernel of choice.
Finally, let $P : \tarS \otimes \tarS \to \tarS \otimes \tarS$ be the kernel $\Pimh$ from \eqref{eq:mh-kernel}.
This requires choosing an involution $\phi : \tarS \otimes \tarS \to \tarS \otimes \tarS$ and acceptance probability $\alpha : \tarS \otimes \tarS \to [0, 1]$.
For these, make the following choices:
\begin{itemize}
	\item $\phi$ is the ``swap'' isomorphism: $\phi(x,y) = (y,x)$
	\item $\alpha$ computes the acceptance probability in step (2) of Algorithm \ref{alg:MH} above: $\alpha(x, y) \coloneqq \min\left(1, \frac{\pi(y) \, q(y,x)}{\pi(x) \, q(x,y)}\right)$.
\end{itemize}
It may now be checked that substituting these choices into the kernel $\Pi$ from \eqref{eq:mu-definition-in-proof} recovers exactly the transition kernel described in Algorithm \ref{alg:MH}.

\subsubsection{The Exchange Algorithm}

As a more advanced example, consider the \textit{Exchange Algorithm} \cite{Murray2012}, which was a major breakthrough in MCMC for \textit{doubly-intractable} models: the situation where we wish to sample from a posterior distribution $\pi(x)\propto p(x) \, \ell_x(z_\mathrm{obs})$, where $p$ is a prior density on $\tarS$, and given $x\in \tarS$, $\ell_x(\cdot)$ is the likelihood for the observed data $z_\mathrm{obs}\in \auxS$. The Exchange Algorithm is tailored for situations where this likelihood has an {intractable} normalising constant (and hence cannot be evaluated), but for any $x\in \tarS$, we are able to generate synthetic data from $\ell_x$. Given a proposal kernel $\Q:\tarS \to \tarS$ as in Algorithm~\ref{alg:MH}, the Exchange Algorithm can be phrased in the framework of \cite{Andrieu2020} as follows:

We take the augmented space $\augS := \tarS \times \auxS \times \tarS$ with augmented measure
$$
\mu(\dif x, \dif z, \dif y) = \pi(x) \, \ell_y(\dif z) \, \Q(x,\dif y),
$$
and the involution is $\phi(x,z,y)= (y,z,x)$. Finally, the acceptance probability is
$$
\alpha(x,z,y) = \min\left(1, \frac{\pi(y) \, q(y,x) \, \ell_x(z)}{\pi(x) \, q(x,y) \, \ell_y(z)}\right) = \min\left(1, \frac{p(y) \, q(y,x) \,\ell_y(z_{\mathrm{obs}}) \, \ell_x(z)}{p(x) \, q(x,y) \, \ell_x(z_{\mathrm{obs}}) \, \ell_y(z)}\right),
$$
which can be evaluated since the ratio $\frac{\ell_y(z_{\mathrm{obs}}) \, \ell_x(z)}{\ell_x(z_{\mathrm{obs}}) \, \ell_y(z)}$ cancels the intractable normalising constants.
The results of \cite{Andrieu2020} (which we recover synthetically via our Theorem~\ref{thm:enrich_rev}) then ensure the resulting Markov chain \eqref{eq:mh-kernel} is $\mu$-reversible.

\section{Background on categorical probability} \label{app:intro_to_category_theory}

Copy-discard (CD) and Markov categories can be understood in terms of the following algebraic structures of increasing specificity:
\begin{align*}
	\text{Categories} &\supseteq \text{Symmetric monoidal categories} \\
		&\supseteq \text{CD categories} \\
		&\supseteq \text{Markov categories}.
\end{align*}
Much has been written about these topics in the literature to-date.
We refer the reader to \cite{perrone2023starting} for an introduction to categories and symmetric monoidal categories, and to \cite{Cho2019,Fritz2020} for CD and Markov categories.
We give here an overview of the key terminology and notation we need in the main text.

\subsection{Symmetric monoidal categories}

Recall that a \emph{category} $\C$ is a collection of \emph{objects} $\tarS, \sndS, \ldots$ and a collection of \emph{morphisms} $\P : \tarS \to \sndS, \Q : \sndS \to \auxS, \ldots$, together with an associative composition operation $\circ$ defined for pairs of compatible morphisms, and a distinguished \emph{identity} morphism $\Id_\tarS : \tarS \to \tarS$ defined for each object $\tarS$.
Given objects $\tarS$ and $\sndS$, the set of morphisms from $\tarS$ to $\sndS$ in $\C$ is denoted by $\C(\tarS, \sndS)$, and is referred to as a \emph{hom-set}.

A \emph{symmetric monoidal category} is a category $\C$ equipped with a distinguished object $\ind$ called the \emph{unit}, a \emph{monoidal product} $\otimes$, and \emph{swap} isomorphisms $\swap_{\tarS, \sndS} : \tarS \otimes \sndS \to \sndS \otimes \tarS$ defined for every pair of objects $\tarS$ and $\sndS$.
These components also satisfy certain coherence conditions (see e.g.\ \cite[Chapter 6]{perrone2023starting}) that permit the use of \emph{string diagrams} for reasoning graphically about $\C$.
In string diagrams, a morphism $\P : \tarS \to \sndS$ is depicted as a box with input and output wires as follows:
\begin{cdiag*}
\begin{tikzpicture}
	\begin{pgfonlayer}{nodelayer}
		\node [style=none] (0) at (0, -1.75) {$\tarS$};
		\node [style=morphism] (1) at (0, 0) {$\P$};
		\node [style=none] (2) at (0, 1.75) {$\sndS$};
		\node [style=none] (3) at (0, 1.25) {};
		\node [style=none] (4) at (0, -1.25) {};
	\end{pgfonlayer}
	\begin{pgfonlayer}{edgelayer}
		\draw (3.center) to (1);
		\draw (4.center) to (1);
	\end{pgfonlayer}
\end{tikzpicture}
\end{cdiag*}
Note that the input to $\P$ appears at the bottom, and information flows up the page.
The monoidal unit $\ind$ is thought of as representing ``no information'', and hence wires labelled by it are omitted. %
Identities and composition are denoted suggestively as:
\begin{cdiag*}
	\begin{tikzpicture}
	\begin{pgfonlayer}{nodelayer}
		\node [style=morphism] (0) at (-5.5, -0.75) {$\Q$};
		\node [style=morphism] (1) at (-5.5, 1) {$\P$};
		\node [style=none] (2) at (-5.5, -2) {};
		\node [style=none] (3) at (-5.5, 2.25) {};
		\node [style=morphism] (25) at (-8.5, 0) {$\P \circ \Q$};
		\node [style=none] (26) at (-8.5, 1.5) {};
		\node [style=none] (27) at (-8.5, -1.5) {};
		\node [style=none] (28) at (-6.75, 0) {$=$};
		\node [style=none] (44) at (-15.25, -1) {};
		\node [style=none] (45) at (-15.25, 1) {};
		\node [style=morphism] (46) at (-15.25, 0) {$\Id$};
		\node [style=none] (49) at (-14, 0) {$=$};
		\node [style=none] (50) at (-13, 1) {};
		\node [style=none] (51) at (-13, -1) {};
	\end{pgfonlayer}
	\begin{pgfonlayer}{edgelayer}
		\draw (3.center) to (2.center);
		\draw (27.center) to (26.center);
		\draw (45.center) to (44.center);
		\draw (51.center) to (50.center);
	\end{pgfonlayer}
\end{tikzpicture}
\end{cdiag*}
and the monoidal product and swaps as follows:
\begin{cdiag*}
	\begin{tikzpicture}
	\begin{pgfonlayer}{nodelayer}
		\node [style=morphism] (4) at (-3.75, 0) {$\P$};
		\node [style=none] (6) at (-3.75, -1.5) {};
		\node [style=none] (7) at (-3.75, 1.5) {};
		\node [style=morphism] (8) at (-2.25, 0) {$\R$};
		\node [style=none] (9) at (-2.25, -1.5) {};
		\node [style=none] (10) at (-2.25, 1.5) {};
		\node [style=none] (11) at (5.25, -1) {};
		\node [style=none] (12) at (6.75, 1) {};
		\node [style=none] (13) at (5.25, 1) {};
		\node [style=none] (14) at (6.75, -1) {};
		\node [style=morphism] (29) at (-7, 0) {$\P \otimes \R$};
		\node [style=none] (30) at (-7.75, 0.25) {};
		\node [style=none] (31) at (-6.25, 0.25) {};
		\node [style=none] (32) at (-7.75, 1.5) {};
		\node [style=none] (33) at (-6.25, 1.5) {};
		\node [style=none] (34) at (-7.75, -0.25) {};
		\node [style=none] (35) at (-7.75, -1.5) {};
		\node [style=none] (36) at (-6.25, -0.25) {};
		\node [style=none] (37) at (-6.25, -1.5) {};
		\node [style=none] (38) at (-5, 0) {$=$};
		\node [style=none] (43) at (-3.75, 0.25) {};
		\node [style=morphism] (52) at (2.75, 0) {$\swap$};
		\node [style=none] (53) at (4.5, 0) {$=$};
		\node [style=none] (54) at (2, 0.25) {};
		\node [style=none] (55) at (3.5, 0.25) {};
		\node [style=none] (56) at (2, 1) {};
		\node [style=none] (57) at (3.5, 1) {};
		\node [style=none] (58) at (2, -0.25) {};
		\node [style=none] (59) at (3.5, -0.25) {};
		\node [style=none] (60) at (2, -1) {};
		\node [style=none] (61) at (3.5, -1) {};
	\end{pgfonlayer}
	\begin{pgfonlayer}{edgelayer}
		\draw (7.center) to (6.center);
		\draw (10.center) to (9.center);
		\draw [in=-90, out=90] (14.center) to (13.center);
		\draw [in=90, out=-90, looseness=0.75] (12.center) to (11.center);
		\draw (33.center) to (31.center);
		\draw (32.center) to (30.center);
		\draw (35.center) to (34.center);
		\draw (37.center) to (36.center);
		\draw (61.center) to (59.center);
		\draw (60.center) to (58.center);
		\draw (54.center) to (56.center);
		\draw (55.center) to (57.center);
	\end{pgfonlayer}
\end{tikzpicture}
\end{cdiag*}
with the input and output wires labelled appropriately in all cases.
A fundamental result of \cite{Joyal1991} shows that this notation is sound: if two string diagrams have the same wiring topology (i.e.\ one can be continuously deformed into the other without ``cutting'' or ``reconnecting'' wires), then they denote the same morphism in any symmetric monoidal category.
See \cite[Section 2]{Cho2019} for an accessible introduction to string diagrams in the context of categorical probability.

\subsection{Copy-discard categories} \label{app:CD_cats}

We recall the definition of CD categories \cite[Definition 2.2]{Cho2019} as follows:

\begin{definition} \label{def:cd_cat_full}
	A \textit{copy-discard (CD) category} is a symmetric monoidal category $\C$ such that every object $\tarS$ is equipped with distinguished morphisms $\del_{\tarS}: \tarS \to \ind$ and $\cop_\tarS: \tarS \to \tarS \otimes \tarS$, which are depicted respectively as follows:
	\begin{cdiag*}
	\begin{tikzpicture}
	\begin{pgfonlayer}{nodelayer}
		\node [style=none] (0) at (-5, 0) {$\del_\tarS$};
		\node [style=none] (2) at (-3, 0) {$=$};
		\node [style=bn] (3) at (-1.5, 0.5) {};
		\node [style=none] (4) at (-1.5, -0.5) {};
		\node [style=none] (5) at (3.5, 0) {$\cop_\tarS$};
		\node [style=none] (11) at (6, 0) {$=$};
		\node [style=none] (12) at (8.5, 0) {};
		\node [style=none] (13) at (8.5, -1) {};
		\node [style=none] (14) at (7.5, 1) {};
		\node [style=none] (15) at (9.5, 1) {};
		\node [style=bn] (16) at (8.5, 0) {};
		\node [style=none] (24) at (-1.5, -1) {$\tarS$};
		\node [style=none] (25) at (8.5, -1.5) {$\tarS$};
		\node [style=none] (26) at (7.5, 1.5) {$\tarS$};
		\node [style=none] (27) at (9.5, 1.5) {$\tarS$};
	\end{pgfonlayer}
	\begin{pgfonlayer}{edgelayer}
		\draw (3) to (4.center);
		\draw (12.center) to (13.center);
		\draw [bend right=45] (12.center) to (15.center);
		\draw [bend right=45] (14.center) to (12.center);
	\end{pgfonlayer}
\end{tikzpicture}
	\end{cdiag*}
	These satisfy the following relationships:
	\begin{cdiag*} %
		\begin{tikzpicture}
	\begin{pgfonlayer}{nodelayer}
		\node [style=bn] (0) at (-9.25, -0.25) {};
		\node [style=none] (1) at (-9.25, -1.5) {};
		\node [style=bn] (2) at (-10.25, 1) {};
		\node [style=none] (3) at (-8.25, 1) {};
		\node [style=bn] (4) at (-3, -0.25) {};
		\node [style=none] (5) at (-3, -1.5) {};
		\node [style=none] (6) at (-4, 1) {};
		\node [style=bn] (7) at (-2, 1) {};
		\node [style=none] (8) at (-8.25, 1.5) {};
		\node [style=none] (9) at (-7.25, 0) {$=$};
		\node [style=none] (10) at (-6, -1.5) {};
		\node [style=none] (11) at (-6, 1.5) {};
		\node [style=none] (12) at (-5, 0) {$=$};
		\node [style=none] (13) at (-4, 1.5) {};
		\node [style=bn] (14) at (4, -0.5) {};
		\node [style=none] (15) at (4, -1.5) {};
		\node [style=none] (16) at (3, 0.5) {};
		\node [style=none] (17) at (5, 1.5) {};
		\node [style=bn] (18) at (3, 0.5) {};
		\node [style=none] (19) at (2, 1.5) {};
		\node [style=none] (20) at (4, 1.5) {};
		\node [style=none] (21) at (6, 0) {$=$};
		\node [style=bn] (22) at (8, -0.5) {};
		\node [style=none] (23) at (8, -1.5) {};
		\node [style=none] (24) at (7, 1.5) {};
		\node [style=none] (25) at (9, 0.5) {};
		\node [style=bn] (26) at (9, 0.5) {};
		\node [style=none] (27) at (8, 1.5) {};
		\node [style=none] (28) at (10, 1.5) {};
		\node [style=bn] (29) at (-2, -4.75) {};
		\node [style=none] (30) at (-2, -5.5) {};
		\node [style=none] (31) at (-3, -4) {};
		\node [style=none] (32) at (-1, -4) {};
		\node [style=none] (33) at (-3, -2.5) {};
		\node [style=none] (34) at (-1, -2.5) {};
		\node [style=none] (35) at (-1, -4) {};
		\node [style=none] (36) at (-3, -4) {};
		\node [style=none] (37) at (0, -4) {$=$};
		\node [style=bn] (38) at (2, -4) {};
		\node [style=none] (39) at (2, -5.25) {};
		\node [style=none] (40) at (1, -2.75) {};
		\node [style=none] (41) at (3, -2.75) {};
	\end{pgfonlayer}
	\begin{pgfonlayer}{edgelayer}
		\draw (0) to (1.center);
		\draw [in=180, out=-90, looseness=1.25] (2) to (0);
		\draw [in=270, out=0, looseness=1.25] (0) to (3.center);
		\draw (4) to (5.center);
		\draw [in=180, out=-90, looseness=1.25] (6.center) to (4);
		\draw [in=270, out=0, looseness=1.25] (4) to (7);
		\draw (8.center) to (3.center);
		\draw (11.center) to (10.center);
		\draw (13.center) to (6.center);
		\draw (14) to (15.center);
		\draw [in=180, out=-90, looseness=1.25] (16.center) to (14);
		\draw [in=270, out=0, looseness=1.25] (14) to (17.center);
		\draw [in=180, out=-90, looseness=1.25] (19.center) to (18);
		\draw [in=270, out=0, looseness=1.25] (18) to (20.center);
		\draw (22) to (23.center);
		\draw [in=180, out=-90, looseness=1.25] (24.center) to (22);
		\draw [in=270, out=0, looseness=1.25] (22) to (25.center);
		\draw [in=180, out=-90, looseness=1.25] (27.center) to (26);
		\draw [in=270, out=0, looseness=1.25] (26) to (28.center);
		\draw (29) to (30.center);
		\draw [in=180, out=-90, looseness=1.25] (31.center) to (29);
		\draw [in=270, out=0, looseness=1.25] (29) to (32.center);
		\draw [in=90, out=-90] (33.center) to (35.center);
		\draw [in=90, out=-90] (34.center) to (36.center);
		\draw (38) to (39.center);
		\draw [in=180, out=-90, looseness=1.25] (40.center) to (38);
		\draw [in=270, out=0, looseness=1.25] (38) to (41.center);
	\end{pgfonlayer}
\end{tikzpicture}
	\end{cdiag*}
	and the following compatibility conditions with the monoidal structure:
	\begin{cdiag*} %
		\begin{tikzpicture}
	\begin{pgfonlayer}{nodelayer}
		\node [style=none] (0) at (-4, -1) {};
		\node [style=bn] (1) at (-4, 1) {};
		\node [style=none] (2) at (-4, -1.5) {$\tarS \otimes \sndS$};
		\node [style=none] (3) at (-2.5, 0) {$=$};
		\node [style=none] (4) at (-1, -1) {};
		\node [style=bn] (5) at (-1, 1) {};
		\node [style=none] (6) at (0, -1) {};
		\node [style=bn] (7) at (0, 1) {};
		\node [style=none] (8) at (-1, -1.5) {$\tarS$};
		\node [style=none] (9) at (0, -1.5) {$\sndS$};
		\node [style=bn] (18) at (5.5, -0.25) {};
		\node [style=none] (19) at (5.5, -1.25) {};
		\node [style=none] (20) at (4.5, 1.25) {};
		\node [style=none] (21) at (6.5, 1.25) {};
		\node [style=none] (22) at (5.5, -1.75) {$\tarS \otimes \sndS$};
		\node [style=none] (23) at (8.5, 0) {$=$};
		\node [style=bn] (24) at (10.75, -0.25) {};
		\node [style=none] (25) at (10.75, -1.25) {};
		\node [style=none] (26) at (9.75, 1.25) {};
		\node [style=none] (27) at (11.75, 1.25) {};
		\node [style=bn] (28) at (11.75, -0.25) {};
		\node [style=none] (29) at (11.75, -1.25) {};
		\node [style=none] (30) at (10.75, 1.25) {};
		\node [style=none] (31) at (12.75, 1.25) {};
		\node [style=none] (32) at (10.75, -1.75) {$\tarS$};
		\node [style=none] (33) at (11.75, -1.75) {$\sndS$};
		\node [style=none] (34) at (4.25, 1.75) {$\tarS \otimes \sndS$};
		\node [style=none] (35) at (6.75, 1.75) {$\tarS \otimes \sndS$};
		\node [style=none] (36) at (9.75, 1.75) {$\tarS$};
		\node [style=none] (37) at (10.75, 1.75) {$\sndS$};
		\node [style=none] (38) at (11.75, 1.75) {$\tarS$};
		\node [style=none] (39) at (12.75, 1.75) {$\sndS$};
	\end{pgfonlayer}
	\begin{pgfonlayer}{edgelayer}
		\draw (1) to (0.center);
		\draw (5) to (4.center);
		\draw (7) to (6.center);
		\draw (18) to (19.center);
		\draw [in=180, out=-90, looseness=1.25] (20.center) to (18);
		\draw [in=270, out=0, looseness=1.25] (18) to (21.center);
		\draw (24) to (25.center);
		\draw [in=180, out=-90] (26.center) to (24);
		\draw [in=-90, out=0] (24) to (27.center);
		\draw (28) to (29.center);
		\draw [in=180, out=-90] (30.center) to (28);
		\draw [in=-90, out=0] (28) to (31.center);
	\end{pgfonlayer}
\end{tikzpicture}
	\end{cdiag*}
	In addition, it also holds that
	\begin{equation} \label{eq:del_cop_id}
		\del_\ind = \Id_\ind
		\qquad\qquad
		\cop_\ind = \ind \xrightarrow{\cong} \ind \otimes \ind,
	\end{equation}
	where the right-hand map is the canonical isomorphism supplied by the underlying monoidal structure of $\C$.
\end{definition}

\begin{remark}\label{rem:sigma_finite}
	The primary example of a CD category we consider in the main text is $\sfkern$ (Example \ref{exa:sfkern}), which recall has \emph{$s$-finite kernels} as its morphisms.
	A statistically-trained reader may wonder whether it is possible to construct a CD category analogous to this, but using $\sigma$-finite kernels instead of $s$-finite ones.
	These two notions are closely related: $s$-finite \emph{measures} (countable sums of finite measures \cite[Page 66]{Kallenberg2021}) are precisely the pushforwards of $\sigma$-finite ones \cite[Proposition 7]{Staton2017}, and hence a $\sigma$-finite measure is always $s$-finite.
	However, the converse is false: the pushforward of a $\sigma$-finite measure by a measurable function need not be $\sigma$-finite, which means that $\sigma$-finiteness is in general not preserved under composition \cite[Section 5.3]{Staton2017}.
	For this reason, it is standard practice to use $\sfkern$ as the canonical CD category of kernels in categorical probability \cite{Cho2019, Staton2017}.
\end{remark}

\subsection{Markov categories}

We recall the definition of Markov categories \cite[Definition 2.1]{Fritz2020} as follows:

\begin{definition}\label{def:Markov_cat}
	A \emph{Markov category} is a CD category $\C$ such that for any morphism $P : \tarS \to \sndS$, we have
	\begin{cdiag} \label{eq:del_nat}
		\begin{tikzpicture}
	\begin{pgfonlayer}{nodelayer}
		\node [style=morphism] (0) at (-2, 0) {$\P$};
		\node [style=bn] (1) at (-2, 1.25) {};
		\node [style=none] (2) at (-2, -1.25) {};
		\node [style=none] (4) at (1.75, -1.25) {};
		\node [style=bn] (5) at (1.75, 1.25) {};
		\node [style=none] (6) at (0, 0) {$=$};
		\node [style=none] (7) at (-2, -1.75) {$\tarS$};
		\node [style=none] (8) at (1.75, -1.75) {$\tarS$};
	\end{pgfonlayer}
	\begin{pgfonlayer}{edgelayer}
		\draw (1) to (0);
		\draw (0) to (2.center);
		\draw (5) to (4.center);
	\end{pgfonlayer}
\end{tikzpicture}
	\end{cdiag}
\end{definition}
 
\begin{remark}
	In Markov categories, the final two CD category conditions \eqref{eq:del_cop_id} follow automatically from \eqref{eq:del_nat} (see \cite[Remark 2.3]{Fritz2020}).
	As such, they are often omitted from standard definitions of Markov categories in the literature (e.g.\ Definition 2.1 of \cite{Fritz2020}).
\end{remark}

\begin{example}\label{exa:stoch}
	The canonical example of a Markov category is the category $\stoch$ of measurable spaces and Markov kernels.
	Composition and the monoidal product are defined as described in Example \ref{ex:stoch-defn}.
	The monoidal unit $\ind$ is the singleton measurable space $\{*\}$ equipped with the trivial $\sigma$-algebra, and the identity, swap, deletion, and copy morphisms are defined as follows:
	\begin{align*}
		\Id_{\tarS}(x, \dif x') &\coloneqq \delta_x(\dif x') \\
		\swap_{\tarS,\sndS}((x, y), (\dif x', \dif y')) &\coloneqq \delta_{(y, x)}(\dif x', \dif y') \\
		\del_\tarS(x, \dif z) &\coloneqq \delta_{*}(\dif z) \\
		\cop_{\tarS}(x, (\dif x', \dif x'') ) &\coloneqq \delta_{(x,x)}(\dif x', \dif x''),
	\end{align*}
	where $\delta$ is the Dirac delta.
	It is shown e.g.\ in \cite[Example 2.5]{Cho2019} that these components satisfy the required axioms of a Markov category.
\end{example}

\section{Proofs}

\subsection{$\Kern$ is a category} \label{app:kern}

In Example~\ref{exa:sfkern} we described the category $\Kern$, whose objects are measurable spaces and whose morphisms are kernels (not necessarily $s$-finite).
Recall that a kernel is a function $K:\tarS \times \Sigma_{\sndS} \to \left[0,\infty\right]$ such that $x\mapsto K(x,A)$ is measurable and $A\mapsto K(x,A)$ is a measure.
Given another kernel $L: \sndS \times \Sigma_{\auxS} \to \left[0,\infty\right]$, the composition $L\circ K$ in $\Kern$ is defined as
\begin{equation} \label{eq:composition-of-kernels}
    (L\circ K)(x,A) \coloneqq \int K(x,\mathrm{d}y) \, L(y,A).
\end{equation}
It is easily checked that $L \circ K$ is in fact a kernel $\tarS \times \Sigma_{\auxS} \to \left[0,\infty\right]$, and that composition is unital with respect to the Dirac kernels, which serve as identities in $\Kern$.
However, since the Fubini--Tonelli theorem does not hold for general kernels, some care is required to check that composition is associative.
For completeness, we give a proof of this fact in the remainder of this section.

\begin{lemma} \label{lem:For-any-simple}
	For any $x \in \tarS$ and measurable function $f : \auxS \to [0,\infty]$, we have that
	\[
	\int(L\circ K)(x,\mathrm{d}z) \, f(z) = \int K(x,\mathrm{d}y)\left[\int L(y,\mathrm{d}z) \, f(z)\right].
	\]
\end{lemma}

\begin{proof}
	We first show the result when $f$ is an indicator function.
	For any $A\in\Sigma_{\auxS}$ and $x\in\tarS$, we have
	\begin{align*}
	\int (L\circ K)(x,\mathrm{d}z) \, \indicator_A(z) &= \int K(x,\mathrm{d}y) \, L(y,A) \\
	&= \int K(x,\mathrm{d}y)\left[\int L(y,\mathrm{d}z) \, \indicator_A(z)\right]
	\end{align*}
	where the first equality follows from the definition of $L\circ K$ as in \eqref{eq:composition-of-kernels}.
	By linearity, this now extends to all simple functions $f$, and hence to all measurable $f:\auxS \to[0,\infty]$ by the monotone convergence theorem.
\end{proof}

\begin{proposition} \label{prop:kern-composition-associative}
	Composition in $\Kern$ is associative.
\end{proposition}

\begin{proof}
	With $K$ and $L$ as above, let $M:\auxS \times\Sigma_{\W}\to\left[0,\infty\right]$ be another kernel.
	Choose any $x \in \tarS$ and $B \in \Sigma_{\W}$.
	We then have
	\begin{align*}
	(M \circ (L \circ K))(x,B) &= \int (L\circ K)(x,\mathrm{d}z) \, M(z,B) \\
	&= \int K(x,\mathrm{d}y)\left[\int L(y,\mathrm{d}z) \, M(z,B)\right] \\
	&= ((M\circ L)\circ K)(x,B),
	\end{align*}
	where the second step holds by Lemma \ref{lem:For-any-simple} applied to the function $f(z) = M(z,B)$.
	Since $x$ and $B$ were arbitrary, we conclude that $M\circ\left(L\circ K\right)=\left(M\circ L\right)\circ K$.
\end{proof}

\subsection{Substochasticity, probabilities, and convex combinations} \label{sec:appendix-substochasticity}

Let $\C$ be a symmetric monoidal category equipped with an initial object $\zeroobj$ and binary coproducts $\oplus$.
Recall that $\C$ is \emph{distributive} if the canonical morphisms
\begin{equation*} %
	\zeroobj \to \zeroobj \otimes \tarS
	\qquad
	\qquad
	(\tarS \otimes \sndS) \oplus (\tarS \otimes \auxS) \to \tarS \otimes (\sndS \oplus \auxS)
\end{equation*}
induced by initiality and the coproducts are isomorphisms.
A \emph{distributive Markov category} \cite[Definition 2]{ackerman2024probabilistic} is a Markov category equipped with an initial object and binary coproducts, whose underlying symmetric monoidal category is distributive, and whose chosen coproduct inclusions are all deterministic.

\begin{proof}[Proof of Proposition~\ref{prop:distributive}]
	We first show that the initial object $\zeroobj$ in $\C$ is initial in $\Cnorm$ as well.
	Observe that for any object $\tarS$, the following diagram in $\C$ commutes by initiality:
	\begin{cdiag*}
		\begin{tikzcd}
			\zeroobj \ar{r}{!} \ar[swap]{dr}{\del} & \tarS \ar{d}{\del} \\
			& \ind
		\end{tikzcd}
	\end{cdiag*}
	where here $!$ is the unique morphism $0 \to \tarS$ in $\C$.
	This shows that $!$ is normalised, and so $\zeroobj$ is initial in $\Cnorm$.

	Next, suppose $\tarS \oplus \sndS$ is a coproduct in $\C$.
	This gives the following diagram in $\C$, where $i_\tarS$ and $i_\sndS$ are the coproduct injections:
	\begin{cdiag} \label{eq:coproduct-inclusions-normalised-proof}
		\begin{tikzcd}
			\tarS \ar{r}{i_\tarS} \ar[swap]{dr}{\del} & \tarS \oplus \sndS \ar[dashed]{d}{\del} & \sndS \ar[swap]{l}{i_\sndS} \ar{dl}{\del} \\
			& \ind &
		\end{tikzcd}
	\end{cdiag}
	This diagram commutes by the assumption that $i_\tarS$ and $i_\sndS$ are deterministic and normalised. %
	From this it follows that $\del_{\tarS \oplus \sndS}$ is equal to the coproduct comparison map $[\del_\tarS, \del_\sndS]$ as indicated in the diagram.
	Now let $f : \tarS \to \auxS$ and $g : \sndS \to \auxS$ be morphisms in $\Cnorm$.
	Then the following diagram in $\C$ commutes:
	\begin{cdiag} \label{eq:coproduct-in-norm-proof}
		\begin{tikzcd}
			\tarS \ar{r}{i_\tarS} \ar[equals]{d} & \tarS \oplus \sndS \ar[dashed]{d}{[f,g]} & \sndS \ar[swap]{l}{i_\sndS} \ar[equals]{d} \\
			\tarS \ar[swap]{dr}{\del} \ar{r}{f} & \auxS \ar{d}{\del} & \sndS \ar{dl}{\del} \ar[swap]{l}{g} \\
			& \ind
		\end{tikzcd}
	\end{cdiag}
	By comparing with \eqref{eq:coproduct-inclusions-normalised-proof} and using the universal property of the coproduct in $\C$, we now have
	\[
		\del_\auxS \circ [f, g] = [\del_\tarS, \del_\sndS] = \del_{\tarS \oplus \sndS}.
	\]
	This shows that $[f, g]$ is normalised and hence that \eqref{eq:coproduct-in-norm-proof} is a diagram in $\Cnorm$, and so $i_\tarS$ and $i_\sndS$ constitute coproduct injections in $\Cnorm$ as well.

	It remains to show that $\Cnorm$ is distributive.
	For this, recall that since $\C$ has a $\CMon$-enrichment, its initial object and binary coproducts constitute a zero object and biproducts respectively \cite[Proposition 5]{garner2016when}. 
	But now it is standard to show that each functor $T_\tarS \coloneqq - \otimes \tarS$ preserves biproducts (i.e.\ $\C$ is distributive as a monoidal category) if and only if it is \emph{additive} in the sense that
	\begin{equation} \label{eq:additive-functor}
		T_\tarS(f + g) = T_\tarS(f) + T_\tarS(g).
	\end{equation}
	(See e.g.\ \cite[Proposition 4 in Chapter 8]{mac1998categories} for the analogous case of an \emph{additive} category.)
	From the definition of $T_\tarS$, \eqref{eq:additive-functor} says $(f + g) \otimes \Id_\tarS = (f \otimes \Id_\tarS) + (g \otimes \Id_\tarS)$, which is easily seen to be equivalent to the $\CMon$-enrichment being monoidal.
\end{proof}

\subsection{Cancellativity and $\sigma$-finiteness} \label{sec:appendix-sigma-finiteness}

\subsubsection{Proof of Proposition \ref{prop:sigma-finiteness-cancellative}}

In this section we give a proof of Proposition \ref{prop:sigma-finiteness-cancellative}.
This requires the following two lemmas.

\begin{lemma} \label{lem:measure-cancellative}
	Suppose $\mu$ is a measure on a measurable space $\tarS$.
	Then $\mu$ is $\sigma$-finite if and only if it is cancellative. %
\end{lemma}

\begin{proof}
	For the ``only if'' part, suppose $\mu$ is $\sigma$-finite.
	Then there exists disjoint measurable $A_n \subseteq \tarS$ such that $\tarS = \bigcup_{n=1}^\infty A_n$ and $\mu(A_n) < \infty$.
	Now suppose $\nu$ and $\rho$ are measures on $\tarS$ with $\mu + \nu = \mu + \rho$.
	Then for any measurable set $A \in \Sigma_\tarS$, we have
	\[
		\mu(A \cap A_n) + \nu(A \cap A_n) = \mu(A \cap A_n) + \rho(A \cap A_n).
	\]
	Since $\mu(A \cap A_n) \leq \mu(A_n) < \infty$, we may cancel this term on both sides to obtain $\nu(A \cap A_n) = \rho(A \cap A_n)$.
	By summing both sides over $n$ and using disjointness of the $A_n$, we then obtain $\nu(A) = \rho(A)$.
	Hence $\nu = \rho$ and so $\mu$ is cancellative.

	For the ``if'' part, suppose $\mu$ is not $\sigma$-finite.
	Say that a measurable set $A \subseteq \tarS$ is \emph{$\mu$-$\sigma$-finite} if the restriction of $\mu$ to $A$ is $\sigma$-finite.
	Next, for any measurable $A \subseteq \tarS$, define
	\[
		\nu(A) \coloneqq \begin{cases}
			0 & \text{if $A$ is $\mu$-$\sigma$-finite} \\
			\infty & \text{otherwise.}
		\end{cases}
	\]
	It is straightforward to check that $\nu$ is a measure on $\tarS$.
	Likewise, for any measurable $A \subseteq \tarS$, we have
	\begin{align*}
		\mu(A) + \nu(A) &= \begin{cases}
			\mu(A) + 0 & \text{if $A$ is $\mu$-$\sigma$-finite} \\ 
			\mu(A) + \infty & \text{otherwise}
		\end{cases} \\
		&= \mu(A)
	\end{align*}
	where the second step follows from the easily checked fact that $\mu(A) = \infty$ if $A$ is not $\mu$-$\sigma$-finite.
	It follows that $\mu + \nu = \mu$.
	But now $\nu$ is not the zero measure, since by assumption $\tarS$ is not $\mu$-$\sigma$-finite.
	This shows that $\mu$ is not cancellative.
\end{proof}

\begin{lemma} \label{lem:kernel-cancellative}
	Let $\K : \tarS \times \Sigma_\sndS \to [0, \infty]$ be a kernel.
	If the measure $\K(x, \dif y)$ is cancellative for all $x \in \tarS$, then $\K$ is cancellative.
	If all singleton sets $\{x\} \subseteq \tarS$ are measurable, then the converse also holds.
\end{lemma}

\begin{proof}
	For first claim, suppose each $\K(x, \dif y)$ is cancellative, and let $\Lx$ and $\Mx$ be kernels with $\K + \Lx = \K + \Mx$.
	Then for any $x \in \tarS$, we have
	\[
		\K(x, \dif y) + \Lx(x, \dif y) = \K(x, \dif y) + \Mx(x, \dif y)
	\]
	which gives $\Lx(x, \dif y) = \Mx(x, \dif y)$ and hence $\Lx = \Mx$.

	For the second part of the statement, suppose $K$ is cancellative.
	Choose $x_0 \in \tarS$, and let $\nu$ and $\rho$ be any measures on $\sndS$ such that
	\begin{equation} \label{eq:kernel-cancellative-goal}
		\K(x_0, \dif y) + \nu(\dif y) = \K(x_0, \dif y) + \rho(\dif y).
	\end{equation}
	Now define the following:
	\[
		\Lx(x, \dif y) \coloneqq \begin{cases}
			\nu(\dif y) & \text{if $x = x_0$} \\
			0 & \text{otherwise},
		\end{cases}
	\]
	and similarly for $\Mx$ with $\rho$ in place of $\nu$.
	It follows from our assumption that all singleton sets are measurable that $\Lx$ and $\Mx$ are kernels.
	This gives $\K + \Lx = \K + \Mx$ by \eqref{eq:kernel-cancellative-goal}, which by cancellativity of $\K$ gives $\Lx = \Mx$, and hence $\nu = \rho$ as required.
\end{proof}

We are now able to give a proof of Proposition \ref{prop:sigma-finiteness-cancellative} from the main text as follows.

\begin{proof}
	By Lemma \ref{lem:measure-cancellative}, $\sigma$-finiteness of $\K(x, \dif y)$ is equivalent to cancellativity of $\K(x, \dif y)$.
	The result now follows directly from Lemma \ref{lem:kernel-cancellative}.
\end{proof}

\subsubsection{Proof of Proposition \ref{prop:sfkern-finite}} \label{sec:appendix-sfkern-finite}

\begin{proof}
	Recall that when $\P$ is a kernel, the effect $\one \circ \P$ corresponds to the measurable function $f : \tarS \to [0, \infty]$ defined by
	\[
		f(x) \coloneqq \P(x, \sndS).
	\]
	(see Example~\ref{exa:substoch_sfkern}).
	If $f$ is finite, then it is certainly cancellative.
	So suppose it is not finite. 
	Then since $\{\infty\}$ is measurable in $[0, \infty]$, so is the set
	\[
		A \coloneqq f^{-1}(\{\infty\}) \neq \emptyset
	\]
	Now define the effect $\wx : \tarS \to \ind$ by
	\[
		\wx(x, \ind) \coloneqq \begin{cases}
			1 & \text{if $x \in A$} \\
			0 & \text{otherwise,}
		\end{cases}
	\]
	which is well-defined since $A$ is measurable.
	This gives $\one \circ \P + \wx = \one \circ \P$, but $\wx \neq 0$ since $A$ is nonempty.
\end{proof}

\subsection{Absolute continuity} \label{sec:appendix-lebesgue-decompositions}

\subsubsection{Proof of Proposition \ref{prop:ac-sfkern}} \label{sec:appendix-ac-sfkern}

\begin{proof}
	For the ``only if'' direction, take $\R : \ind \to \tarS$ to be the Dirac measure on $x$, and $\S : \sndS \to \ind$ to be the effect corresponding to the indicator function of a measurable $A \subseteq \sndS$ (i.e.\ $\S(y, \ind) = \indicator_A(y)$).
	Then the equations in \eqref{eq:ac-preorder-condition} become equivalently $\Q(x, A) = 0$ and $\P(x, A) = 0$ respectively, which shows absolute continuity.

	For the ``if'' direction, suppose each $\P(x, \dif y)$ is absolutely continuous with respect to $\Q(x, \dif y)$.
	Now let $\R : \W \to \tarS$ and $\S : \sndS \to \auxS$ satisfy $\S \circ \Q \circ \R = 0$, and choose $w \in \W$ and measurable $A \subseteq \auxS$ arbitrarily.
	It is a standard result that we can choose some constants $a_i > 0$ and measurable $A_i \subseteq \sndS$ such that
	\[
		\S(y, A) = \sum_{i=1}^\infty a_i \, \indicator_{A_i}(y)
	\]
	holds for all $y \in \sndS$.
	This gives, for all $x \in \tarS$,
	\[
		(\S \circ \Q)(x, A) = \int \Q(x, \dif y) \, \S(y, A) = \sum_{i=1}^\infty a_i \, \Q(x, A_i)
	\]
	by the dominated convergence theorem.
	But then since $\S \circ \Q \circ \R = 0$, we also have the following for $\R(w, \dif x)$-almost all $x \in \tarS$:
	\[
		0 = (\S \circ \Q)(x, A) = \sum_{i=1}^\infty a_i \, \Q(x, A_i).
	\]
	Since each $\P(x, \dif y)$ is absolutely continuous with respect to $\Q(x, \dif y)$, this gives
	\[
		(\S \circ \P)(x, A) = \sum_{i=1}^\infty a_i \, \P(x, A_i) = 0.
	\]
	for $\R(w, \dif x)$-almost all $x \in \tarS$.
	We therefore obtain
	\[
		(\S \circ \P \circ \R)(w, A) = \int \R(w, \dif x) \, (\S \circ \P)(x, A) = 0,
	\]
	and the result now follows since $w$ and $A$ were arbitrary.
\end{proof}

\subsubsection{Proof of Proposition \ref{prop:ac-comparison}} \label{sec:appendix-ac-comparison}

\begin{proof}
	First note that if $\ux : \U \to \ind$ and $\vx : \ind \to \V$ are morphisms in any $\CMon$-enriched CD category, we have that $ \vx \circ \ux= 0$ if and only if $\ux \otimes \vx = 0$ (as can be seen by applying the left and right unitors on both sides).
	It follows that when there are no zero divisors, the condition $\vx \circ \ux = 0$ implies either $\vx = 0$ or $\ux = 0$.

	Now suppose we have $\mu$ and $\nu$ as in the statement, where $\nu$ is absolutely continuous with respect to $\mu$ in the sense of \cite[Definition 2.2.1]{Fritz2023}. 
	Let $\R : \auxS \to \ind$ and $\S : \tarS \to \sndS$ be morphisms such that $\S \circ \mu \circ \R = 0$.
	By the previous paragraph, this means either $\S \circ \mu = 0$ or $\R = 0$ holds.
	If it is the latter, then clearly $\S \circ \mu \circ \R = 0$ also.
	So suppose $\S \circ \mu = 0$ holds instead.
	This gives 
	
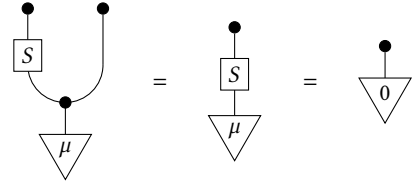
\begin{cdiag*}
		\begin{tikzpicture}
	\begin{pgfonlayer}{nodelayer}
		\node [style=none] (0) at (0, -0.25) {$=$};
		\node [style=state] (1) at (-2.5, -2) {$\mu$};
		\node [style=none] (5) at (-2.5, -1.75) {};
		\node [style=none] (6) at (-2.5, -0.75) {};
		\node [style=none] (7) at (-3.5, 0.25) {};
		\node [style=none] (8) at (-1.5, 0.25) {};
		\node [style=bn] (9) at (-2.5, -0.75) {};
		\node [style=bn] (10) at (-3.5, 1.75) {};
		\node [style=bn] (11) at (-1.5, 1.75) {};
		\node [style=morphism] (12) at (-3.5, 0.5) {$S$};
		\node [style=state] (13) at (2, -1.5) {$\mu$};
		\node [style=none] (17) at (2, -1.25) {};
		\node [style=none] (18) at (2, -0.25) {};
		\node [style=none] (19) at (2, 0.25) {};
		\node [style=bn] (22) at (2, 1.25) {};
		\node [style=morphism] (24) at (2, 0) {$S$};
		\node [style=none] (25) at (4, -0.25) {$=$};
		\node [style=state] (26) at (6, -0.5) {$0$};
		\node [style=none] (27) at (6, -0.25) {};
		\node [style=bn] (28) at (6, 0.75) {};
	\end{pgfonlayer}
	\begin{pgfonlayer}{edgelayer}
		\draw (6.center) to (5.center);
		\draw [bend left=45] (6.center) to (7.center);
		\draw [bend right=45] (6.center) to (8.center);
		\draw (10) to (7.center);
		\draw (11) to (8.center);
		\draw (18.center) to (17.center);
		\draw (22) to (19.center);
		\draw (28) to (27.center);
	\end{pgfonlayer}
\end{tikzpicture}
	\end{cdiag*}
	Here the right-hand side is the zero morphism $0 : \ind \to \ind$ (see Remark \ref{rem:diagrammatic-reasoning-semiadditive}).
	By the zero-monic condition, we therefore have
	
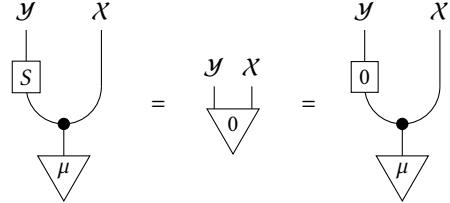
\begin{cdiag*}
		\begin{tikzpicture}
	\begin{pgfonlayer}{nodelayer}
		\node [style=none] (0) at (0, 0) {$=$};
		\node [style=state] (1) at (-2.5, -1.75) {$\mu$};
		\node [style=none] (2) at (-2.5, -1.5) {};
		\node [style=none] (3) at (-2.5, -0.5) {};
		\node [style=none] (4) at (-3.5, 0.5) {};
		\node [style=none] (5) at (-1.5, 0.5) {};
		\node [style=bn] (6) at (-2.5, -0.5) {};
		\node [style=none] (7) at (-3.5, 2) {};
		\node [style=none] (8) at (-1.5, 2) {};
		\node [style=morphism] (9) at (-3.5, 0.75) {$S$};
		\node [style=none] (13) at (1.5, -0.25) {};
		\node [style=none] (14) at (1.5, 0.5) {};
		\node [style=state] (15) at (2, -0.5) {$0$};
		\node [style=none] (16) at (2.5, -0.25) {};
		\node [style=none] (17) at (2.5, 0.5) {};
		\node [style=none] (18) at (-3.5, 2.5) {$\sndS$};
		\node [style=none] (19) at (-1.5, 2.5) {$\tarS$};
		\node [style=none] (22) at (1.5, 1) {$\sndS$};
		\node [style=none] (23) at (2.5, 1) {$\tarS$};
		\node [style=state] (24) at (6.5, -1.75) {$\mu$};
		\node [style=none] (25) at (6.5, -1.5) {};
		\node [style=none] (26) at (6.5, -0.5) {};
		\node [style=none] (27) at (5.5, 0.5) {};
		\node [style=none] (28) at (7.5, 0.5) {};
		\node [style=bn] (29) at (6.5, -0.5) {};
		\node [style=none] (30) at (5.5, 2) {};
		\node [style=none] (31) at (7.5, 2) {};
		\node [style=morphism] (32) at (5.5, 0.75) {$0$};
		\node [style=none] (33) at (5.5, 2.5) {$\sndS$};
		\node [style=none] (34) at (7.5, 2.5) {$\tarS$};
		\node [style=none] (35) at (4, 0) {$=$};
	\end{pgfonlayer}
	\begin{pgfonlayer}{edgelayer}
		\draw (3.center) to (2.center);
		\draw [bend left=45] (3.center) to (4.center);
		\draw [bend right=45] (3.center) to (5.center);
		\draw (7.center) to (4.center);
		\draw (8.center) to (5.center);
		\draw (14.center) to (13.center);
		\draw (17.center) to (16.center);
		\draw (26.center) to (25.center);
		\draw [bend left=45] (26.center) to (27.center);
		\draw [bend right=45] (26.center) to (28.center);
		\draw (30.center) to (27.center);
		\draw (31.center) to (28.center);
	\end{pgfonlayer}
\end{tikzpicture}
	\end{cdiag*}
	where the second step also uses Remark \ref{rem:diagrammatic-reasoning-semiadditive}.
	Since $\nu$ is absolutely continuous with respect to $\mu$ \cite[Definition 2.2.1]{Fritz2023}, the same equation holds with $\nu$ in place of $\mu$.
	By deleting the second output on both sides, we therefore obtain $\S \circ \nu = 0$.
	In turn, this gives $\S \circ \nu \circ \R = 0$, which shows $\nu \ll \mu$ as required.
\end{proof}

\subsubsection{Proof of Proposition \ref{prop:singular-measures-meet}} \label{sec:appendix-singular-measures-meet}

\begin{proof}
	Suppose each pair of $\P(x, \dif y)$ and $\Q(x, \dif y)$ is singular.
	Then given $x \in \tarS$, there exists a measurable set $S_x \subseteq \sndS$ such that
	\[
		\P(x, S_x) = 0 \qquad\qquad \Q(x, S_x^c) = 0.
	\]
	Now suppose $\R \ll \P$ and $\R \ll \Q$.
	Choose arbitrary $x \in \tarS$ and measurable $A \subseteq \sndS$.
	Then we have
	\[
		\P(x, A \cap S_x) \leq \P(x, S_x) = 0,
	\]
	and so $\R(x, A \cap S_x) = 0$ since $\R \ll \P$.
	A similar argument gives $\Q(x, A \cap S_x^c) = 0$ and hence $\R(x, A \cap S_x^c) = 0$.
	We now have
    \begin{equation} \label{eq:lebesgue-decomposition-ac-meet}
        \R(x, A) = \R(x, A \cap S_x) + \R(x, A \cap S_x^c) = 0 + 0 = 0.
	\end{equation}
	It follows that $\R = 0$, which shows $0$ is a meet of $\P$ and $\Q$, or equivalently $\P \perp \Q$.

	Now suppose that $\P \perp \Q$, and let $x \in \tarS$ be arbitrary.
	Since $\P(x, \dif y)$ and $\Q(x, \dif y)$ are $\sigma$-finite, we may apply the classical Lebesgue decomposition theorem \cite[Section~32]{halmos2013measure} to write
	\[
		\P(x, \dif y) = \nu_\ac + \nu_\si
	\]
	where $\nu_\ac$ and $\nu_\si$ are respectively absolutely continuous and singular with respect to $\Q(x, \dif y)$.
	Since $\tarS$ has measurable singletons, we may define the following kernel:
	\[
		\R(x', \dif y) = \begin{cases}
			\nu_\ac & \text{if $x' = x$} \\
			0 & \text{otherwise.}
		\end{cases}
	\]
	It is clear that $\R \leq \P$ and hence $\R \ll \P$.
	Likewise, since $\nu_\ac$ is absolutely continuous with respect to $\Q(x, \dif y)$, Proposition \ref{prop:ac-sfkern} gives $\R \ll \Q$.
	Since $0$ is a meet of $\P$ and $\Q$, it follows that $\R \ll 0$, and hence $\R = 0$ by Proposition \ref{prop:ac-zero-monic}.
	But this means $\nu_\ac = 0$, and hence $\P(x, \dif y)$ and $\Q(x, \dif y)$ are singular measures as required.
\end{proof}

\subsubsection{Proof of Proposition \ref{prop:lebesgue-decomposition-sfkern}} \label{sec:appendix-lebesgue-decomposition-meet}

\begin{proof}
	Suppose that for each $x \in \tarS$, it holds that $\P(x, \dif y) = \P_\ac(x, \dif y) + \P_\si(x, \dif y)$ is a Lebesgue decomposition of $\Q(x, \dif y)$.
	Proposition \ref{prop:singular-measures-meet} immediately gives $\P_\si \perp \Q$. %
	It remains to show that $\P_\ac$ is a meet of $\P$ and $\Q$.
	Now, by definition of a Lebesgue decomposition and Proposition \ref{prop:ac-sfkern}, we have immediately that $\P_\ac \ll \P$ and $\P_\ac \ll \Q$.
	So given $\R$ satisfying $\R \ll \P$ and $\R \ll \Q$, we must show that $\R \ll \P_\ac$.
	
	For this, let $x \in \tarS$ be arbitrary, and suppose we have measurable $A \subseteq \sndS$ with $\P_\ac(x, A) = 0$.
	As in the proof of Proposition \ref{prop:singular-measures-meet}, there exists some measurable set $S_x \subseteq \sndS$ such that
	\[
		\P_\si(x, S_x) = 0 \qquad\qquad \Q(x, S_x^c) = 0.
	\]
	This gives
	\[
		\P(x, A \cap S_x) = \P_\ac(x, A \cap S_x) \leq \P_\ac(x, A) = 0,
	\]
	and hence $\R(x, A \cap S_x) = 0$ since $\R \ll \P$.
	Similarly, we have
	\[
		\Q(x, A \cap S_x^c) \leq \Q(x, S_x^c) = 0,
	\]
	and hence $\R(x, A \cap S_x^c) = 0$ since $\R \ll \Q$.
	As in \eqref{eq:lebesgue-decomposition-ac-meet}, this gives $\R(x, A) = 0$, which shows that $\R \ll \P_\ac$ as required for the first claim.

	For the second claim, suppose that $\tarS$ has measurable singletons, $\P$ and $\Q$ are pointwise $\sigma$-finite, and $(\P_\ac, \P_\si)$ is a $\Q$-decomposition.
	By Proposition \ref{prop:ac-sfkern}, each $\P_\ac(x, \dif y)$ is absolutely continuous with respect to $\Q(x, \dif y)$.
	Likewise, since $\P_\si \leq \P$ is pointwise $\sigma$-finite, Proposition \ref{prop:singular-measures-meet} gives that each $\P_\si(x, \dif y)$ is singular with respect to $\Q(x, \dif y)$.
	This shows that $\P_\ac(x, \dif y) + \P_\si(x, \dif y)$ is a Lebesgue decomposition of $\P(x, \dif y)$ with respect to $\Q(x, \dif y)$ as required.
\end{proof}

\section{Additional examples from computational statistics}

We include here two examples of other techniques from computational statistics that are easily formulated in terms of categorical probability.
We hope these examples will spark further interest in this topic as a source of applications for categorical probability.

\subsection{Gibbs sampling}\label{subsec:Gibbs}

We formulate here the \emph{Gibbs sampler} \cite{Geman1984} in terms of Markov categories.
This requires the notion of \textit{conditionals}: recall that a morphism $\P : \tarS \to \sndS$ is a \emph{conditional} of $\mu : \ind \to \tarS \otimes \sndS$ if
\begin{cdiag} \label{eq:conditional_def}
\begin{tikzpicture}
	\begin{pgfonlayer}{nodelayer}
		\node [style=state] (0) at (-3, -0.5) {$\mu$};
		\node [style=none] (1) at (-3.5, -0.25) {};
		\node [style=none] (2) at (-3.5, 0.75) {};
		\node [style=none] (3) at (-2.5, 0.75) {};
		\node [style=none] (4) at (-2.5, -0.25) {};
		\node [style=none] (5) at (-3.5, 1.25) {$\mathcal X$};
		\node [style=none] (6) at (-2.5, 1.25) {$\mathcal Y$};
		\node [style=none] (7) at (0, 0) {$=$};
		\node [style=state] (8) at (4, -1.75) {$\mu$};
		\node [style=none] (9) at (3.5, -1.5) {};
		\node [style=none] (10) at (3.5, -0.25) {};
		\node [style=bn] (11) at (4.5, -0.75) {};
		\node [style=none] (12) at (4.5, -1.5) {};
		\node [style=none] (13) at (2.5, 2.75) {$\mathcal X$};
		\node [style=none] (15) at (2.5, 0.75) {};
		\node [style=none] (16) at (4.5, 2.25) {};
		\node [style=none] (17) at (4.5, 2.75) {$\mathcal Y$};
		\node [style=none] (18) at (4.5, 0.75) {};
		\node [style=none] (19) at (2.5, 2.25) {};
		\node [style=morphism] (20) at (4.5, 1.25) {$\P$};
		\node [style=bn] (21) at (3.5, -0.25) {};
	\end{pgfonlayer}
	\begin{pgfonlayer}{edgelayer}
		\draw (2.center) to (1.center);
		\draw (4.center) to (3.center);
		\draw (10.center) to (9.center);
		\draw (12.center) to (11);
		\draw [bend left=45] (10.center) to (15.center);
		\draw [bend right=45] (10.center) to (18.center);
		\draw (18.center) to (16.center);
		\draw (19.center) to (15.center);
	\end{pgfonlayer}
\end{tikzpicture}
\end{cdiag}
Conditionals for morphisms with more outputs are defined analogously: for example, $\Q : \tarS \otimes \auxS \to \sndS$ is a conditional of $\nu : \ind \to \tarS \otimes \sndS \otimes \auxS$ if we have
\begin{cdiag} \label{eq:conditional-example}
\begin{tikzpicture}
	\begin{pgfonlayer}{nodelayer}
		\node [style=none] (0) at (0, 0) {$=$};
		\node [style=long state] (1) at (-3.5, -0.5) {$\nu$};
		\node [style=none] (2) at (-4.5, -0.25) {};
		\node [style=none] (3) at (-3.5, -0.25) {};
		\node [style=none] (4) at (-2.5, -0.25) {};
		\node [style=none] (5) at (-4.5, 0.5) {};
		\node [style=none] (6) at (-3.5, 0.5) {};
		\node [style=none] (7) at (-2.5, 0.5) {};
		\node [style=none] (8) at (-4.5, 1) {$\mathcal{X}$};
		\node [style=none] (9) at (-3.5, 1) {$\mathcal{Y}$};
		\node [style=none] (10) at (-2.5, 1) {$\mathcal{Z}$};
		\node [style=long state] (11) at (3.75, -1.5) {$\nu$};
		\node [style=none] (12) at (2.75, -1.25) {};
		\node [style=none] (13) at (3.75, -1.25) {};
		\node [style=none] (14) at (4.75, -1.25) {};
		\node [style=bn] (15) at (3.75, -0.5) {};
		\node [style=none] (16) at (2.75, 0) {};
		\node [style=none] (17) at (4.75, 0) {};
		\node [style=none] (18) at (3.25, 0.5) {};
		\node [style=none] (19) at (2.25, 0.5) {};
		\node [style=none] (20) at (4.25, 0.5) {};
		\node [style=none] (21) at (5.25, 0.5) {};
		\node [style=medium box] (22) at (3.75, 1) {$\Q$};
		\node [style=none] (23) at (2.25, 2) {};
		\node [style=none] (24) at (5.25, 2) {};
		\node [style=none] (25) at (3.25, 0.75) {};
		\node [style=none] (26) at (4.25, 0.75) {};
		\node [style=bn] (27) at (2.75, 0) {};
		\node [style=bn] (28) at (4.75, 0) {};
		\node [style=none] (29) at (3.75, 1.25) {};
		\node [style=none] (30) at (3.75, 2) {};
		\node [style=none] (31) at (2.25, 2.5) {$\mathcal{X}$};
		\node [style=none] (32) at (5.25, 2.5) {$\mathcal{Z}$};
		\node [style=none] (33) at (3.75, 2.5) {$\mathcal{Y}$};
	\end{pgfonlayer}
	\begin{pgfonlayer}{edgelayer}
		\draw (7.center) to (4.center);
		\draw (3.center) to (6.center);
		\draw (2.center) to (5.center);
		\draw (13.center) to (15);
		\draw (16.center) to (12.center);
		\draw (14.center) to (17.center);
		\draw [bend right=45] (16.center) to (18.center);
		\draw [bend right=45] (20.center) to (17.center);
		\draw [bend left=45] (16.center) to (19.center);
		\draw [bend right=45] (17.center) to (21.center);
		\draw (19.center) to (23.center);
		\draw (24.center) to (21.center);
		\draw (26.center) to (20.center);
		\draw (18.center) to (25.center);
		\draw (29.center) to (30.center);
	\end{pgfonlayer}
\end{tikzpicture}
\end{cdiag}
Note that there is some potential for ambiguity when output wires are repeated, e.g.\ when $\tarS = \sndS = \auxS$, but this will not be an issue in what follows.

\begin{example}\label{exa:conditionals}
		In $\stoch$, the defining property of conditionals \eqref{eq:conditional_def} says that the following holds:
	\[
	\mu(A \times B) = \int_A \mu(\dif x, \sndS) \, \P(x, B).
	\]
	In other words, $\P$ is a \emph{disintegration} of $\mu$ given $\tarS$ in the classical sense (e.g. \cite[Theorem~3.4]{Kallenberg2021}).
	In $\stoch$, conditionals do not always exist.
	However, they do exist in the Markov category $\borelstoch$, which is the subcategory of $\stoch$ consisting of \emph{standard Borel spaces} and all Markov kernels between these. The Markov category structure of $\borelstoch$ is inherited directly from $\stoch$ (see \cite[Section~3]{Cho2019} for more details).
\end{example}

The \emph{systematic-scan Gibbs sampler} is a classical MCMC procedure for sampling from targets defined on a product space \cite{Geman1984}.
In this approach, each coordinate is updated in turn by sampling from its conditional distribution given the other coordinates.
It is a classical result that this leaves the target invariant.
While this is straightforward to show in simple settings (e.g.\ finite state spaces or densities), the general measure-theoretic result is much more technical (see e.g.\ \cite[Theorem~10.6]{Robert2004}). 
In this section, we give a very simple but fully general proof in the abstract setting of Markov categories instead.

Given objects $\tarS_1, \ldots, \tarS_n$ in a Markov category $\C$, we will use ``slice'' notation to denote their finite monoidal products:
\[
\tarS_{i:j} := \tarS_i \otimes \cdots \otimes \tarS_j,
\]
with the convention that $\tarS_{i:j} \coloneqq \ind$ when the slice is empty (i.e.\ when $i > j$).
Now suppose that all conditionals exist in $\C$, and we are given a distribution $\mu : \ind \to \tarS_{1:n}$ on the full joint space.
For each $i=1,\dots, n$, define the morphism $\P_i$ as follows:
\begin{cdiag*}
\begin{tikzpicture}
	\begin{pgfonlayer}{nodelayer}
		\node [style=morphism] (0) at (-3, 0) {$\P_i$};
		\node [style=none] (1) at (-3, -0.25) {};
		\node [style=none] (2) at (-3, 0.25) {};
		\node [style=none] (3) at (-3, 1.25) {};
		\node [style=none] (4) at (-3, -1.25) {};
		\node [style=none] (5) at (-3, -2) {$\mathcal{X}_{1:n}$};
		\node [style=none] (6) at (-3, 2) {$\mathcal{X}_{1:n}$};
		\node [style=none] (7) at (0, 0) {$\coloneqq$};
		\node [style=none] (9) at (3.25, -1.5) {};
		\node [style=none] (11) at (5.25, -1.5) {};
		\node [style=none] (13) at (3.25, -0.5) {};
		\node [style=none] (14) at (5.25, -0.5) {};
		\node [style=none] (15) at (3.75, 0) {};
		\node [style=none] (16) at (2.75, 0) {};
		\node [style=none] (17) at (4.75, 0) {};
		\node [style=none] (18) at (5.75, 0) {};
		\node [style=medium box] (19) at (4.25, 0.5) {$\mu_{i|-i}$};
		\node [style=none] (20) at (2.75, 1.5) {};
		\node [style=none] (21) at (5.75, 1.5) {};
		\node [style=none] (22) at (3.75, 0.25) {};
		\node [style=none] (23) at (4.75, 0.25) {};
		\node [style=bn] (24) at (3.25, -0.5) {};
		\node [style=bn] (25) at (5.25, -0.5) {};
		\node [style=none] (26) at (4.25, 0.75) {};
		\node [style=none] (27) at (4.25, 1.5) {};
		\node [style=none] (28) at (2.5, 2) {$\mathcal{X}_{1:i-1}$};
		\node [style=none] (29) at (6.25, 2) {$\mathcal{X}_{i+1:n}$};
		\node [style=none] (30) at (4.25, 2) {$\mathcal{X}_i$};
		\node [style=none] (31) at (2.5, -2) {$\mathcal{X}_{1:i-1}$};
		\node [style=none] (32) at (4.25, -1.5) {};
		\node [style=none] (33) at (4.25, -1) {};
		\node [style=bn] (34) at (4.25, -1) {};
		\node [style=none] (35) at (4.25, -2) {$\mathcal{X}_i$};
		\node [style=none] (36) at (6.25, -2) {$\mathcal{X}_{i+1:n}$};
	\end{pgfonlayer}
	\begin{pgfonlayer}{edgelayer}
		\draw (3.center) to (2.center);
		\draw (4.center) to (1.center);
		\draw (13.center) to (9.center);
		\draw (11.center) to (14.center);
		\draw [bend right=45] (13.center) to (15.center);
		\draw [bend right=45] (17.center) to (14.center);
		\draw [bend left=45] (13.center) to (16.center);
		\draw [bend right=45] (14.center) to (18.center);
		\draw (16.center) to (20.center);
		\draw (21.center) to (18.center);
		\draw (23.center) to (17.center);
		\draw (15.center) to (22.center);
		\draw (26.center) to (27.center);
		\draw (34) to (32.center);
	\end{pgfonlayer}
\end{tikzpicture}
\end{cdiag*}
where here $\mu_{i|-i}$ denotes a conditional of $\mu$.
Notice that $P_i$ updates the $i$-th coordinate using this conditional while leaving the other coordinates unchanged.
We then define $\P_{\gibbs}: \tarS_{1:n} \to \tarS_{1:n}$ as the $n$-fold composition that simply applies each $\P_i$ in turn:
$$
\P_{\gibbs} := \P_n \circ \cdots \circ \P_1.
$$
If we take our ambient Markov category to be $\borelstoch$ (which admits all conditionals, Example~\ref{exa:conditionals}), this recovers the usual systematic-scan Gibbs sampler exactly.
The following straightforward result then shows it is invariant.

\begin{proposition}
	In any Markov category $\C$ with conditionals, $\P_{\gibbs}$ as defined above is $\mu$-invariant.
\end{proposition}

\begin{proof}
	Since each $\mu_{i|-i}$ is a conditional of $\mu$, we have by definition that
	\begin{cdiag*}
	\begin{tikzpicture}
	\begin{pgfonlayer}{nodelayer}
		\node [style=state] (10) at (9, -0.5) {$\mu$};
		\node [style=none] (11) at (9, 0.5) {};
		\node [style=none] (12) at (0, 0) {$=$};
		\node [style=none] (13) at (-2, 1.5) {$\mathcal X_{1:n}$};
		\node [style=state] (17) at (-2, -1.5) {$\mu$};
		\node [style=morphism] (18) at (-2, 0) {$P_i$};
		\node [style=none] (19) at (9, 1) {$\mathcal X_{1:n}$};
		\node [style=none] (20) at (7, 0) {$=$};
		\node [style=none] (25) at (-2, 1) {};
		\node [style=none] (26) at (2.5, -1.75) {};
		\node [style=none] (27) at (4.5, -1.75) {};
		\node [style=none] (28) at (2.5, -0.75) {};
		\node [style=none] (29) at (4.5, -0.75) {};
		\node [style=none] (30) at (3, -0.25) {};
		\node [style=none] (31) at (2, -0.25) {};
		\node [style=none] (32) at (4, -0.25) {};
		\node [style=none] (33) at (5, -0.25) {};
		\node [style=medium box] (34) at (3.5, 0.25) {$\mu_{i|-i}$};
		\node [style=none] (35) at (2, 1.25) {};
		\node [style=none] (36) at (5, 1.25) {};
		\node [style=none] (37) at (3, 0) {};
		\node [style=none] (38) at (4, 0) {};
		\node [style=bn] (39) at (2.5, -0.75) {};
		\node [style=bn] (40) at (4.5, -0.75) {};
		\node [style=none] (41) at (3.5, 0.5) {};
		\node [style=none] (42) at (3.5, 1.25) {};
		\node [style=none] (43) at (1.75, 1.75) {$\mathcal{X}_{1:i-1}$};
		\node [style=none] (44) at (5.5, 1.75) {$\mathcal{X}_{i+1:n}$};
		\node [style=none] (45) at (3.5, 1.75) {$\mathcal{X}_i$};
		\node [style=none] (47) at (3.5, -1.75) {};
		\node [style=none] (48) at (3.5, -1.25) {};
		\node [style=bn] (49) at (3.5, -1.25) {};
		\node [style=long state] (50) at (3.5, -2) {$\mu$};
	\end{pgfonlayer}
	\begin{pgfonlayer}{edgelayer}
		\draw (10) to (11.center);
		\draw (17) to (18);
		\draw (25.center) to (18);
		\draw (28.center) to (26.center);
		\draw (27.center) to (29.center);
		\draw [bend right=45] (28.center) to (30.center);
		\draw [bend right=45] (32.center) to (29.center);
		\draw [bend left=45] (28.center) to (31.center);
		\draw [bend right=45] (29.center) to (33.center);
		\draw (31.center) to (35.center);
		\draw (36.center) to (33.center);
		\draw (38.center) to (32.center);
		\draw (30.center) to (37.center);
		\draw (41.center) to (42.center);
		\draw (49) to (47.center);
	\end{pgfonlayer}
\end{tikzpicture}
	\end{cdiag*}
	(see e.g.\ the example in \eqref{eq:conditional-example}).
	This shows each $\P_i$ is $\mu$-invariant, and so their composition $\P_{\gibbs}$ is also $\mu$-invariant by Proposition~\ref{prop:reversibility-implies-invariance}.
\end{proof}

\subsection{Importance sampling}

Radon--Nikodym derivatives give a way to talk about \emph{importance sampling} in CD categories.
Let $\f : \tarS \to [0, \infty]$ be a measurable function, and hence an effect $\tarS \to \ind$ in $\sfkern$ (Example~\ref{exa:effects-sfkern}), and suppose $\rx = \frac{\dif \pi}{\dif \mu}$ is a Radon--Nikodym derivative (Definition \ref{def:radon-nikodym}).
Postcomposing both sides of \eqref{diag:IS_basic} by $f$ yields
\begin{cdiag} \label{diag:importance-sampling}
	\begin{tikzpicture}
	\begin{pgfonlayer}{nodelayer}
		\node [style=none] (0) at (0, 0) {$=$};
		\node [style=none] (1) at (3, 0) {};
		\node [style=none] (2) at (2, 1) {};
		\node [style=none] (3) at (4, 1) {};
		\node [style=none] (4) at (3, -1) {};
		\node [style=bn] (5) at (3, 0) {};
		\node [style=state] (6) at (-2.5, -0.75) {$\pi$};
		\node [style=none] (7) at (-2.5, -0.5) {};
		\node [style=none] (8) at (-2.5, 0.5) {};
		\node [style=likelihood] (9) at (-2.5, 0.75) {$f$};
		\node [style=likelihood] (10) at (2, 1.25) {$f$};
		\node [style=likelihood] (11) at (4, 1.25) {$\rx$};
		\node [style=state] (12) at (3, -1.25) {$\mu$};
	\end{pgfonlayer}
	\begin{pgfonlayer}{edgelayer}
		\draw [bend left=45] (3.center) to (1.center);
		\draw [bend left=45] (1.center) to (2.center);
		\draw (1.center) to (4.center);
		\draw (8.center) to (7.center);
	\end{pgfonlayer}
\end{tikzpicture}
\end{cdiag}
Notice that the effect $f \ast \rx$ appears on the right-hand side (see Remark \ref{rem:effect-multiplication}), so this may also be written suggestively as
\[
	f \circ \pi = \left(f \ast \frac{\dif \pi}{\dif \mu}\right) \circ \mu.
\]
As explained in Example \ref{exa:effects-sfkern}, in the context of $\sfkern$ this translates to
\[
	\int \f(x) \, \pi(\dif x) = \int \f(x) \, \frac{\dif \pi}{\dif \mu}(x) \, \mu(\dif x),
\]
which recovers the standard formula for \emph{importance sampling} \cite[Section~3.3]{Robert2004}.
More generally, \eqref{diag:importance-sampling} may be regarded as an abstract version of importance sampling that makes sense in any CD category.

\end{document}